\documentclass[11pt]{article}
\usepackage[colorlinks=true,linkcolor=blue,citecolor=blue]{hyperref}
\usepackage{lmodern} 
\usepackage[T1]{fontenc}

\usepackage{float}
\usepackage{color}
\usepackage{amsmath, amssymb, amsthm}
\usepackage{mathtools}
\usepackage[margin=1in]{geometry}
\usepackage{graphics}
\usepackage{pifont}
\usepackage{tikz}
\usepackage{bbm, bm}
\DeclareMathOperator*{\argmax}{arg\,max}

\usetikzlibrary{arrows.meta}
\usepackage{environ}
\usepackage{framed}
\usepackage{url}
\usepackage[linesnumbered,ruled,vlined]{algorithm2e}
\usepackage[noend]{algpseudocode}
\usepackage[labelfont=bf]{caption}
\usepackage{framed}
 \usepackage[framemethod=tikz]{mdframed}
\usepackage{appendix}
\usepackage{graphicx}
\usepackage[textsize=tiny]{todonotes}
\usepackage{tcolorbox}
\allowdisplaybreaks[1]
\usepackage{nicefrac}
\usepackage{thm-restate}
\usepackage[noabbrev,capitalize,nameinlink]{cleveref}
\crefname{equation}{}{}
\usepackage{makecell}

\usepackage{subcaption}

\DontPrintSemicolon
\SetKw{KwAnd}{and}
\SetProcNameSty{textsc}
\SetFuncSty{textsc}

\newcommand\remove[1]{}

\newtheorem{lemma}{Lemma}[section]
\newtheorem*{lemma*}{Lemma}
\newtheorem{theorem}[lemma]{Theorem}
\newtheorem{corollary}[lemma]{Corollary}

\newtheorem{remark}[lemma]{Remark}
\newtheorem*{corollary*}{Corollary}

\newtheorem*{theorem*}{Theorem}
\newtheorem*{inducthyp*}{Inductive Hypothesis}
\newtheorem*{definition*}{Definition}
\newtheorem{definition}[lemma]{Definition}

\newtheorem*{rem*}{Remark}

\newtheorem{problem}{Problem}

\newcommand\N{\mathbb{N}}

\newcommand\Z{\mathbb{Z}}

\renewcommand{\forall}{\mathrm{\text{ for all }}}

\newcommand{\degree}{\mathtt{degree}}
\newcommand{\splay}{\textsc{Splay}}
\newcommand{\roott}{\mathtt{root}}
\newcommand{\subtree}{\mathtt{subtree}}
\newcommand{\subtreeT}{\mathtt{subtreeT}}
\newcommand{\parent}{\mathtt{parent}}
\newcommand{\son}{\mathtt{children}}
\newcommand{\sonT}{\mathtt{childrenT}}
\newcommand{\neighbor}{\mathtt{neighbor}}
\newcommand{\size}{\mathtt{size}}
\newcommand{\pre}{\mathtt{pred}}
\newcommand{\suc}{\mathtt{succ}}
\newcommand{\leftson}{\mathtt{left}}
\newcommand{\rightson}{\mathtt{right}}
\newcommand{\father}{\mathtt{parentT}}
\newcommand{\final}{\mathtt{final}}
\newcommand{\fire}{\mathtt{fire}}
\newcommand{\nil}{\mathtt{nil}}
\newcommand{\ovf}{\mathtt{overflow}}
\newcommand{\bucket}{\mathtt{Bucket}}
\newcommand{\res}{\mathtt{res}}
\newcommand{\ta}{\mathtt{a}}
\newcommand{\tb}{\mathtt{b}}
\newcommand{\num}{\mathtt{num}}

\newcommand{\bc}{\mathbf{c}}
\newcommand{\bd}{\boldsymbol{d}}
\renewcommand{\bf}{\bm{f}}

\newcommand{\bp}{\boldsymbol{p}}
\newcommand{\bq}{\boldsymbol{Q}}
\newcommand{\vq}{\boldsymbol{q}}

\newcommand{\bx}{\boldsymbol{x}}
\newcommand{\by}{\boldsymbol{y}}
\newcommand{\ba}{\boldsymbol{a}}
\newcommand{\bb}{\boldsymbol{b}}

\newcommand{\findmin}{\textsc{FindMin}}
\newcommand{\Delete}{\textsc{Delete}}
\newcommand{\Insert}{\textsc{Insert}}
\newcommand{\rotate}{\textsc{Rotate}}
\newcommand{\zig}{\textsc{Zig}}
\newcommand{\zigzag}{\textsc{Zig-Zag}}
\newcommand{\zigzig}{\textsc{Zig-Zig}}

\DeclareFontFamily{U}{mathb}{\hyphenchar\font45}
\DeclareFontShape{U}{mathb}{m}{n}{<5> <6> <7> <8> <9> <10> gen * mathb
<10.95> mathb10 <12> <14.4> <17.28> <20.74> <24.88> mathb12}{}
\DeclareSymbolFont{mathb}{U}{mathb}{m}{n}
\DeclareMathSymbol{\rcirclearrow}{\mathbin}{mathb}{'367}

\foreach \x in {A,...,Z}{%
	\expandafter\xdef\csname m\x\endcsname{\noexpand\mathbf{\x}}
}

\foreach \x in {A,...,Z}{%
	\expandafter\xdef\csname c\x\endcsname{\noexpand\mathcal{\x}}
}

\newif\ifrandom
\randomtrue

\newcommand{\polylog}{{\mathrm{polylog}}}

\newcommand{\rank}{\mathsf{rank}}

\newcommand{\todolater}[1]{}

\definecolor{pinky}{RGB}{255, 51, 204}

\newcommand{\down}{\downarrow}

\newcommand{\upward}{\textsc{SolvePartial}}
\newcommand{\downward}{\textsc{SolveComplete}}
\newcommand{\cnt}{\textsc{ComputeC}}
\newcommand{\updatedsupward}{\textsc{Update}}
\newcommand{\revertds}{\textsc{Revert}}
\newcommand{\mergeupward}{\textsc{Merge}}
\newcommand{\deltasum}{\textsc{DeltaSum}}
\newcommand{\pushup}{\textsc{PushUp}}
\newcommand{\pushdown}{\textsc{PushDown}}
\newcommand{\simulation}{\textsc{Simulation}}
\newcommand{\findoneintree}{\textsc{FindOneInTree}}
\newcommand{\findonebeforetree}{\textsc{FindOneBeforeTree}}
\newcommand{\splitds}{\textsc{Split}}
\newcommand{\deltaquery}{\textsc{DeltaQuery}}
\newcommand{\initds}{\textsc{Initialize}}
\newcommand{\inctime}{\textsc{IncTime}}
\newcommand{\newnode}{\textsc{NewNode}}
\newcommand{\moment}{\texttt{moment}}
\newcommand{\diff}{\texttt{diff}}
\newcommand{\dist}{\texttt{dist}}
\newcommand{\sumdiff}{\texttt{sumdiff}}
\newcommand{\timestamp}{\texttt{timestamp}}
\newcommand{\timemin}{\texttt{timemin}}
\newcommand{\timemax}{\texttt{timemax}}
\newcommand{\indexx}{\texttt{visit\_time}}
\newcommand{\dfn}{\texttt{dfs\_order}}
\newcommand{\presigma}{\Tilde{\sigma}}
\newcommand{\findoneA}{\textsc{FindOneInTree}}
\newcommand{\findoneB}{\textsc{FindOneBeforeTree}}
\newcommand{\True}{\texttt{true}}
\SetKw{Break}{break}
\usetikzlibrary{arrows}
\tikzset{edge/.style={->,> = latex'}}
\tikzset{special/.style={fill=red!50,circle,minimum size=0.6cm,inner sep=1pt}}

\newcommand{\downwardm}{\textsc{PathComplete}}
\newcommand{\revertm}{\textsc{PathRevert}}
\newcommand{\delquem}{\textsc{PathQuery}}

\DeclareUnicodeCharacter{2113}{$\ell$}

\usepackage[backend=biber, isbn=false, style=alphabetic, backref=true, doi=false, url=false, maxcitenames=10, mincitenames=5, maxalphanames=10, maxbibnames=10, minbibnames=5, minalphanames=5, defernumbers=true, sortlocale=en]{biblatex}
\begin{document}

\title{Sandpile Prediction on Undirected Graphs}
\author{
Ruinian Chang\\ Tsinghua University \\ Ruinian127@gmail.com
\and
Jingbang Chen\\ University of Waterloo\\ j293chen@uwaterloo.ca
\and 
Ian Munro \\ University of Waterloo \\ imunro@uwaterloo.ca
\and
Richard Peng \\ Carnegie Mellon University \\ yangp@cs.cmu.edu
\and
Qingyu Shi\\ Peking University\\ qingyuqwq@gmail.com
\and
Zeyu Zheng\\ Carnegie Mellon University\\
zeyuzhen@andrew.cmu.edu
}

\maketitle
\begin{abstract}
The \textit{Abelian Sandpile} model is a well-known model used in exploring \textit{self-organized criticality}. Despite a large amount of work on other aspects of sandpiles, there have been limited results in efficiently computing the terminal state, known as the \textit{sandpile prediction} problem. 

On graphs with special structures, we present algorithms that compute the terminal configurations for sandpile instances in $O(n \log n)$ time on trees and $O(n)$ time on paths, where $n$ is the number of vertices.  Our algorithms improve the previous best runtime of $O(n \log^5 n)$ on trees [Ramachandran-Schild SODA '17] and $O(n \log n)$ on paths [Moore-Nilsson '99]. To do so, we move beyond the simulation of individual events by directly computing the number of firings for each vertex. The computation is accelerated using splittable binary search trees. In addition, we give algorithms in $O(n)$ time on cliques and $O(n \log^2 n)$ time on pseudotrees.

On general graphs, we propose a fast algorithm under the setting where the number of chips $N$ could be arbitrarily large. 
We obtain a $\log N$ dependency, improving over the $\mathtt{poly}(N)$ dependency in purely simulation-based algorithms.
Our algorithm also achieves faster performance on various types of graphs, including regular graphs, expander graphs, and hypercubes. We also provide a reduction that enables us to decompose the input sandpile into several smaller instances and solve them separately. 

\end{abstract}

\tableofcontents
\newpage
\section{Introduction}

The concept of \textit{self-organized criticality} was first proposed by Bak, Tang, and Wiesenfeld in 1987 \cite{bak1987self}. It helps to understand how power-law distributions arise and how complex systems inherently exhibit critical behavior, encapsulating the interaction between local activities and global dynamics. 
It is often referenced when studying many natural phenomena, such as earthquakes, forest fires, and avalanches \cite{bak2013nature}. It has also been identified and scrutinized across a diverse range of disciplines such as sociology \cite{dmitriev2021identification, kron2009society}, geophysics \cite{smalley1985renormalization,smyth2019self}, and neuroscience \cite{linkenkaer2001long,beggs2003neuronal,chialvo2004critical}. Self-organized criticality has also played a significant role in the understanding of economic systems \cite{biondo2015modeling,scheinkman1994self}, evolutionary biology \cite{phillips2014fractals}, materials science \cite{ramos2009avalanche}, astrophysics \cite{aschwanden2011self}, statistical physics \cite{dhar2006theoretical}, and epidemiology \cite{saba2014self}.


The \textit{Abelian sandpile} model, which is the first discovered dynamical system exhibiting self-organized criticality, is frequently utilized as a comprehensible and intuitive model for the study of self-organized criticality. Dhar~\cite{dhar1990self} offers a generalized interpretation of the Abelian sandpile model on finite graphs, also known as the chip-firing game on graphs \cite{bjorner1991chip}. In this model, chips are added to the vertices of the graph in the beginning, referred to as the initial configuration. If any vertex $x$ has at least $\degree(x)$ chips, it may distribute a single chip to each neighboring vertex. This distributing process is called a ``firing''. The instance either terminates after all possible firings or loops infinitely. 

The sandpile model has attracted considerable attention \cite{klivans2018mathematics}. Contemporary research has delved into various aspects of the model, encompassing topics such as sandpile groups \cite{chen2019sandpile,meszaros2020distribution,zhou2021sandpile,alfaro2021structure}, predictability \cite{montoya2019abelian,montoya2022predictability}, special variants of the model \cite{dukes2019abelian,kim2020stochastic,dukes2021sandpile,eckmann2023abelian}, algebraic connections \cite{abrams2023connections}, and its impact on real-world scenarios \cite{martucci2021hospital}.

Bjorner et al. \cite{bjorner1991chip} showed any firing order leads to the same result.
This raises a natural algorithmic problem: \textit{Sandpile Prediction}.

\begin{problem}[Sandpile Prediction]
\label{problem:prediction}
Given a graph $G$ and an initial configuration $\sigma$, the sandpile prediction problem is to determine whether the sandpile instance $S(G,\sigma)$ terminates and to compute its corresponding terminal configuration if it exists.
\end{problem}

This prediction problem holds significant importance in the fields of physics \cite{garber2009predicting}, computer science \cite{montoya2011computational}, and mathematics \cite{biggs1997algebraic}. Moreover, the sandpile prediction problem has direct connections with practical applications such as load balancing \cite{rabani1998local} and the derandomization of models like internal diffusion-limited aggregation \cite{diaconis1991growth,lawler1992internal}. In general, sandpile prediction unfolds into two different lines of research, one focusing on mathematically bounding the number of firings and the other on algorithmically predicting the result faster than mere simulation (\textit{prediction algorithms}). 

Despite the abundant literature on other aspects of the sandpile model, there have been limited results in developing prediction algorithms. 
On structured graphs including trees \cite{goles1996sand} and high dimensional grids \cite{moore1999computational}, the prediction problem has been shown to be P-Complete, which means it is difficult to develop parallel algorithms. 
On general graphs, there is no algorithm that works faster than simulation.

\subsection{Results}
In this paper, our work is divided into two types: solving the sandpile prediction problem on structured graphs and general graphs. 
\subsubsection{Sandpile Prediction on Structured Graphs}
In this paper, we solve the sandpile prediction problem on various structured graphs. We believe studying solving sandpile prediction on structured graphs is a necessary step for developing efficient algorithms on arbitrary graphs. Starting from trees and paths, we propose new algorithms that outperform the previous best runtimes. Our algorithm for sandpile prediction on trees (\cref{sec:overview}) achieves a time complexity of $O(n \log n)$ and requires only $O(n)$ memory, where $n$ represents the number of vertices in the tree. 

\begin{restatable}[Sandpile Prediction on Trees]{theorem}{maintheorem}
\label{theorem:main}
Given a sandpile instance $S(G, \sigma)$ such that $G$ is a tree, there is an algorithm that determines whether $S$ terminates and computes the terminal configuration of $S$ in $O(n \log n)$ time, with $O(n)$ memory.
\end{restatable}

Compared to the previous fastest algorithm \cite{RS17} that runs in $O(n \log^5 n)$ time, our algorithm takes a distinct approach, not relying on the decomposition of trees into paths. Instead, we compute the number of firings that occur at any given vertex $u$. The terminal configuration can in turn be constructed.

When the input graph is a path, we can also slightly modify our algorithm to run in linear time (\cref{sec:path}). This improvement surpasses the previous result presented in \cite{moore1999computational}, which required $O(n \log n)$ time to compute the terminal configuration. We also provide an algorithm for cliques that runs in linear time as well (\cref{sec:clique}).

\begin{restatable}[Sandpile Prediction on Paths]{theorem}{paththeorem}
\label{theorem:main-path}
Given a sandpile instance $S(G, \sigma)$ such that $G$ is $\text{Path}_n$, there is an algorithm that determines whether $S$ terminates and computes the terminal configuration of $S$ in $O(n)$ time, with $O(n)$ memory.
\end{restatable}

We also believe our method of using splittable search trees accelerating dynamic programming (\cref{sec:ds}) is of independent interest for developing algorithms on structured graphs.

\subsubsection{Sandpile Prediction on General Graphs} 
\paragraph{Simulation-based Algorithms} We first study the performance of simulation-based approaches on general graphs. Such approaches can still enjoy speedups: vertices with a lot of chips can fire multiple times at once. Although there are results on bounding the number of firings or moving chips, there is no prior work on analyzing the performance of these simulation-based algorithms on either general graphs or special graphs.


\cite{tardos1988polynomial} shows that the number of firings on any sandpile that terminates is $\Theta(n^4)$, which can be regarded as a bound for simulation as well. However, it is much different if we consider the generalized sandpile model with sinks. Sinks are vertices that cannot fire and do not affect the uniqueness of the terminal configuration. When a sandpile contains sinks, it always terminates \cite{klivans2018mathematics}, and the number of firings can be $\mathtt{poly}(N)$ \cite{holroyd2008chip} where $N$ denotes the total number of chips. Therefore, the performance of simulation-based approaches could be significantly worsened.

We propose a new simulation-based algorithm that works reasonably even with sinks (\cref{sec:simu}), addressing the above two concerns. It follows a simple greedy strategy and is very easy to implement. To better capture the performance of simulation-based algorithms, we analyze in terms of the number of iterations. An iteration means we execute one or several firings simultaneously on a single vertex. Note that as all sinks can be merged into one without affecting the result, we assume there is one sink in the graph.

\begin{restatable}[Sandpile Prediction on General Graphs]{theorem}{simulation}
\label{theo:simu}
Given a sandpile instance $S(G,\sigma)$ that contains exactly one sink, there is a simulation-based algorithm that terminates in $O(Rm^2\log(nN))$ iterations, where $m$ denotes the number of edges, $N$ denotes the total number of chips, and $R$ denotes the maximum effective resistance between the sink and any other vertex.    
\end{restatable}
\cref{theo:simu} shows a logarithmic dependency on $N$, which is largely distinct from the $\mathtt{poly}(N)$ bound of chips moving \cite{holroyd2008chip}. Thus, our algorithm highly reduces the effect when the number of chips becomes extremely large. Moreover, our algorithm has a better performance guarantee if the input graph has special structures. As shown in \cref{table:simulation}, we provide several results for our algorithm on several special graph classes.

\renewcommand{\arraystretch}{2}
\begin{table}[H]
\centering
\begin{tabular}{|c|l|c|}\hline
\textbf{}                                  \textbf{Graph Classes} & \makecell[c]{\textbf{Number of Iterations}}                 & \textbf{Formal Statement}                      \\ \hline
general graphs                              & $O(Rm^2\log(nN))$                   &     \cref{theo:simu}          \\ \hline
$d$-regular graphs                          &   $O(n^3\log(nN))$                   &   \cref{theo:simuRegular}     \\ \hline
$\epsilon$-vertex expanders with minimum degree $\delta$                &     $O(m^2\log(nN)/(\delta +1))$     &     \cref{coro:simuExpander} \\ \hline
$d$-regular $\epsilon$-vertex expanders     &    $O(n^2\log n\log(nN))$            &      \cref{theo:simuRegularExpander}  \\ \hline
hypercubes                                   &       $O(n^2\log(nN)) $              &      \cref{theo:hyperCube}            \\ \hline
graphs with maximum degree at most $\Delta$ &     $O(\Delta^2n^3\log(nN))   $      &      \cref{coro:simuBoundedDegree}    \\ \hline
planar graphs                               &    $O(n^3\log(nN))$                  &       \cref{coro:simuPlanar}       \\ \hline
\end{tabular}
\caption{Algorithmic Result on Various Structured Graphs}
\label{table:simulation}
\end{table}

\paragraph{Interactions with Graph Decomposition}
When the input graph is large, a common idea is to decompose it into several subgraphs and solve the problem on them separately. Moreover, if subgraphs in the decomposition have special structures, we might be able to apply specific algorithms to them. Therefore, it is natural to consider if we can solve the sandpile prediction problem in such a way. We answer this question affirmatively in \cref{sec:reduction}. We develop a reduction scheme that works on general graphs (\cref{theorem:general}). 

For computing the terminal configuration on a general graph, we are able to reduce the problem into predicting multiple sandpile instances with sink vertices by removing some vertices. Specifically, we provide a reduction that transforms the prediction problem on an arbitrary graph into problems on its subgraphs separated by any vertex set $P$. The reduction gives a time complexity of $O(\log^{|P|} n \cdot T)$ where $T$ denotes the total time to solve the prediction on each subgraph.


\subsection{Related Work}
\paragraph{Bounding the Number of Firings} Numerous studies estimate the number of chip firings necessary to arrive at a terminal configuration. This aspect has been examined for various classes of directed graphs with sinks \cite{montoya2009complexity}. Eriksson et al. \cite{eriksson1991no} showed that no polynomial bound exists for general directed graphs without sinks. Considering undirected graphs without sinks, Tardos et al. \cite{tardos1988polynomial} proposed a bound of $\Theta(n^4)$ for the firing number in a graph with $n$ vertices and $m$ edges. An alternative bound was offered by Bjorner et al. \cite{bjorner1991chip}, suggesting that a maximum of $nk/\lambda_2$ firings can occur, where $k$ represents the total number of chips and $\lambda_2$ stands for the smallest non-trivial eigenvalue of the graph Laplacian. Holroyd et al. \cite{holroyd2008chip} presented an improved bound for sandpiles with sinks, stating that the number of chip movements can be at most $2NmR$, where $N$ is the number of chips, and $R$ is the maximum effective resistance between the sink and any vertex. For an $n \times n$ grid, Babai et al. \cite{babai2007sandpile} introduced the concept of the \textit{transience class} to explore the maximum number of chips to be added to a sandpile instance with sinks before entering a recurrent state, initially providing an $O(n^{30})$ polynomial bound. Choure et al. \cite{choure2012random} enhanced the upper bound to $O(n^7)$ and also proved a lower bound of $\Omega(n^3)$. Durfee et al. \cite{durfee2018nearly} used techniques from electrical networks to offer a nearly tight upper bound of $O(n^4 \log^4 n)$ and a lower bound of $\Omega(n^4)$. This work was also extended to $n^d$-sized $d$-dimensional grids, providing an upper bound of $O(n^{3d-2}\log^{d+2}n)$ and a lower bound of $\Omega(n^{3d-2})$.

\paragraph{Non-Simulation Approaches}
This line of work is dedicated to calculating the terminal configuration of sandpile instances without sinks faster than by simple simulation.
Moore and Nilsson \cite{moore1999computational} proposed an algorithm that solves the prediction on a path of length $n$ in $O(n \log n)$ time.
They also provided a parallel algorithm that runs in $O(\log^3 n)$, showing that the sandpile prediction on a path is in $\mathbf{NC}^3$.

In contrast, the sandpile prediction problem has been classified as P-Complete for various classes of graphs, including tree structures \cite{goles1996sand} and grids with a dimension exceeding three \cite{moore1999computational}. Thus, an $O(\polylog(n))$ depth parallel algorithm would imply $\mathbf{P}=\mathbf{NC}$.
Ramachandran and Schild proposed an algorithm that solves the sandpile prediction problem on trees in $O(n \log^5 n)$ time \cite{RS17}.


\section{Preliminaries}
\label{sec:prelim}

We assume all graphs are undirected, unweighted, and simple (no self-loop or duplicate edge). For a graph $G$, we use $V(G)$ and $E(G)$ to denote the vertex set and edge set, respectively. For any vertex $v \in V(G)$, we define $\degree(v)$ as the number of neighbors, and $\neighbor(v)$ as the set of the neighbor vertices of vertex $v$. $n=|V(G)|$ refers to the number of vertices. As is standard, we assume the word-RAM model with $\Theta(\lg n)$-size words. 



Given a graph $G$ and a configuration vector $\sigma\in \mathbb{N}^{n}$, we define the sandpile instance on them as $S(G,\sigma)$. Configurations represent the number of chips on each vertex. A vertex $v$ is said to be full if and only if $\sigma_v \geq \degree(v)$. A firing operation is defined on any full vertex $v$, which will change $\sigma$ in the following way:

\[\sigma'_u = \begin{cases}\sigma_u - \degree(u) & u = v \\ \sigma_u + 1 & u \in \neighbor(v) \\ \sigma_u & \text{otherwise}\end{cases}.\]

The configuration $\sigma'$ obtained by firing any vertex $u$ in a configuration $\sigma$, denoted by $\fire(\sigma, u)$, is called the successor of $\sigma$. By definition, the sum of chips will remain constant after any firing operation, i.e. $\sum_{u \in V} \sigma_u = \sum_{u \in V} \sigma'_u$ for any configuration $\sigma$ and its successor $\sigma'$.

The firing operation can be viewed as adding a vector to the configuration vector. We use $F(u)$ to denote the following vector of length $n$:

$$F(u)_v = \begin{cases} 1 & v \in \neighbor(u) \\
-\degree(u) & v = u\\
0 & \text{otherwise} \end{cases}, v \in V(G)$$ Then the configuration obtained by firing vertex $u$ is $\sigma + F(u)$. Note that $F(u)$ is a column vector of $G$'s Laplacian matrix.

For a given sandpile instance $S = (G, \sigma)$, if there is no full vertex, we say $\sigma$ is a terminal configuration. A sandpile instance is a terminal instance if it is possible to perform a finite number of firing operations to obtain a terminal configuration. Otherwise, we call it a recurrent instance.

In solving the sandpile prediction problem, there is one key background theorem:
\begin{theorem}[\cite{bjorner1991chip}]
\label{theorem:orders}
    For any terminal instance of the sandpile prediction problem, the terminal configuration and the number of times that each vertex fires are both unique and independent of the order of firings.
\end{theorem}


\Cref{theorem:orders} shows that the number of firings is independent of the order of firings. Thus, for a sandpile instance, we can well-define the firing number, indicating the number of firings performed on each vertex to make the configuration terminal. Formally:

\begin{definition}[Firing number]
\label{def:firing-number}
Given a terminal sandpile instance $S(G, \sigma)$, consider the process of firing all full vertices until the configuration is terminal. The number of firings performed on each vertex $v$ is denoted by $\bc(v)$.
\end{definition}

Because of the commutativity of configuration addition, if we can calculate the firing number $\bc(v)$ for each vertex $v$, then we can easily find that the terminal configuration is 
\begin{align}
\sigma + \sum_{v \in V(G)} \bc(v) \cdot F(v)\label{formula:recover-from-firing-number}
\end{align}

\subsection{Local Behavior on Trees}
\label{sec:overview:definitions-for-local-behaviors}
We root the tree at an arbitrary vertex $r$ and further define $\subtree(v),\son(v),\parent(v)$ for each vertex $v$ as the vertex set of its subtree (including $v$), children and its direct ancestor, respectively. Our algorithm relies on the concept of computing the outcomes after all firings in subtrees have occurred. Thus, it is crucial to establish clear definitions for events and configurations within a subtree. Furthermore, we also need to prove that they are consistent with the global behavior of the entire tree.

\begin{definition}[Local terminal configuration]
\label{def:localterminalconf}
Let $S(G, \sigma)$ be a sandpile instance. For $S \subseteq V(G)$, if all the vertices $v \in S$ satisfy $\sigma_v < \degree(v)$, then $\sigma$ is said to be local terminal in $S$.

Specially, if $G$ is a tree rooted at $r$. For a vertex $u \in V(G)$, if all the vertices $v \in \subtree(u)$ satisfy $\sigma_v < \degree(v)$, then $\sigma$ is local terminal in $\subtree(u)$.
\end{definition}

\Cref{theorem:orders}, which shows that the terminal configuration is unique for any sandpile instance, can be generalized to any local subset of vertices $S \subseteq V$. Formally we have the following two lemmas:



\begin{lemma}[Unique local terminal configuration]
\label{lemma:unique-terminal-configuration}
Let $S \subseteq V(G)$ be any subset of vertices. Suppose the process that keeps firing all the full vertices in $S$ until $\sigma$ is local terminal in $S$. Then: 

\begin{enumerate}
    \item Any firing order will reach the same local terminal configuration.
    \item For each vertex $u$, any firing order will fire $u$ the same number of times.
\end{enumerate}

\end{lemma}

\Cref{lemma:unique-terminal-configuration} is proved in \Cref{sec:proofinprelim}.

\begin{definition}[Local finalize operation]
\label{def:local-final}
Let $S(G, \sigma)$ be a sandpile instance where $G$ is a tree rooted at $r$. For a vertex $u \in V(G)$, let $\final(\sigma, u)$ be the configuration obtained by firing all full vertices in the subtree of $u$ until every vertex in $\subtree(u)$ is not full. Formally, 
$$\final(\sigma, u) = \begin{cases}
    \sigma & \sigma \text{ is local terminal in the subtree of }$u$ \\
    \final(\fire(\sigma, v), u) & v \in \subtree(u) \wedge \sigma_v \geq \degree(v)
\end{cases}$$

We also define the partial firing numbers, denoted by $\bc^{\down}(\sigma, u)$, as the number of the firing operations performed on vertex $u$ to make $\sigma$ local terminal in the subtree of $u$. Formally,

$$\bc^{\down}(\sigma, u) = \begin{cases} 0 
    & \sigma \text{ is local terminal in the subtree of }$u$ \\
    \bc^{\down}(\fire(\sigma, v), u) + [u=v] & v \in \subtree(u) \wedge \sigma_v \geq \degree(v)\end{cases}$$
    
\end{definition}
By \Cref{lemma:unique-terminal-configuration}, for any $S \subseteq V(G)$, any firing order in the set $S$ will lead to the same local terminal configuration in $S$, and the value of $\bc^{\downarrow}(\sigma,u)$ is also independent of the order of the firings. Thus the definitions in \Cref{def:local-final} are well-defined. We will use the notation $\bc^{\down}(u)$ to denote $\bc^{\down}(\sigma, u)$ as we are only considering a single given sandpile instance $S(G, \sigma)$.

By \Cref{theorem:orders}, the terminal configuration is independent of the order of firings. Thus any orders of the firings will obtain the same final configuration, which gives us:

\begin{lemma}
\label{lemma:fire-in-final}
Let $\sigma$ be a configuration and $\sigma'$ be a configuration obtained by performing several firing operations in $\subtree(u)$ on $\sigma$. For any configuration $\sigma^{*}$, $\final(\sigma + \sigma^{*}, u) = \final(\sigma' + \sigma^{*}, u)$.
\end{lemma}

\begin{lemma}
\label{lemma:final-sigma-add}
Let $\sigma$ and $\sigma'$ be any two configurations and $u \in V(G)$ be any vertex. Then $\final(\sigma + \sigma', u) = \final(\final(\sigma, u) + \final(\sigma', u), u)$
\end{lemma}

\Cref{lemma:fire-in-final} and \Cref{lemma:final-sigma-add} are proved in \Cref{sec:proofinprelim}.

We use $\final(\sigma)$ to refer to $\final(\sigma, r)$, and \Cref{problem:prediction} is equivalent to find $\final(\sigma)$.












\section{Sandpile Prediction on Trees}
\label{sec:overview}

The main idea of our algorithm is to compute the value of $\bc(v)$ for all $v \in V(G)$. After that, we can apply \Cref{formula:recover-from-firing-number} to retrieve the terminal configuration. It is difficult to calculate the value of $\bc(v)$ directly. However, we are able to complete the calculation by two steps: Partial Firing and Complete Firing.

\subsection{Partial Firing} 

We root the tree at an arbitrary vertex $r$. The first phase reduces the configuration $\sigma$ to a state which is local terminal in all vertices excluding $r$. In other words, after this round of firings, all the vertices other than $r$ are not full. This is done by firing from bottom to top, and as a result, we compute $\bc^{\down}(v)$ for all vertices $v \in V(G)\setminus\{r\}$. We propose \Cref{algorithm:upward} to correctly and efficiently finish this phase, which will be discussed later in this section. The following lemma summarizes the process:

\begin{restatable}{lemma}{upwardtheorem}\label{theorem:upward}
    \upward($u$, $G$, $\sigma'$) computes the value of all $\bc^{\down}(v)$ for all $v \in \subtree(u)$. In particular, it computes the value of $\bc^{\down}(v)$ for all $v \in V(G)$ in $O(n \log n)$ time. It also converts the initial configuration $\sigma$ into another configuration that is local terminal in the subtree of $u$ for any non-root vertex $u$. 
\end{restatable}

\IncMargin{1em}
\begin{algorithm}[H]
  \SetKwData{Left}{left}\SetKwData{This}{this}\SetKwData{Up}{up}
  \SetKwFunction{SolveUpward}{\upward}
  \SetKwFunction{SolveDownward}{\downward}
  \SetKwFunction{CalculateNumberOfFirings}{\cnt}
  \SetKwFunction{UpdateDataStructure}{\updatedsupward}
  \SetKwFunction{MergeDataStructure}{\mergeupward}
  \SetKwFunction{DeltaSum}{\deltasum}
  \SetKwFunction{InitializeDataStructure}{\initds}
  \SetKwFunction{PushDown}{\pushdown}
  \SetKwInOut{Input}{input}\SetKwInOut{Output}{output}  
  \SetKwFunction{FSolveUpward}{SolveUpward}\SetKwFunction{FDFSUpward}{DFS-Upward}

  \SetKwProg{Fn}{Function}{:}{}
  \BlankLine
    $D_u \leftarrow \varnothing$\; \label{algorithm:upward:init}
  $\dfn_u \leftarrow \indexx$\;\label{algorithm:upward:set-dfn}
  $\indexx \leftarrow \indexx + 1$\;\label{algorithm:upward:index}
  \If{$\son(u) = \varnothing$}{ \label{algorithm:upward:leaf-begin}
    $\sigma'_{parent(u)} \leftarrow \sigma'_{parent(u)} + \sigma'_u$\;
    $\bc^{\down}(u) \leftarrow \sigma'_u$\;
    $\sigma'_u \leftarrow 0$\;
    \Return\;
  }\label{algorithm:upward:leaf-end}
  \For{$v \in \son(u)$ in arbitrary order $\mathcal{I}$}{\label{algorithm:upward:merge-ds-begin}
    \SolveUpward($v$, $G$, $\sigma'$)\;
    \MergeDataStructure($u$, $v$)\;\label{algorithm:upward:merge-ds-mid}
  }\label{algorithm:upward:merge-ds-end}
  $\presigma_u\leftarrow \sigma'_u$\;
  $k \leftarrow$ \CalculateNumberOfFirings($u$, $\sigma'_u$)\; \label{algorithm:upward:k}
  $\bc^{\down}(u) \leftarrow k$\;\label{algorithm:upward:c-down}
  $\sigma'_u \leftarrow \sigma'_u + \DeltaSum(u) - \bc^{\down}(u) \cdot \degree(u)$\; \label{algorithm:upward:psi}
  \If{$u$ is not the root of $G$}{\label{algorithm:upward:update-parent-begin}
    $\sigma'_{parent(u)} \leftarrow \sigma'_{parent(u)} + k$\;
  }\label{algorithm:upward:update-parent-end}
  \UpdateDataStructure($D_u$)\;\label{algorithm:upward:update-ds}
  \caption{\upward($u$, $G$, $\sigma'$)}\label{algorithm:upward}
\end{algorithm}\DecMargin{1em}

Assume that we are currently visiting vertex $u$. Since the process is from bottom to top, we further assume the computation has already been done on $u$'s children. The main difficulty of such recursive computation is that after we fire $u$, chips will be sent to its subtree and might cause further firings in the subtree. What's worse, the firing in the subtree might cause another firing on $u$ if they return enough chips back to $u$. Such repetition could happen many times before every vertex in $\subtree(u)$ (including $u$ itself) becomes not full. 

We want to figure out a way to avoid these repetitions. By maintaining extra information of each child, we can safely compute the state after firing $u$ to not full without going down into $u$'s subtree again. More precisely, for any vertex $u$, we maintain how many chips will be returned to the parent of $u$ after $x$ chips are added to the vertex $u$ and all full vertices in $\subtree(u)$ were fired so that the configuration becomes local terminal in $\subtree(u)$. The formal definition of such quantity is as follows.

\begin{definition}[Local Upward Contribution]
\label{def:delta}
Let $S(G, \sigma)$ be a sandpile instance where $G$ is a tree rooted at $r$. For a vertex $u \in V(G)$ ($u \ne r$) such that $\sigma$ is local terminal in the subtree of $u$, the local upward contribution of adding $x$ chips to the vertex $u$ is denoted as $\delta(u,x)$, where $\delta(u, x) = \final(\sigma + x_u, u)_{\parent(u)} - \final(\sigma, u)_{\parent(u)}$. $x_u$ denotes a vector of all zeros except the value of the $u$-th term is $x$.
\end{definition}

\Cref{lemma:remaining-chips} is proved in \Cref{sec:proofinprelim}. It shows that the number of the remaining chips on vertex $u$ after firing $u$ exactly $k$ times and make $\sigma$ be local terminal in all $\subtree(v_i)$ for $v_i \in \son(u)$ is exactly $\psi_u(k) \stackrel{\text{def}}{=} \sigma_u - k \cdot \degree(u) + \sum_{v \in \son(u)} \delta(v, k)$. 

\begin{lemma}
\label{lemma:remaining-chips}
Let $u \in V(G)$ and $\sigma$ be local terminal in the subtree of all its children $v_i \in \son(u)$. For any positive integer $k$, if $\psi_u(k-1) \geq \degree(u)$, then 
\begin{itemize}
    \item It is possible to fire vertex $u$ at least $k$ times without firing any vertex not in $\subtree(u)$.
    \item Assume we fired vertex $u$ exactly $k$ times, and fired all full vertices in $\subtree(v_i)$ for all $v_i \in \son(u)$, while not firing any vertex outside $\subtree(u)$. Then the number of chips at vertex $u$ is exactly $\psi_u(k)$.
\end{itemize}
\end{lemma}

We further show $\delta(u,k)$ has monotonicity:

\begin{lemma}
\label{lemma:delta-differs-at-most-one}
For any vertex $u \in V(G) \notin r$ and integer $k \geq 0$, $\delta(u,k) \leq \delta(u,k+1) \leq \delta(u,k)+1$.
\end{lemma}

\begin{proof}


We prove the lemma by induction. For all the leaf vertices $u$, $\delta(u,k) = k$ must be held. So the lemma is correct for all the leaf vertices.

Consider any vertex $u \in V(G)$, and for all vertices $v \in \son(u)$ the inequality $\delta(v,k) \leq \delta(v,k+1) \leq \delta(v,k)+1$ holds for all non-negative integers $k$ by the inductive hypothesis. Consider $\sigma' = \final(\sigma + k_{u}, u)$, there are two cases.

\begin{enumerate}
    \item $\sigma'_u < \degree(u) - 1$. Then putting one more chip on the vertex $u$ does not make more firing operations available, since $\sigma'_u + 1 < \degree(u)$. So $\delta(u,k+1) = \delta(u,k)$ in this case.
    \item $\sigma'_u = \degree(u) - 1$. Then we perform one firing operation on vertex $u$. All the children $v \in \son(u)$ will receive one more chip after the operation, but since $\delta(v,k+1) \leq \delta(v,k) + 1$ holds for all $k$ on vertex $v$, there will be at most one more chip received from vertex $v$ after making $\sigma$ being local terminal in $\subtree(v)$ again. So there will be no more than $|\son(u)|$ chips after doing all firing operations in $\subtree(u)$. Since $|\son(u)| < \degree(u)$ for all $u \ne r$, no more fire operations on vertex $u$ are possible. So there will be exactly one additional firing operation performed on vertex $u$, thus $\delta(u,k+1) = \delta(u,k) + 1$ in this case.
\end{enumerate}

This shows $\delta(u,k+1) \in \{\delta(u,k), \delta(u,k) + 1\}$. Thus $\delta(u,k) \leq \delta(u,k+1) \leq \delta(u,k)+1$
\end{proof}

\begin{lemma}
\label{lemma:f-is-monotonically-decreasing}
For any vertex $u \in V(G)$ the $\psi_u(k)$ is monotonically non-increasing. In other words, $\psi_u(k) \geq \psi_u(k+1)$ for all $k \in \N$.
\end{lemma}

\begin{proof}
By the definition $\psi_u(k) = \sigma_u - k \cdot \degree(u) + \sum_{v \in \son(u)} \delta(v, k)$, we have $\psi_u(k+1) - \psi_u(k) = -\degree(u) + \sum_{v \in \son(u)} \left(\delta(v,k+1) - \delta(v,k)\right)$. By \Cref{lemma:delta-differs-at-most-one}, $\delta(v,k+1) - \delta(v,k) \leq 1$, so $\sum_{v \in \son(u)} \left(\delta(v,k+1) - \delta(v,k)\right) \leq |\son(u)| \leq \degree(u)$. This proves $\psi_u(k+1) - \psi_u(k) \leq 0$, thus $\psi_u(k) \geq \psi_u(k+1)$.
\end{proof}

\begin{lemma}
\label{lemma:c-down}
Let $k$ be the smallest non-negative integer such that $\psi_u(k) < \degree(u)$. Then $\bc^{\downarrow}(u) = k$.
\end{lemma}

\begin{proof}
By the definition of $k$ we have either $k = 0$ or $\psi_u(k-1) \geq \degree(u)$. 

If $k = 0$, we have $\psi_u(0) = \sigma_u < \degree(u)$, which means we can not perform any operation on vertex $u$. Otherwise, by \Cref{lemma:remaining-chips}, we can perform $k$ firing operations on vertex $u$, and there are $\psi_u(k)$ chips located on vertex $u$ after all these firings. Since $\psi_u(k) < \degree(u)$ and $\sigma$ became local terminal in all the subtree of $v_i$ for $v_i \in \son(u)$, the current configuration must be local terminal in the subtree of $u$. Thus $\bc^{\downarrow}(u) = k$. 
\end{proof}

By \Cref{lemma:c-down}, our task is to find the smallest integer $k$ such that $\psi_u(k) < \degree(u)$ on a monotonically non-increasing function $\psi_u$. This integer $k$ is exactly the value of $\bc^{\down}(u)$. We use a data structure $D_u$ to maintain the following value: 
\begin{itemize}
    \item The value of $\sum_{v \in \son(u)} \delta(v, k)$ for a given vertex $v$ and integer $k$.
    \item The smallest integer $k$ such that $\psi_u(k) < \degree(u)$.
\end{itemize}

Since the function $\psi_u(k)$ is monotonically non-increasing (\Cref{lemma:f-is-monotonically-decreasing}), the value of $\bc^{\down}(u)$ can be found by performing a binary search procedure. To check if a specific value $k_0$ meets the inequality $\sigma_u - k \cdot \degree(u) + \sum_{v \in \son(u)} \delta(v,k) < \degree(u)$, we need to find the value of $\sum_{v \in \son(u)} \delta(v, k)$ efficiently. This is done by maintaining and querying on $D_u$. It supports the following queries, allowing us to speed up the calculation process:

\begin{itemize}
    \item \cnt($u$,$\sigma'_u$): return the value of $\bc^{\downarrow}$. In the function, we use the merged data structure $D_u$ to speed up the computation.
    \item \deltasum($u$): return the value of $\sum_{v \in \son(u)} \delta(v, \bc^{\down}(u))$, where $k$ is the returned value of $\cnt(u)$ (in other words, $k =\bc^{\down}(u)$).
\end{itemize}

Furthermore, it supports the following modifications to update the status of the data structures.

\begin{itemize}
    \item \mergeupward($u$, $v$): merge all information from $D_v$ to $D_u$. 
     \item \updatedsupward($u$): Update the information in $D_u$ to match the current vertex $u$.
\end{itemize}

\Cref{theorem:ds} ensures that the data structure costs $O(n \log n)$ time in total in the whole procedure of our algorithm to handle all the requests. Implementation details will be discussed in \Cref{sec:ds}.

Throughout the algorithm, we maintain two global arrays $\dfn$ and $\num$ and two global variables $\indexx$ and $r$. Note that these will be used in both phases and data structures. $r$ denotes the root and $\indexx$ is set to keep track of the visit order in $\upward$, which will be stored to $\dfn_u$ when visiting $u$. 

If $u$ is a leaf, since $\son(u) = \varnothing$, firing vertex $u$ is equivalent to moving a chip from vertex $u$ to $\parent(u)$. So, we have $\bc^{\down}(u) = \sigma_u$, and the algorithm will update $\sigma'_{\parent(u)}$ and $\sigma'_u$ correspondingly (\Cref{algorithm:upward:leaf-begin} to \Cref{algorithm:upward:leaf-end}).

Otherwise, the procedure initializes the data structure $D_u$ (\Cref{algorithm:upward:init}). It maintains the visit order for each vertex (\Cref{algorithm:upward:set-dfn} to \Cref{algorithm:upward:index}). It ensures that the earlier the vertex is accessed, the smaller $\dfn$ value is given. After that, we merge all the children of $u$ together to get $D_u$ (\Cref{algorithm:upward:merge-ds-begin} to \Cref{algorithm:upward:merge-ds-end}; here, $\mathcal{I}$ is an arbitrary order of merging the children). By \Cref{lemma:c-down}, the variable $k$ computed in \Cref{algorithm:upward:k} is exactly the value of $\bc^{\down}(u)$. Then, it computes the number of the remaining chips on $u$ after finishing all firings in $\subtree(u)$ (\Cref{algorithm:upward:psi}). By \Cref{lemma:remaining-chips}, the number of the chips on vertex $u$ will be changed to $\psi_u(k)$. The value of $\sum_{v \in \son(u)} \delta(v, k)$ can be computed by \deltasum($u$, $k$). Finally, we update the value of $\sigma'_{\parent(u)}$ (\Cref{algorithm:upward:update-parent-begin} to \Cref{algorithm:upward:update-parent-end}) and the data structure $D_u$ (\Cref{algorithm:upward:update-ds}), so that $D_u$ has the full information in $\subtree(u)$.

Note that we did not update the value of $\sigma'_v$ for $v \in \son(u)$ explicitly. This is because the number of chips moved from $v$ to $u$ is already calculated as $\delta(v, k)$ (\Cref{algorithm:upward:psi}). After visiting vertex $u$, we will no longer use the value of $\sigma'_v$ for all $v \in \subtree(u)$. Thus,the value of $\sigma'_v$ for $v \in \subtree(u)$ can be ignored.

\subsection{Complete Firing} 

After we calculate the values of $\bc^{\down}(u)$ for all $u \in V(G)$, we will recover all the $\bc(u)$ from the top to the bottom. This process is based on the relationship between firing numbers and partial firing numbers described in \Cref{lemma:value-of-c}  (Proof can be found in \Cref{sec:proofinprelim}).

\begin{lemma}
\label{lemma:value-of-c}
For each vertex $u \in V(G)$ such that $u \ne r$, $\bc(u) = \bc^{\down}(u) + \delta(u, \bc(\parent(u)))$. Specially, for the root vertex, $\bc(r) = \bc^{\down}(r)$,
\end{lemma}

By \Cref{lemma:value-of-c}, as long as we maintained $c(u)$ and the value of $\delta(u, i)$ for all required $i$ correctly, we can recursively recover all the values of $c(v)$ for $v \in \subtree(u)$. To speed up the whole process, we need to extend the data structure we described in the first phase of our algorithm. Specifically, we subsequently traverse the tree and compute the firing number $\bc(u)$ based on the results obtained from $\parent(u)$. We also need to restore information for vertices in $\subtree(u)$ before visiting them, which is done by reverting operations on the data structure.

In \Cref{algorithm:upward} we use $\mergeupward(D_u, D_v)$ to make every $D_u$ store the information of $\subtree(u)$. Now we need to revert all these changes. For each vertex $u \in V(G)$, we will revert the changes that were made in $\updatedsupward(D_u)$. After that for each child $v \in \son(u)$ we will recover the structure of $D_v$ by splitting $D_u$.

We denote $\deltaquery(u, k)$ as a function that returns the value of $\delta(u, k)$. In this function, we will use the data structure $D_u$ to speed up the computation.

In addition, we need the following interface to update the data structure:

\begin{itemize}
    \item \revertds($u$): revert the data structure $D_u$ to the one before the procedure $\cnt(u)$ called.
    \item \splitds($u$, $v$): split the data structure $D_v$ from $D_u$. The data structure $D_v$ will become the one before the procedure \mergeupward($u$, $v$) called.
\end{itemize}

Using the data structure described in \Cref{theorem:ds} (proved in \Cref{sec:ds:overall-analysis-of-ds}), we can prove that the values of $\bc(u)$ for all $u \in V(G)$ are computed correctly after calling $\downward(r, G)$.

\IncMargin{1em}
\begin{algorithm}
  \SetKwData{Left}{left}\SetKwData{This}{this}\SetKwData{Up}{up}
  \SetKwComment{Comment}{$\triangleright$\ }{}
  \SetKwFunction{SolveUpward}{\upward}
  \SetKwFunction{SolveDownward}{\downward}
  \SetKwFunction{CalculateNumberOfFirings}{\cnt}
  \SetKwFunction{UpdateDataStructure}{\updatedsupward}
  \SetKwFunction{MergeDataStructure}{\mergeupward}
  \SetKwFunction{DeltaSum}{\deltasum}
  \SetKwFunction{InitializeDataStructure}{\initds}
  \SetKwFunction{DeltaQuery}{\deltaquery}
  \SetKwFunction{RevertDataStructure}{\revertds}
  \SetKwFunction{SplitDataStructure}{\splitds}
  \SetKwInOut{Input}{input}\SetKwInOut{Output}{output}  
  \SetKwFunction{FSolveDownward}{SolveDownward}

  \SetKwProg{Fn}{Function}{:}{}
  \If{$u$ is the root of $G$}{\label{algorithm:downward:k-begin}
    $k \leftarrow 0$\; \label{algorithm:downward:c-root}
  }
  \Else {
    $k \leftarrow \DeltaQuery(u, \bc(\parent(u)))$\; \label{algorithm:downward:calck}
    \Comment*[r]{Use the maintained data structure $D_u$ to compute $\delta(u, \bc(parent(u)))$.}\label{algorithm:downward:c-other}
  }\label{algorithm:downward:k-end}
  \RevertDataStructure($u$)\; \label{algorithm:downward:revertds}
  $\bc(u) \leftarrow \bc^{\down}(u) + k$\;\label{algorithm:downward:c-update}\label{algorithm:downward:c}
  \For{$v \in \son(u)$ in the reversed order of $\mathcal{I}$}{\label{algorithm:downward:iterate-sons-begin}
  
  \Comment*[l]{Iterate the children of $u$ in the reversed order}
    \SplitDataStructure($u$, $v$)\;
    \SolveDownward($v$, $G$)\;
  }\label{algorithm:downward:iterate-sons-end}
  \BlankLine

  \caption{\downward($u$, $G$)}\label{algorithm:downward}
\end{algorithm}\DecMargin{1em}

The following theorem can summarize the whole phase.

\begin{restatable}{lemma}{downwardtheorem}\label{theorem:downward}
    \downward($u$, $G$) computes the value of all $\bc(v)$ for all $v \in \subtree(u)$, based on the value of $\bc^{\down}(u)$ found in the \upward{} part. In particular, it can compute the value of $\bc(v)$ for all $v \in V(G)$ in $O(n \log n)$ time. 
\end{restatable}

In \Cref{algorithm:downward}, we first handle the case if $u$ is the root, in which we have $\bc(u) = \bc^{\down}(u)$. If $u$ is not the root, then we calculate the value of $\delta(\parent(u), \bc(\parent(u)))$ using the data structure (\Cref{algorithm:downward:c-other}). After that, we are able to update the value of $\bc(u)$ (\Cref{algorithm:downward:c}) and revert the data structure $D_u$ to allow us to proceed with queries on the vertex $u$ (\Cref{algorithm:downward:revertds}). Finally, we will recursively process all the children of $u$ in the reversed order of $\mathcal{I}$ (\Cref{algorithm:downward:iterate-sons-begin} to \Cref{algorithm:downward:iterate-sons-end}), where $\mathcal{I}$ is the same as the one in \Cref{algorithm:upward}. This makes sure the algorithm reverts all the merging in the correct order.



\subsection{Overall Analysis} 
\IncMargin{1em}
\begin{algorithm}
  \SetKwData{Left}{left}\SetKwData{This}{this}\SetKwData{Up}{up}
  \SetKwComment{Comment}{$\triangleright$\ }{}
  \SetKwFunction{SolveUpward}{\upward}
  \SetKwFunction{SolveDownward}{\downward}
  \SetKwFunction{CalculateNumberOfFirings}{\cnt}
  \SetKwFunction{UpdateDataStructure}{\updatedsupward}
  \SetKwFunction{MergeDataStructure}{\mergeupward}
  \SetKwFunction{DeltaSum}{\deltasum}
  \SetKwFunction{InitializeDataStructure}{\initds}
  \SetKwFunction{DeltaQuery}{\deltaquery}
  \SetKwFunction{RevertDataStructure}{\revertds}
  \SetKwFunction{SplitDataStructure}{\splitds}
  \SetKwInOut{Input}{input}\SetKwInOut{Output}{output}  
  \SetKwFunction{FSolve}{Solve}

  \SetKwProg{Fn}{Function}{:}{}
  \Input{tree $G$, configuration $\sigma$}
  \Output{the terminal configuration $\sigma^{T}$ of the instance $S(T, \sigma)$}
  \BlankLine 
  \If{$\sum_{u \in V(G)} \sigma_u > |V(G)| - 2$}{ \label{algorithm:main:sp-start}
    \Comment*[l]{$S(G, \sigma)$ must be a recurrent instance}
    \Return $\bot$\;
  }\label{algorithm:main:sp-end}
  $r$ $\leftarrow$ arbitrary vertex in $V(G)$\; \label{algorithm:main:root}
  $\sigma'$ $\leftarrow$ $\sigma$\; \label{algorithm:main:init}
  $\indexx \leftarrow 1$\;\label{algorithm:main:index}
  \SolveUpward($r$, $G$, $\sigma'$)\; \label{algorithm:main:upward}
  \SolveDownward($r$, $G$)\;\label{algorithm:main:downward}
  \For{$u \in V(G)$}{\label{algorithm:main:recover-start}
    $\sigma_u \leftarrow \sigma_u - \bc(u) \cdot \degree(u)$\;
    \For{$v \in \neighbor(u)$}{
        $\sigma_v \leftarrow \sigma_v + \bc(u)$\;
    }
  }\label{algorithm:main:recover-end}
  \Return{$\sigma$}\;
  
  \caption{\textsc{Solve}($G$, $\sigma$)}\label{algorithm:main}
\end{algorithm}\DecMargin{1em}
Finally, we present the main structure of our proposed algorithm in \Cref{algorithm:main}. In the beginning of the algorithm, we skip the case if the given instance is recurrent. This is done by applying the following lemma, which is proved by directly applying Theorem 3.3 in~\cite{bjorner1991chip}. 

\begin{lemma}
\label{lemma:bounds-of-terminal-instance}
$S(G,\sigma)$ be a terminal instance if and only if $\sum_{v \in V(G)} \sigma_v \leq |V(G)| - 2$.
\end{lemma}

We root the tree at an arbitrary vertex $r$ (\Cref{algorithm:main:root}). We initialize $\sigma'$ as the input configuration (\Cref{algorithm:main:init}) and set a global variable $\indexx$ to $1$ (\Cref{algorithm:main:index}). After that, we call $\upward$ to compute all $\bc^\down$ values, the partial firing numbers when we only consider the final state of the subtree (\Cref{algorithm:main:upward}). Subsequently, we further call $\downward$ to compute all $\bc$ values (\Cref{algorithm:main:downward}), which can be converted to the final terminal configuration (\Cref{algorithm:main:recover-start} to \Cref{algorithm:main:recover-end}).

\paragraph{Performance Analysis} The performance of our algorithm depends on the data structure we use throughout \Cref{algorithm:upward} and \Cref{algorithm:downward}. We show that by using splittable binary search trees, we can implement a data structure that supports all operations in $O(n \log n)$ time in total, with $O(n)$ memory. The implementation details of such a data structure are discussed in \Cref{sec:ds}. As a result, we prove \Cref{theorem:main} in \Cref{sec:ds:overall-analysis-of-ds} in the end.


\paragraph{Sandpile Prediction on Paths} Paths can be considered a special variant of trees, and our algorithm successfully demonstrates the unification of these two graph structures. Furthermore, based on the \textit{Dynamic Optimality Conjecture} \cite{sleator1985self} of the splay tree, it is conceivable to conjecture that our algorithm on trees could potentially achieve a linear runtime if the input graph is a path. As a result, we have successfully modified our algorithm to leverage the \textit{Dynamic Finger Theorem} \cite{cole2000dynamic,cole2000dynamic2} instead, leading to a provable linear runtime in \Cref{theorem:main-path}. Details are analyzed in \Cref{sec:path}.
%

\section{Sandpile Prediction on General Graphs}
\label{sec:general}

In this section, we discuss developing efficient sandpile prediction algorithms on general graphs. The general idea is to transform the prediction problem on an arbitrary graph into problems on its subgraphs separated by any vertex set. Since our reduction creates sink vertices, we first introduce sinks to the sandpile model in \Cref{sec:sanpilewithsink}. Then, we propose a simulation-based algorithm that works on arbitrary graphs with sinks in \Cref{sec:simu}. We discuss its performance on various types of graphs, including regular graphs, expander graphs, and hypercubes. In the end, we propose the reduction scheme that decomposes the graphs by vertex removal in \Cref{sec:reduction}.

\subsection{Sandpile with Sinks}
\label{sec:sanpilewithsink}

\begin{figure}[h]
\centering
\begin{tikzpicture}
 \tikzstyle{every node}=[fill=black!30,circle,minimum size=0.6cm,inner sep=1pt]

  \node (1) at (0,0) {2}; 
  \node (2) at (-1,-1){0}; 
  \node[fill=blue!50] (3) at (-1,1){}; 
  \node (4) at (1,-1){1}; 
  \node (5) at (1,1){1}; 
  \node[fill=blue!50] (6) at (2,0){}; 
  \node (7) at (3,-2){1};
  \node[fill=blue!50] (8) at (3,-1){};
  \node (9) at (3,0){0};
  \node (10) at (3,1){0};

  \foreach \from/\to in {1/2,1/3,1/4,1/5,1/6,4/6,5/6,4/7,6/9,6/7,6/8,7/8,8/9,5/10,6/10}
    \draw (\from) -- (\to);

\end{tikzpicture}
\caption{A sandpile instance with sinks. The sinks in the figure are all marked as blue.}
\label{figure:sinks-example}
\end{figure}
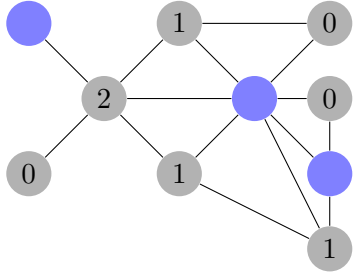

\begin{figure}[h]
\centering
\begin{tikzpicture}

\begin{scope}[shift={(-4,0)}]
 \tikzstyle{every node}=[fill=black!30,circle,minimum size=0.6cm,inner sep=1pt]
  \node (1) at (0,0) {2}; 
  \node (2) at (-1,-1){0}; 
  \node[fill=blue!50] (3) at (-1,1){}; 
  \node[fill=red!60,label=$u$] (4) at (1,-1){3}; 
  \node (5) at (1,1){1}; 
  \node[fill=blue!50] (6) at (2,0){}; 
  \node (7) at (3,-2){1};
  \node[fill=blue!50] (8) at (3,-1){};
  \node (9) at (3,0){0};
  \node (10) at (3,1){0};

  \foreach \from/\to in {1/2,1/3,1/4,1/5,1/6,4/6,5/6,4/7,6/9,6/7,6/8,7/8,8/9,5/10,6/10}
    \draw (\from) -- (\to);
\end{scope}

\begin{scope}[shift={(0,0)}]
\draw[line width=2pt,->] (0,0) -- (2,0);
\node at (1,-0.5){$\fire(u)$};
\end{scope}

\begin{scope}[shift={(4,0)}]
 \tikzstyle{every node}=[fill=black!30,circle,minimum size=0.6cm,inner sep=1pt]
  \node (1) at (0,0) {3}; 
  \node (2) at (-1,-1){0}; 
  \node[fill=blue!50] (3) at (-1,1){}; 
  \node[label=$u$] (4) at (1,-1){0}; 
  \node (5) at (1,1){1}; 
  \node[fill=blue!50] (6) at (2,0){}; 
  \node (7) at (3,-2){2};
  \node[fill=blue!50] (8) at (3,-1){};
  \node (9) at (3,0){0};
  \node (10) at (3,1){0};

  \foreach \from/\to in {1/2,1/3,1/4,1/5,1/6,4/6,5/6,4/7,6/9,6/7,6/8,7/8,8/9,5/10,6/10}
    \draw (\from) -- (\to);
  \foreach \from/\to in {4/1,4/6,4/7}
    \draw[edge](\from) -- (\to);
\end{scope}

\end{tikzpicture}
\caption{A firing operation on a full vertex. The chips transferred to a sink vertex are ignored, and no firing operation can happen on a sink vertex.}
\label{figure:sinks-firing}
\end{figure}
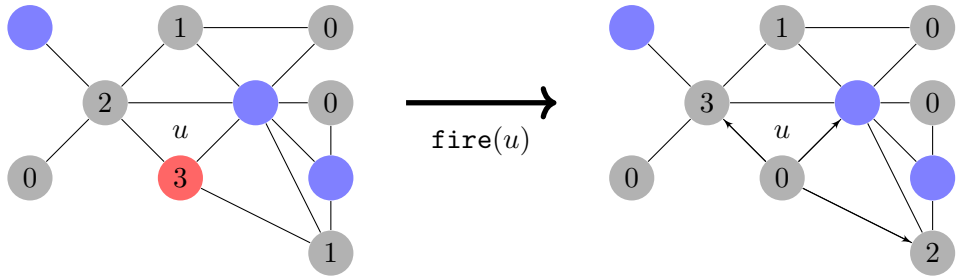

Although we focus on solving sandpile prediction on graphs without sinks (\Cref{problem:prediction}), our reduction scheme adds sinks into the graph to replace the removed vertices. Therefore, we need to introduce sinks to the sandpile model as a supportive tool to solve the prediction problem.

To begin with, we define the sandpile model on undirected graphs with multiple sinks. The model is generalized from Chapter 2.5 in \cite{klivans2018mathematics} which only defines the sandpile model with one sink in the graph. We study the first type of chip-firing process from Definition 2.5.1 in \cite{klivans2018mathematics}, in which the sinks do not fire. 
Correspondingly, we define the terminal configuration as follows: 
\begin{definition}[Terminal with Sinks]
For a given sandpile instance $S(G, \sigma, M)$ with a non-empty set of sinks $M$, we call the configuration $\sigma$ a terminal configuration if and only if for all vertices $u \in V(G) \setminus M$, we have $\sigma_u < \degree(u)$.

Call the instance $S(G, \sigma,M)$ a terminal instance if and only if it is possible to perform firing operations on any vertex not in $M$ to make the configuration terminal. A non-terminal instance is called a recurrent instance.
\end{definition}

\Cref{figure:sinks-example} shows an example of a sandpile instance with sinks. The integer marked on each vertex corresponds to the number of chips on the vertices. Since no firing operation could happen on the sink vertex, we can ignore the number of chips on the sink vertex and only care about the other vertices. Whenever a firing operation happens, as in \Cref{figure:sinks-firing}, the chips transferred on sink vertices can be treated as removed from the graph.

Like the model without sinks, sandpile models with sinks still preserve global and local uniqueness. That is, any order of the firing operations leads to the same (local) terminal configuration. We provide an elaborate analysis in \Cref{sec:sinkmodel}.

Moreover, any sandpile instance with sinks always terminates: 

\begin{lemma}[\cite{klivans2018mathematics}]
\label{lemma:sandpilesinkterminate}
For any sandpile instance with sinks $S(G, \sigma, M)$ such that $M \ne \varnothing$ and $G$ is a connected graph,  $S$ is always a terminal instance. That is, the terminal configuration always exists.
\end{lemma}

Therefore, the sandpile prediction problem on sandpile with sinks only needs to compute the terminal configuration, which does not need to determine the recurrent case. However, the number of chips can be arbitrarily large, which is likely to affect the time complexity if it becomes too large. We formally state such prediction problem as follows:

\begin{problem}[Sandpile Prediction with Sinks]
\label{pro:sink}
For a given sandpile instance $S = (G, \sigma,M)$, the sandpile prediction with sinks Problem is to compute the terminal configuration of $S$.
\end{problem}

In most cases, we can assume $|M|=1$. That is, there is only one sink in the graph. For a graph with multiple sinks, we can merge all sink vertices into one, and this does not affect the terminal configuration. However, such a merge may destroy the original structure of the graph. For example, if there is a tree with multiple sinks, after merging them, the graph has multiple cycles and is no longer a tree. In \Cref{sec:sinktree}, we show how to modify our tree algorithm to solve \Cref{pro:sink} on a tree that contains at most three sinks. On the other hand, in \Cref{sec:simu}, we provide a simulation-based algorithm that works on arbitrary graphs with only one sink. Together, these two algorithms actually demonstrate a trade-off between the number of sinks and the graph structure.

\subsection{Simulation-based Algorithm}

\label{sec:simu}
By merging all sinks altogether, we can ensure there is only one sink $s$ in the graph. In the following analysis, we always assume there is only one sink in the graph $G$. We propose \Cref{algorithm:simulation} that can be applied on any graph to solve \Cref{pro:sink}.

\IncMargin{1em}
\begin{algorithm}
  \SetKwData{Left}{left}\SetKwData{This}{this}\SetKwData{Up}{up}
  \SetKwComment{Comment}{$\triangleright$\ }{}
  \SetKwFunction{SolveUpward}{\upward}
  \SetKwFunction{SolveDownward}{\downward}
  \SetKwFunction{CalculateNumberOfFirings}{\cnt}
  \SetKwFunction{UpdateDataStructure}{\updatedsupward}
  \SetKwFunction{MergeDataStructure}{\mergeupward}
  \SetKwFunction{DeltaSum}{\deltasum}
  \SetKwFunction{InitializeDataStructure}{\initds}
  \SetKwFunction{DeltaQuery}{\deltaquery}
  \SetKwFunction{RevertDataStructure}{\revertds}
  \SetKwFunction{SplitDataStructure}{\splitds}
  \SetKwInOut{Input}{input}\SetKwInOut{Output}{output}  
  \SetKwFunction{FSolveDownward}{SolveDownward}

  \SetKwProg{Fn}{Function}{:}{}

    \While{$\sigma$ is not a terminal configuration}{ \label{simu:repeat}
        $u \leftarrow \argmax_{v\neq s} \left\lfloor\frac{\sigma_v}{\degree(v)}\right\rfloor$\; \label{simu:findmin}
        $k \leftarrow \left\lfloor \frac{\sigma_u}{\degree(u)}\right\rfloor$\; \label{simu:k}
        $\sigma_u \leftarrow \sigma_u - k \cdot \degree(u)$\; \label{simu:fire1}
        \For{$v\neq s$ such that $(u,v) \in E(G)$}{ \label{simu:fire2}
            $\sigma_v \leftarrow \sigma_v + k$\; \label{simu:fire3}
        }
    }
  
  \BlankLine

  \caption{\textsc{Solve}($G$, $\sigma$)}\label{algorithm:simulation}
\end{algorithm}\DecMargin{1em}

In every iteration, we pick a non-sink vertex $u$ with the maximum ratio of $\left\lfloor\sigma_u/\degree(u)\right\rfloor$ to fire (\Cref{simu:findmin}). $\left\lfloor\sigma_u/\degree(u)\right\rfloor$ is the number of firings that could happen on $u$ with the current number of chips, denoted as $k$ (\Cref{simu:k}). We apply all these firings at once and add $k$ to its neighbors (\Cref{simu:fire1} to \Cref{simu:fire3}). We repeat this process until there is no more vertex to fire (\Cref{simu:repeat}), thus we have the terminal configuration. Maintaining the vertex with maximum ratio can be done by heaps or other data structures that support insertions and the \textsc{findmin} operation.

To better capture the performance of this simulation-based algorithm, we provide analysis on the number of iterations that  \Cref{algorithm:simulation} needs to terminate on various types of graphs. Throughout the performance analysis of this algorithm, we use $N$ to denote the total number of chips on the non-sink vertices in the initial configuration.

\subsubsection{Performance Analysis on General Graphs}

\simulation*

To prove \cref{theo:simu}, we use the following result from \cite{holroyd2008chip} that gives an upper bound of the firing number.
\begin{lemma}[\cite{holroyd2008chip}]\label{lemma:boundsink}
    Let $G$ be an $m$-edge graph in which a special vertex is chosen to be a sink, and the maximum effective resistance between the sink and any other vertex is $R$. For any chip configuration $\sigma$ with a total number of $N$ chips on the non-sink vertices, the number of chips moves needed to stabilize $\sigma$ is bounded by $2mNR$.
\end{lemma} 

\begin{proof}[Proof of \cref{theo:simu}]
Let $t$ denote the current remaining number of firings. Let $N_c$ denote the total number of chips in the current configuration. By \cref{lemma:boundsink}, we have 
\begin{align} \label{formu:127}
    t \leq 2mN_cR
\end{align} at any time. The algorithm terminates when $t=0$. Let $t_0$ denote the initial total number of firings. Consider the current configuration $\sigma$, let $k=\max_{v \neq s} \lfloor \frac{\sigma_{v \neq s}}{\degree(v)} \rfloor$. If $k=0$, the algorithm terminates. Otherwise, we have 
\begin{align} \label{formu:128}
    k \geq \frac{\sum_{v\neq s}\sigma_v}{\sum_{v\neq s}\degree(v)}-1=\frac{N_c}{2m}-1.
\end{align} When $N_c \geq 4m$, $k \geq 1$, by applying \cref{formu:127} and \cref{formu:128}, we have 
\[t' = t-k  \leq t(1-\frac{1}{4Rm^2}+\frac{1}{2RN_cm})\leq t(1-\frac{1}{8Rm^2}).\] Otherwise, if $N_c \leq 4m$, we can apply the naive bound in \cref{lemma:boundsink} directly. Therefore, the algorithm terminates in $$
\log_{1-1/8Rm^2}\frac{1}{t_0} + 8Rm^2=
O(Rm^2\log(t_0))$$ iterations. Plugging in $t_0\le 2mN_c R\le 2n^3N$ gives us the result.
\end{proof}

\begin{remark}
    Note that the logarithmic dependency on the total number of chips $N$ is essential. Consider an $n$-vertex graph $G$ of max degree $\Delta$, our algorithm takes at least $\log_\Delta N$ iterations.
\end{remark}

\subsubsection{Performance Analysis on Structured Graphs}

Our proposed algorithm performs better on graphs with some special properties. Below, we present the method we will be using in performance analysis.

Let $t$ denote the current number of firings and let $\sigma^0,\sigma^1,\ldots, \sigma^t$ be a sequence of chip configurations, where $\sigma^{i+1}$ is got by firing some non-sink vertex in $\sigma^i$, $\sigma^0$ is the initial configuration and $\sigma^t$ is the terminal configuration. Define the weight $w(v)$ of vertex $v$ to be the expected time $\mathbb{E}_vT_s$ of a random walk $T_s$ that starts from $v$ to hit the sink $s$, and the weight of a chip configuration $$
w(\sigma^i)=\sum_{v\neq s}\sigma^i_vw(v).
$$ Conditioning on the first step $X_1$ of the random walk, we have $$
\Delta w(v) = \mathbb{E}_v(\mathbb{E}_{X_1} T_s -T_s) = -1.
$$ This shows that if we fire $v$ in $\sigma^i$ to get $\sigma^{i+1}$, the total weight goes down by $\degree(v)$.

\paragraph{Regular Graphs}

A $d$-regular graph is a graph in which every vertex has degree $d$. Note that in $d$-regular graphs, the number of edges is $nd/2$, and $R$ can be bounded by $n$. Plugging in these in \cref{theo:simu} gives us $O(d^2n^3\log(dn^2N)) = O(d^2n^3\log(nN))$. The following theorem shows the algorithm actually does better than that, especially when $d$ is large.

\begin{theorem}\label{theo:simuRegular}
    Given a sandpile instance $S(G,\sigma)$ such that $G$ is $d$-regular and contains exactly one sink, \cref{algorithm:simulation} terminates in $O(n^3\log(nN))$ iterations.
\end{theorem}

\begin{proof}[Proof of \cref{theo:simuRegular}] By $d$-regularity, the total weight goes down by exactly $d$ each time when firing a vertex. By \cite{lovasz1993random}, the weight of each vertex is bounded by $2n^2$ for regular graphs. Let $N_c$ denote the total number of chips in the current configuration. We have \begin{align}
    t=\frac{1}{d}\sum_{i=0}^{t-1}(w({\sigma^i})-w({\sigma^{i+1}}))  =\frac{1}{d}(w(\sigma^0)-w(\sigma^t))\le \frac{2n^2N_c}{d}.
    \end{align}

    The rest of the proof proceeds the same as in the proof of \cref{theo:simu}. Above, we showed that $t\le 2n^2N_c/d$ at any point in the algorithm. The algorithm terminates when $t=0$. Let $t_0$ denote the initial total number of firings. Let $k=\max_{v \neq s} \lfloor \frac{\sigma_v}{d} \rfloor$. Consider the current configuration $\sigma$. If $k=0$, the algorithm terminates. Otherwise, we have 
    \begin{align}
    k \geq \frac{N_c}{dn}-1.
    \end{align} When $N_c \geq 2dn$, combining above, we have 
    \[t' = t-k  \leq t(1-\frac{1}{2n^3}+\frac{1}{4n^3})\leq t(1-\frac{1}{4n^3}).\] Therefore, the algorithm terminates in $$
    \log_{1-1/4n^3}\frac{1}{t_0} + \frac{2n^2\cdot 2dn}{d}=
    O(n^3\log(t_0))
    =O(n^3\log(2n^2N/d))=O(n^3\log(nN))$$ iterations.
\end{proof}

\paragraph{Expander Graphs}

The notion of expanders arises frequently in many areas of math and CS theory. It has wide applications from constructing error-correcting codes \cite{sipser1996expander}, designing robust networks \cite{song2002expander} to serving as a tool to prove results in complexity theory \cite{ajtai1987deterministic} and number theory \cite{kowalski2019introduction}. Expanders are important and also interesting graph objects because they can be defined in many different languages: combinatorial, probabilistic, and algebraic. In particular, combinatorially speaking, expander graphs are graphs in which every small set of vertices has a (relatively) large boundary. The measure of expansion in an expander can be defined with respect to the number of edges or vertices on the boundary. We will stick with vertex expansion, which is more related to its probabilistic properties.

\begin{definition}
    An $\epsilon$-vertex-expander is a graph $G$ such that every vertex set $X\subset V(G)$ satisfying $|X|\le n/2$, has $|\neighbor(X)-X|\ge \epsilon |X|$.
\end{definition}

A result in \cite{chandra1989electrical} gives an upper bound on the maximum effective resistance of $\epsilon$-vertex-expanders.

\begin{lemma}[\cite{chandra1989electrical}]
    A connected $\epsilon$-vertex-expander $G$, with minimum degree $\delta$, has maximum effective resistance at most $24/(\epsilon^2(\delta+1))$.
\end{lemma}

Hence, we have the following corollary of \cref{theo:simu}.

\begin{corollary}\label{coro:simuExpander}
    Let $\epsilon$ be a constant. \cref{algorithm:simulation} terminates in $O(m^2\log (nN)/(\delta+1))$ iterations for $\epsilon$-vertex-expander $G$ with minimum degree $\delta$,  where $m$ denotes the number of edges.
\end{corollary}

Oftentimes, expanders are explicitly constructed with the additional property that is regular. For regular $\epsilon$-vertex-expander $G$, we have the following stronger result:

\begin{theorem}\label{theo:simuRegularExpander}
    Let $\epsilon$ be a constant. \cref{algorithm:simulation} terminates in $O(n^2\log n\log(nN))$ iterations for $d$-regular $\epsilon$-vertex-expander $G$.
\end{theorem}

\begin{proof} By \cite{rubinfeld1990cover}, the weight of each vertex is bounded by $cn\log n$ for $d$-regular $\epsilon$-vertex-expander, where $c$ is a constant. The rest part of the proof proceeds the same as the proof of \cref{theo:simuRegular}. We have $$
    t=\frac{1}{d}\sum_{i=0}^{t-1}(w({\sigma^i})-w({\sigma^{i+1}}))  =\frac{1}{d}(w(\sigma^0)-w(\sigma^t))\le \frac{cnN_c\log n}{d}.
    $$ Combining this with the fact that
    \begin{align}
        \max_{v \neq s} \lfloor \frac{\sigma_v}{\degree(v)} \rfloor \geq \frac{N_c}{dn}-1,
    \end{align} and truncate at $N_c = 2dn$. We have 
    \[t' = t-\max_{v\neq s} \lfloor \frac{\sigma_v}{\degree(v)} \rfloor  \leq t(1-\frac{1}{cn^2\log n}+\frac{1}{2cn^2\log n})\leq t(1-\frac{1}{2cn^2\log n}).\] Hence, let $t_0$ denote the firing number of the initial configuration, and the algorithm terminates in $$
    \log_{1-1/2cn^2\log n}\frac{1}{t_0}+\frac{cn\log n \cdot2dn }{d}=
    O(n^2\log n \log t_0)
    =O(n^2\log n\log(cnN\log n/d))=O(n^2\log n\log(nN))$$ iterations.
\end{proof}

\paragraph{Other Structured Graphs}

A $d$-dimensional hypercube is a graph defined on the vertex set $\{0,1\}^d$, in which two vertices are connected if and only if they are different in exactly one of the $d$ coordinates. We have the following result for hypercube graphs.

\begin{theorem}\label{theo:hyperCube}
    \cref{algorithm:simulation} terminates in $O(n^2\log(nN))$ iterations for $d$-dimensional hypercube $G$.
\end{theorem}

\begin{proof}
    Since $d$-dimensional hypercube has $2^d$ vertices, we know $d=\log n$. Note that the $d$-dimensional hypercube is $d$-regular, and the total weight goes down by $d$ when firing a vertex. By \cite{sauerwaldsun2011spectral}, the weight of each vertex is bounded by $cn$ for some constant $c$. Plug in these numbers to the framework of the rest proof of \cref{theo:simuRegular}, we get $$
    t'\le t(1-\frac{1}{2cn^2})
    $$ in each iteration after truncating at $N=2dn$. Hence, the number of iterations is in $$
    O(\log_{1-1/2cn^2}\frac{1}{t_0}+\frac{cn\cdot 2dn}{d})
    =O(n^2\log(t_0))=O(n^2\log (cnN/d))=O(n^2\log(nN)).
    $$
\end{proof}

In addition, we list the number of iterations of graphs having some other interesting properties. These bounds can be obtained by simply plugging in the number of edges in terms of the number of vertices in \cref{theo:simu}.

\begin{corollary}[of \cref{theo:simu}]\label{coro:simuBoundedDegree}
    \cref{algorithm:simulation} terminates in $O(\Delta^2n^3\log(nN))$ iterations for graph $G$ with maximum degree at most $\Delta$.
\end{corollary}

\begin{corollary}[of \cref{theo:simu}]\label{coro:simuPlanar}
    \cref{algorithm:simulation} terminates in $O(n^3\log(nN))$ iterations for planar graph $G$.
\end{corollary}



\subsection{Reduction Scheme by Vertex Removal}
\label{sec:reduction}

While the original sandpile instance does not contain any sink, we have to add sinks to the graph to decompose it into instances with smaller sizes or special structures. To begin with, we need to properly define the \textit{vertex removal} in our reduction scheme.

\begin{definition}[Vertex Removal]
\label{definition:vertex-removal}
Let $S(G, \sigma, M)$ be a sandpile instance with a non-empty set of sinks $M$. The instance obtained by removing a set of vertices $T \in V(G)$ is defined using the following procedure:

\begin{itemize}
    \item For every vertex $v \in T$ and each edge $(v, w)$ in the graph, create a new vertex $v'$ as a sink vertex, and add an edge $(v', w)$ into the graph.
    \item Remove the vertex $v$ and all the edges connecting vertex $v$ for all $v \in T$.
\end{itemize}

The graph we obtained after the removal of the vertices in $T$ is denoted as $G \setminus T$, and the instance we obtained is denoted as $S \setminus T$.
\end{definition}

 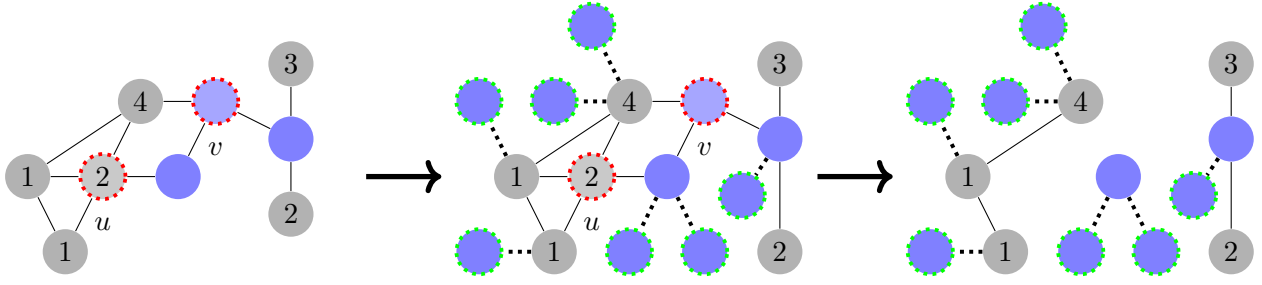
\begin{figure}[H]
    \centering
    \begin{tikzpicture}

 \tikzstyle{every node}=[fill=black!30,circle,minimum size=0.6cm,inner sep=1pt]
 \tikzset{removed/.style={fill=black!21,circle,minimum size=0.6cm,inner sep=1pt,line width=0.05cm,draw=red,dotted}}
 \tikzset{sink/.style={fill=blue!50,circle,minimum size=0.6cm,inner sep=1pt}}
 \tikzset{removed sink/.style={fill=blue!35,circle,minimum size=0.6cm,inner sep=1pt,line width=0.05cm,draw=red,dotted}}
 \tikzset{new sink/.style={fill=blue!50,circle,minimum size=0.6cm,inner sep=1pt,line width=0.05cm,draw=green,dotted}}

\begin{scope}[shift={(-3.5,0)}]
 
  \node[removed,label=below:{$u$}]    (1) at (0,0)            {2}; 
  \node             (2) at (-1,0)           {1}; 
  \node[sink]       (3) at (1,0)            {};
  \node             (4) at (-0.5,-1)        {1};
  \node             (5) at (0.5,1)          {4};
  \node[removed sink,label=below:{$v$}]       (7) at (1.5,1)  {};
  \node[sink]       (8) at (2.5,0.5)        {};
  \node             (9) at (2.5,1.5)        {3};
  \node             (10) at (2.5,-0.5)      {2};

  \foreach \from/\to in {1/2,1/3,1/4,2/5,1/5,2/4,5/7,7/8,8/9,8/10,3/7}
    \draw (\from) -- (\to);
\end{scope}

\begin{scope}[shift={(0,0)}]
\draw[line width=2pt,->] (0,0) -- (1,0);
\end{scope}

\begin{scope}[shift={(3,0)}]

  \node[removed,label=below:{$u$}]    (1) at (0,0)            {2}; 
  \node             (2) at (-1,0)           {1}; 
  \node[sink]       (3) at (1,0)            {};
  \node             (4) at (-0.5,-1)        {1};
  \node             (5) at (0.5,1)          {4};
  \node[removed sink,label=below:{$v$}]       (7) at (1.5,1)  {};
  \node[sink]       (8) at (2.5,0.5)        {};
  \node             (9) at (2.5,1.5)        {3};
  \node             (10) at (2.5,-1)        {2};

  \node[new sink]       (a) at (-1.5,1)     {};
  \node[new sink]       (b) at (-1.5,-1)    {};
  \node[new sink]       (c) at (-0.5,1)     {};
  \node[new sink]       (d) at (0.5,-1)     {};
  \node[new sink]       (e) at (1.5,-1)     {};
  \node[new sink]       (f) at (0,2)        {};
  \node[new sink]       (g) at (2,-0.25)    {};

  \foreach \from/\to in {1/2,1/3,1/4,1/5,2/5,2/4,5/7,7/8,8/9,8/10,3/7}
    \draw (\from) -- (\to);
  \foreach \from/\to in {2/a,4/b,5/c,3/d,3/e,5/f,8/g}
    \draw[line width=1.5pt,dotted] (\from) -- (\to);
\end{scope}

\begin{scope}[shift={(6,0)}]
\draw[line width=2pt,->] (0,0) -- (1,0);
\end{scope}

\begin{scope}[shift={(9,0)}]
  \node             (2) at (-1,0)           {1}; 
  \node[sink]       (3) at (1,0)            {};
  \node             (4) at (-0.5,-1)        {1};
  \node             (5) at (0.5,1)          {4};
  \node[sink]       (8) at (2.5,0.5)        {};
  \node             (9) at (2.5,1.5)        {3};
  \node             (10) at (2.5,-1)        {2};

  \node[new sink]       (a) at (-1.5,1)         {};
  \node[new sink]       (b) at (-1.5,-1)        {};
  \node[new sink]       (c) at (-0.5,1)         {};
  \node[new sink]       (d) at (0.5,-1)         {};
  \node[new sink]       (e) at (1.5,-1)         {};
  \node[new sink]       (f) at (0,2)            {};
  \node[new sink]       (g) at (2,-0.25)        {};

  \foreach \from/\to in {2/5,2/4,8/9,8/10}
    \draw (\from) -- (\to);
  \foreach \from/\to in {2/a,4/b,5/c,3/d,3/e,5/f,8/g}
    \draw[line width=1.5pt,dotted] (\from) -- (\to);
\end{scope}

\end{tikzpicture}
    \caption{A vertex removal with $T=\{u,v\}$. The blue vertices are sinks, and the gray vertices are normal vertices. }
    \label{figure:vertex-removal}
\end{figure}

We demonstrate this procedure in \Cref{figure:vertex-removal} by setting $T=\{u,v\}$. To remove $T$ (vertices with red dashes), we first add sinks (vertices with greed dashes) to all neighbors of $T$. Then, we remove $T$ and disconnect the graph into components, as shown on the right.



To remove vertices, we need to predict their firings on the graph and execute them first. Therefore, we need to figure out how to determine their firing numbers since it tells how many chips will be sent to their neighbors. On the other hand, if we already know how many times they will fire, we can safely ignore any chips sent from their neighbors, which is done by replacing their original positions with sinks. The formal reduction is given in \cref{theorem:general}. Removing any vertex needs a $O(\log n)$ factor multiplied to the total time complexity. In the following, we discuss the algorithmic details.

\begin{restatable}[Reduction by Vertex Removal]{theorem}{generalreduction}
\label{theorem:general}
Given a sandpile instance $S(G,\sigma)$ and a vertex set $P \subseteq V(G)$, let $\mathcal{G}$ be the set of connected components in $G \setminus P$. There is an algorithm that determines whether $S$ terminates and computes the terminal configuration of $S$ in $O\left(\log^{|P|} n \cdot \sum_{g \in \mathcal{G}} T(g)\right)$ time and $O\left(\sum_{g \in \mathcal{G}}M(g)\right)$ memory.
$T(g)$ and $M(g)$ denote the time and space complexity to solve a sandpile prediction on $G[V(g)\cup P]$ with $P$ being the set of sinks, respectively. The total number of chips in each subproblem is guaranteed to be at most $n^6+\lvert\lvert\sigma\rvert\vert_1$.
\end{restatable}

\subsubsection{Capturing Firing Number by Linear Inequalities} We first need to show that the relationships between graphs and firing numbers can be captured through a system of inequalities. Specifically, we prove that an integral feasible solution of this system, exhibiting the smallest partial order, corresponds to a vector constructed by the firing numbers of the vertices:

\begin{restatable}{lemma}{inequality}
\label{theorem:inequality}
Given any sandpile instance $S(G,\sigma)$, let $\bc\in \mathbb{N}^n$ be the firing number vector, in which $\bc(v)$ is the firing number of vertex $v$. Consider the following system of linear inequalities:

\begin{align}
    \left(\sum_{v\in \neighbor(u)}\bf(v)\right)-\bf(u)\cdot \degree(u)+\sigma_u<\degree(u), \forall v \in V(G) \label{linearsystem}.
\end{align}

Among every non-negative integer solution of \Cref{linearsystem}, $\bc$ is the one with the minimum partial order. If there is no feasible solution, the instance is recurrent.
\end{restatable}

The proof of \cref{theorem:inequality} can be found in \cref{sec:proofinprelim}. This is also known as the \textit{Least Action Principle)} \cite{klivans2018mathematics}.
We also give a version on sandpile with sinks in \cref{corollary:inequality}:

\begin{corollary}[Corollary of \Cref{theorem:inequality}]
\label{corollary:inequality}
Given any sandpile instance $S(G,\sigma, M)$ with the non-empty set of sinks $M$, let $\bc\in \mathbb{N}^{n-|M|}$ be the firing number vector, in which $\bc(v)$ is the firing number of vertex $v$, for $v\notin M$. Consider the following system of linear inequalities:

\begin{align}
    \left(\sum_{v\in \neighbor(u) \setminus M}\bf(v)\right)-\bf(u)\cdot \degree(u)+\sigma_u<\degree(u), \forall u \in V(G) \setminus M\label{formula:linearsystemwithsinks}.
\end{align}

Among every non-negative integer solution of \Cref{formula:linearsystemwithsinks}, $\bc$ is the one with the minimum partial order.
\end{corollary}

Our proposed tree algorithm can be viewed as a means of identifying the feasible solution with the smallest partial order of \Cref{linearsystem} on the given sandpile instance. This interpretation highlights the connection between the algorithm's execution and the problem's mathematical formulation.

\subsubsection{Independent Monotonicity of Firing Number} Furthermore, in \cref{linearsystem}, we show that for every $u \in V(G)$, there is a threshold value $\bp(u)$ such that if and only if $\bf(u) \geq \bp(u)$, the feasible solution exists. We formally state and prove such monotonicity in \cref{theorem:monofiring}. Another version on sandpile with sinks is given as \cref{corollary:monofiring}.

\begin{restatable}{lemma}{monofiring}
\label{theorem:monofiring}
Given any sandpile instance $S(G,\sigma)$ that terminates, for each vertex $u \in V(G)$, there exists a non-negative threshold value $\bp(u)$, such that for any non-negative integer $k$, if and only if $k\geq \bp(u)$, there exists a feasible solution $\bf$ satisfying \Cref{linearsystem} and $\bf(u)=k$. Moreover, $\bp$ is equal to the firing number vector $\bc$.
\end{restatable}

\begin{proof}
    For any integral non-negative feasible solution $\bf$ of \Cref{linearsystem}, if we let $\bf'(u)=\bf(u)+1$ for each $u \in V(G)$, $\bf'(u)$ is still a feasible solution. Since the firing number vector $\bc$ is also a feasible solution, for any $u \in V(G)$ and integer $k \geq \bc(u)$, we can construct a feasible solution $\bc'(v)=\bc(v)+(k-\bc(u), v\in V(G)$. By \Cref{theorem:inequality}, if $k<\bc(u)$, there is no feasible solution in which $\bq(u)=k$. Otherwise, it contradicts the assumption that $\bc$ takes the minimum value of all feasible solutions on each index $u \in V(G)$. Thus, we have proved our lemma.
\end{proof}

\begin{corollary}[Corollary of \Cref{theorem:monofiring}]
\label{corollary:monofiring}
Given any sandpile instance $S(G,\sigma,M)$ that terminates, for each vertex $u \in V(G)\setminus M$, there exists a non-negative threshold value $\bp(u)$, such that for any non-negative integer $k$, if and only if $k\geq \bp(u)$, there exists a feasible solution $(\bf(v))_{v\in V\setminus M}$ satisfying \Cref{formula:linearsystemwithsinks} and $\bf(u)=k$. Moreover, $(\bp(v))_{v \in V \setminus M}$ is equal to the firing number vector $\bc$ of $S$.
\end{corollary}

\subsubsection{Vertex Removal by Binary Search} With \Cref{corollary:inequality} and \Cref{corollary:monofiring}, we can remove a vertex at the cost of a $O(\log n)$-factor in the overall complexity. Here, we reduce the problem to the bounded sandpile prediction problem (\cref{prob:bounded}). It is a special prediction problem defined with two parameters, $L_1$ and $L_2$, indicating restrictions to the maximum numbers of firings and chips. If the firing number exceeds the limit, we terminate our algorithm and report with $\ovf$, where $\ovf$ is a special vector that exceeds the firing vector $\bc$ on each non-sink vertex. It signals that the restriction has been violated.

\begin{problem}[Bounded Sandpile Prediction with Sinks]
\label{prob:bounded}
For a given sandpile instance $S=(G,\sigma,M)$ and two parameters $L_1$ and $L_2$ such that $\lvert\lvert \sigma \rvert\rvert_1 \leq L_1$. The bounded sandpile prediction problem with sinks is to determine whether the firing vector $\bc$ is uniformly bounded by $L_2$. If yes, compute the terminal configuration; otherwise, return $\ovf$.
\end{problem}

In the following lemma, we state how to reduce the sandpile with sinks problem (\cref{pro:sink}) into several bounded subproblems (\cref{prob:bounded}) by removing a single vertex $p$.

\begin{restatable}{lemma}{generalsink}
\label{lemma:generalsink}
Given a sandpile instance $S(G,\sigma,M)$ and a vertex $p \in V(G)$, let $\mathcal{G}$ be the set of connected components in $G \setminus p$. There is an algorithm that solves the bounded sandpile prediction problem with sinks with parameters $L_1$ and $L_2$ in $O\left(\log L_1 \cdot \sum_{g \in \mathcal{G}} T(g)\right)$ time and $O\left(\sum_{g \in \mathcal{G}}M(g)\right)$ memory. $T(g)$ and $M(g)$ denote the time and space complexity to solve the bounded sandpile prediction problem with sinks on $g$ where $L_1'=L_1+\degree(p)\cdot L_2$ and $L_2'=L_2$, respectively.
\end{restatable}

The proof of \Cref{lemma:generalsink} can be found in \cref{sec:proofinprelim}.








\begin{remark}
Given two sandpile instances $S(G,\sigma,M)$ and $S(G,\sigma',M)$, we denote the corresponding firing number vector computed in the bounded sandpile prediction as $\bc$ and $\bc'$. If $\sigma \leq \sigma'$ pointwisely, we have $\bc \leq \bc'$ pointwisely.
\end{remark}


\subsubsection{Overall Analysis}
Now, we are ready to prove \Cref{theorem:general}.



\begin{proof}[Proof of \Cref{theorem:general}]


To begin with, we apply \Cref{theorem:monofiring} to an arbitrary vertex $u \in P$, utilizing a binary search to calculate its firing number. This approach reduces the problem to a bounded sandpile prediction with sinks. We conduct a search for the firing number $\bc(u)$ in the range $[0,n^4]$, and set $n^4$ as the $L_2$ bound for the remaining sandpile prediction problem with sinks. By \cite{tardos1988polynomial}, if we are unable to find a feasible value for $\bf(u)$ within this range, we can conclusively say the instance is recurrent.

To determine if $mid$ is legal, by \Cref{corollary:monofiring}, if $mid$ is at least $\bc(u)$, there should be a feasible solution where $\bf(u)=mid$. We apply $mid$ times of firings on $u$. Then we turn $u$ into a sink vertex and replace $\bf(u)$ with $mid$. In this way, we reduce the problem to a bounded prediction problem where $L_1=\lvert\lvert \sigma \rvert\rvert + \degree(u) \cdot mid$ and $L_2=n^4$. After computing the terminal configuration of this problem and its corresponding firing number vector $\bd$, by \Cref{corollary:inequality}, $\bd$ must be a feasible solution to the reduced problem with the smallest partial order. Therefore, if there is any feasible solution where $\bf(u)=mid$, $\bd$ should also satisfy the inequality of $u$. Thus, we determine $\{\bf(u)=mid\} \bigcup \{\bf(v)=\bd_v \mid v\in V(G),v\neq u\}$ is a feasible solution. In this way, we can continue the binary search by properly narrowing the range down.

Therefore, we reduce the problem with an extra cost of $O(\log n)$ runtime to a new bounded sandpile prediction with sinks with $L_1=o(n^6)$ and $L_2=n^4$. If we continue applying \Cref{lemma:generalsink} on another vertex in $P$, we pay another $O(\log L_1)=O(\log n)$ cost to reduce to a new instance where $L_2' \leftarrow L_2 + \degree \cdot L_1$ and $L_1' \leftarrow L_1$. One can observe that $L_2'$ stays in $o(n^6)$. We can prove our theorem by repeatedly applying \Cref{lemma:generalsink}. Since the binary search does not cost extra space, the space complexity is the same as the summation of all subproblems.
\end{proof}

\paragraph{Application by Decomposing into Trees} Combined with our tree algorithm shown in \cref{sec:sinktree}, we give the following corollary, which provides a more specific algorithmic result.

\begin{restatable}[Reduction to Trees with Sinks]{corollary}{reductiontree}
\label{coro:reductiontree}
Given a sandpile instance $S(G,\sigma)$ and a vertex set $P \subseteq V(G)$, let $\mathcal{G}$ be the set of connected components in $G \setminus P$. If for any $g \in \mathcal{G}$, $g$ is a tree and is adjacent to at most $3$ vertices in $P$, there is an algorithm that determines whether $S$ terminates and computes the terminal configuration of $S$ in $O\left(n \log^{|P|+1} n \right)$ time and $O\left(n\right)$ memory.
\end{restatable}

\begin{proof}
By \Cref{theorem:treesink}, we need $O(n \log n + \log \lvert\lvert \sigma \rvert\rvert _1 \log n)$ time to compute a sandpile with sinks problems on a tree with at most $3$ sinks. In the sandpile prediction model without sinks, $\lvert\lvert\sigma\rvert\rvert_1$ is not changed in the whole process. Thus, if $\lvert\lvert\sigma\rvert\rvert_1 \geq \sum_{v \in V(G)} \degree(v) = 2|E(G)|$, the instance does not terminate. After ruling out this case in the beginning, we can assume $\lvert\lvert \sigma \rvert\rvert = O(|E|) = O(n^2)$ while the algorithm proceeds in the instance. Thus, the time it consumed is $O(n \log n + \log^2 n) = O(n \log n)$. Note that for the bounded version, we need to check if $L_1$ is exceeded after computation. If so, we need to return $\ovf$ instead. Combine with \Cref{theorem:general}, the corollary follows.
\end{proof}

To demonstrate, we apply the theorem to solve the sandpile prediction problem on a special structured graph: the Pseudotree.

\begin{definition}[Pseudotree \cite{gabow1988linear}]
    Pseudotree is defined as an undirected connected graph that contains at most one cycle. Equivalently, it is an undirected connected graph in which the number of edges is at most the number of vertices.
\end{definition}

\begin{restatable}[Sandpile Prediction on a Pseudotree]{theorem}{theoremtreewithedge}    Given a sandpile instance $S(G,\sigma)$ in which $G$ is a pseudotree, there is an algorithm that determines whether $S$ terminates and compute the terminal configuration in $O(n \log^2 n)$ time and $O(n)$ memory. 
\end{restatable}
\begin{proof}
By definition, a pseudotree is either a tree or a tree with an extra edge. If it is a tree, we can apply \Cref{theorem:main} directly. For a tree with an extra edge, exactly one cycle exists in the graph. If we remove an arbitrary vertex $u$ on the cycle, the graph is reduced to trees. Therefore, we let $P=\{u\}$ and apply \Cref{coro:reductiontree}. This gives an algorithm that runs in $O(n \log^2 n)$ time, with $O(n)$ memory. 
\end{proof}

\begin{remark}
The time complexity can be improved to $O(n \log n)$ if the given graph is only a cycle of size $n$. Removing a vertex reduces the input graph to multiple path instances with sinks. We can modify algorithms in \Cref{sec:path} in a similar way.
\end{remark}


\section{Data Structure for Sandpiles on Trees}
\label{sec:ds}

In this section, we will introduce the data structure by proving \Cref{theorem:ds}.

\begin{restatable}[Data Structure Theorem]{theorem}{dstheorem}
\label{theorem:ds}
There exists a series of data structures $\mathcal{D}=\{D_u, u \in V(G)\}$ that satisfies the following:
\begin{itemize}
        \item All operations of $\mergeupward{}$, $\updatedsupward{}$, $\revertds$, $\splitds$, $\cnt{}$, $\deltasum{}$ and $\deltaquery{}$ are correctly called and produce correct results among the entire execution of \Cref{algorithm:main}.
        \item All operations of $\mergeupward{}$, $\updatedsupward{}$, $\revertds$, $\splitds$, $\cnt{}$, $\deltasum{}$ and $\deltaquery{}$ cost $O(n \log n)$ time in total among the entire execution of \Cref{algorithm:main}.
        \item $\mathcal{D}$ takes $O(n)$ memory in total at any moment among the entire execution of \Cref{algorithm:main}.
\end{itemize}

\end{restatable}

To describe how we maintain the data structure, we will first introduce the concept of \textit{key pairs}. 
\subsection{Overview}
\begin{definition}[Key Pairs]
\label{def:keypair}
A pair $(u, k)$ such that $u \in V(G)$ and $k \in \N_{+}$ is said to be a key pair if and only if $\delta(u, k) = \delta(u, k-1)$.
\end{definition}

By \Cref{lemma:delta-differs-at-most-one}, the value of $\delta(u,k) - \delta(u,k-1)$ will be either $0$ or $1$. If, for a given vertex $u \in V(G)$ and integer $k$, we can find the number of key pairs $(u,k')$  such that $k' \leq k$, denoted as $C$, then we can calculate the value of $\delta(u,k)$ which would be exactly $k-C$. Formally:

\begin{lemma}
\label{lemma:deltaequal}
Let $u$ be an arbitrary vertex $u \in V(G)$ and $k$ be a non-negative integer. Then 
\[
\delta(u,k) = k - \sum_{i=1}^{k}[\delta(u,i) = \delta(u,i-1)]
\]
\end{lemma}

\begin{proof}
We can prove the lemma by induction. 

First, the lemma is trivial for $k = 0$. For a positive integer $k \in \N_{+}$, we have $\delta(u,k-1) = (k-1) - \sum_{i=1}^{k-1} [\delta(u,i) = \delta(u,i-1)]$ by inductive hypothesis. Then $\delta(u,k) = \delta(u,k-1) + (1 - [\delta(u,k) = \delta(u,k-1)]$ since $\delta(u,k) - \delta(u,k-1)$ could be either $0$ or $1$. By substituting $\delta(u,k-1) = (k-1) - \sum_{i=1}^{k-1} [\delta(u,i) = \delta(u,i-1)]$, we have $\delta(u,k) = (k-1) - \sum_{i=1}^{k-1} [\delta(u,i) = \delta(u,i-1)] + 1 - [\delta(u,k) = \delta(u,k-1)] = k - \sum_{i=1}^k [\delta(u,i)=\delta(u,i-1)]$.
\end{proof}

We will use the splay tree $D_u$ to maintain all the pairs of $(u, k)$ for a fixed vertex $u \in V(G)$. 

Each node $x$ in the splay tree $D_u$ represents a key pair $(u, k)$. Let $\moment_x$ denote the value $k$ for the node $x$. For any two different nodes $x_1, x_2 \in D_u$, we define $x_1 < x_2$ if and only if $\moment_{x_1} < \moment_{x_2}$. This is obviously a well-defined partial order. 

Therefore, we design each $D_u (u \in V(G))$ to be a structure that maintains the orders of key pair nodes implemented by a splay tree. In addition, we maintain some arrays of length $n$ to store necessary information for algorithms as well as operations on $\mathcal{D}$. This part takes $O(n)$ memory to store. Each node $x$ on $D_u$ is mapped to the exact position on these arrays, which means the corresponding information can be accessed in $O(1)$ while accessing node $x$.

\begin{itemize}
    \item $\moment$: $\moment_x$ denotes the value of $k$ where $(u, k)$ is the key pair corresponding to the node $x$.
    \item $\timestamp$: $\timestamp_x$ denotes the value of $\dfn_u$, where the tree vertex $u$ satisfying node $x$ was first added to $D_u$.
    \item $\timemin$ and $\timemax$: $\timemin_x$ and $\timemax_x$ denotes the minimum and the maximum value of $\timestamp_{x'}$, where $x' \in \subtreeT(x)$. 
    \item $\ta,\tb$: two supportive integer arrays to store tags for the lazy propagation.
\end{itemize}

\begin{figure}[H]
    \centering
    \begin{tikzpicture}[->]

 \tikzstyle{special}=[rounded corners=3pt, draw, fill=black!0, align=center,minimum width=90pt, minimum height=35pt]
 
    \node[special] (2) at (3,0) {$\varnothing$};
    \node[special] (3) at (10,0) {$\{(v, k) \mid (v, k) \in D_v\}$\\$v \in \son(u)$};
    \node[special] (4) at (17,0) {$\{(u, k) \mid k \in \Z^{+}\}$\\$\delta(u,k)=\delta(u,k-1)$};
    
    \draw[->,line width=2pt,draw=blue, shorten <=5mm, shorten >=5mm] (2) to [bend left] node[midway, label = {[text=blue]above:{$\mergeupward(u, v)$}}]{} (3) ;
    \draw[->,line width=2pt,draw=blue,align=center, shorten <=5mm, shorten >=5mm] (3) to [bend left] node[midway, label = {[text=blue]above:{$\cnt(u)$\\$\mergeupward(u,v)$}}]{} (4) ;

    \draw[->,line width=2pt,draw=red, shorten <=5mm, shorten >=5mm] (4) to [bend left] node[midway, label = {[text=red]below:{$\revertds(u)$}}]{} (3) ;
    \draw[->,line width=2pt,draw=red, shorten <=5mm, shorten >=5mm] (3) to [bend left] node[midway, label = {[text=red]below:{$\splitds(u)$}}]{} (2) ;
    
\end{tikzpicture}
    \caption{This figure shows the life cycle of $D_u$ where arrows describe the calling order and blocks contain the current maintained information. Blue ones are called in $\upward(u,G,\sigma')$ and red ones are in $\downward(u, G)$.}
    \label{fig:enter-label}
\end{figure}
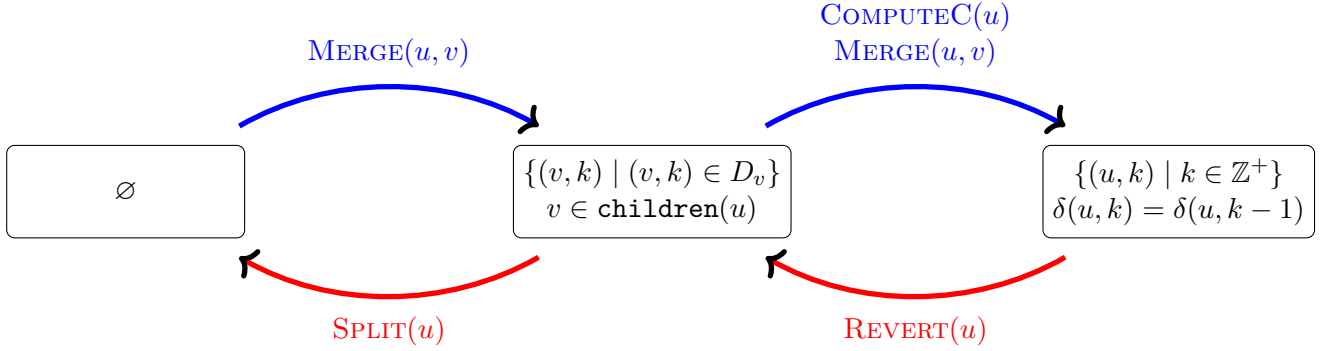

\subsection{Splay Trees}

A splay tree supports accessing, inserting and deleting a node in an amortized $O(\log n)$ time  \cite{sleator1985self}. We define the following basic interfaces for our splay tree:

\begin{itemize}
    \item $\newnode(k, t)$: Create a new node $x$. It will also create a  mapping to $\timestamp_x$, $\timemax_x$, $\timemin_x$, $\ta_x$ and $\tb_x$. Then it initializes
    \begin{itemize}
        \item $\moment_x \leftarrow k$
        \item $\timestamp_x \leftarrow t, \timemin_x \leftarrow t, \timemax_x \leftarrow t$
        \item $\ta_x \leftarrow 0, \tb_x \leftarrow 0$
    \end{itemize}
    \item $\Insert_{T}(x)$: Insert the node $x$ into the splay tree $T$.
    \item $\Delete_{T}(x)$: Remove the node $x$ from the splay tree $T$. Note that after removing the node $x$, we won't delete the information in $\timestamp_x, \timemax_x, \timemin_x, \ta_x$ and $\tb_x$. 
\end{itemize}

We give some notations on any splay tree $T$:
\begin{itemize}
    \item $\subtreeT(x)$ denotes the union of nodes where node $x$ is on the path to the splay tree root.
    \item $\size(x)$ denotes the number of nodes in $\subtreeT(x)$; Specially, $\size(T)$ denotes the number of nodes in the splay tree $T$.
    \item $x \in T$ denotes that $T$ contains the node $x$;
    \item For any node $x \in T$, $\father(x)$, $\leftson(x)$ and $\rightson(x)$ denotes the father, left child and right child of node $x$ on $T$ respectively. If not exists, the value will be $\nil$ by default. Furthermore, define $\sonT(x)$ as the set of the children for the node $x$. Note that $\nil$ is not considered as an element of $\sonT(x)$. 
    \item $\roott(T)$ denotes the root node of $T$. If $T$ is empty, $\roott(T)$ will be $\nil$.
    \item $\splay(x)$ denotes the operation to make node $x$ as the root by a series of rotations. 
    \item $\rank_{T}(x)$ denotes the rank of the node $x$ in the splay tree $T$. Here, the rank of a node is defined as the number of nodes that precede it while performing an in-order walk.
    \item $\pre_{T}(x)$ and $\suc_{T}(x)$ denotes the predecessor and successor of the node $x$, respectively. Formally, if $\rank_{T}(x) = k$, then $\pre_{T}(x)$ is the node of the rank $k - 1$ and $\suc_{T}(x)$ is the node of the rank $k+1$. Specially, if there's no such node, then the predecessor (or the successor) of the node $x$ will be considered as $\nil$.
\end{itemize}

On a splay tree, we define $\findmin(u)$ as a subroutine to find the node $x$ with the minimum rank on $D_u$. This subroutine is described in \Cref{algorithm:findmin}. It is known that such a process runs in $O(\log n)$ amortized time.

\IncMargin{1em}
\begin{algorithm}[H]
  \SetKwData{Left}{left}\SetKwData{This}{this}\SetKwData{Up}{up}

  \SetKwProg{Fn}{Function}{:}{}
  \BlankLine

  $x\leftarrow \roott(D_u)$\;
  
  \While{$\leftson(x)\neq \nil$}{
    $x\leftarrow \leftson(x)$\;
  }

  $\splay(x)$\;
  \Return $x$;
  
  \caption{\findmin($u$)}\label{algorithm:findmin}
\end{algorithm}\DecMargin{1em}

During the splay tree maintenance in our algorithm, we need to perform the following modification to update the information:

\begin{itemize}
    \item For a given node $x \in D_u$ and two parameters $a$ and $b$. Increase the value of $\moment_y$ by $\rank_{\subtreeT(x)}(y) \cdot a + b$ for all $y \in \subtreeT(x)$.
\end{itemize}

However, it is not efficient to perform such a change for a whole $\subtreeT(x)$ every time. Here we will use the classic \textit{lazy propagation} trick. In general, we defer the modification on the node to the time we actually visit it. Since we always visit $\father(x)$ before visiting any node $x$, we execute the modification on $x$ whenever we visit $\father(x)$, clearing the \textit{lazy tag} on $\father(x)$ afterward. Here we use \inctime($x$, $a$, $b$) to denote a modification for node $x$ with two parameter $a$ and $b$. Specifically, for each node $x \in D_u$, we maintain two lazy tags $\ta_{x}$ and $\tb_{x}$ indicating "it should perform a \inctime($y$, $\ta_x$, $\tb_x$) operation for any child $y \in \sonT(x)$ of the node $x$". 

\IncMargin{1em}
\begin{algorithm}
  \SetKwData{Left}{left}\SetKwData{This}{this}\SetKwData{Up}{up}
  \SetKwComment{Comment}{$\triangleright$\ }{}
  \SetKwFunction{SolveUpward}{\upward}
  \SetKwFunction{SolveDownward}{\downward}
  \SetKwFunction{CalculateNumberOfFirings}{\cnt}
  \SetKwFunction{UpdateDataStructure}{\updatedsupward}
  \SetKwFunction{MergeDataStructure}{\mergeupward}
  \SetKwFunction{DeltaSum}{\deltasum}
  \SetKwFunction{InitializeDataStructure}{\initds}
  \SetKwFunction{DeltaQuery}{\deltaquery}
  \SetKwFunction{RevertDataStructure}{\revertds}
  \SetKwFunction{SplitDataStructure}{\splitds}
  \SetKwInOut{Input}{input}\SetKwInOut{Output}{output}  
  \SetKwFunction{FSolveDownward}{SolveDownward}
  \SetKwFunction{IncTime}{IncTime}

  \SetKwProg{Fn}{Function}{:}{}

  $\moment_x \leftarrow \moment_x + \left(\size(\leftson(x)) + 1\right) \cdot a + b$\;
  $\ta_x \leftarrow \ta_x + a$\;
  $\tb_x \leftarrow \tb_x + b$\;
  \BlankLine

  \caption{\inctime($x$, $a$, $b$)}\label{algorithm:inctime}
\end{algorithm}\DecMargin{1em}

\begin{remark}
The reason we can use lazy propagation here is because the \inctime($x$, $a$, $b$) operation follows the associative law.
\end{remark}

Note that although the splay tree will change forms, as long as we push down lazy tags before changing, the correctness is guaranteed. We discuss the details of external information maintenance while node rotations occur based on this in \Cref{sec:dsrotatemaintain}.


Furthermore, the \textit{Dynamic Finger Theorem}, described in \Cref{theorem:dynamic-finger-theorem}, gives us a better bound for a sequence of the access operations. This theorem is vital to our time complexity analysis.

\begin{theorem}[Dynamic Finger Theorem \cite{cole2000dynamic,cole2000dynamic2}]
\label{theorem:dynamic-finger-theorem}
Let $T$ be a splay tree with $n$ nodes. Consider a sequence of accesses in splay tree $T$ (denoted as $a_1, a_2, \cdots, a_m$ and assume $a_0$ is the root of the splay). Then the cost of the access sequence is bounded by 
\[O\left(m + n + \sum_{i=1}^m \log(1+d(a_{j-1}, a_j))\right),\] where $d(x, y)$ denotes the difference between the ranks of the node $x$ and $y$.
\end{theorem}

With such dynamic finger property, when merging multiple splay trees with the small-to-large trick, we can reach a total time complexity of $O(n \log n)$ where $n$ denotes the total number of nodes. Conversely, we can split the final splay trees back in the same complexity, which can be regarded as undoing the merging. In our tree algorithm, we design two interfaces $\mergeupward$ and $\splitds$ to support merging and splitting splay trees. Specifically, we have the following two lemmas:

\begin{restatable}{lemma}{mergelemma}
\label{lemma:merge}
    \mergeupward($u$, $v$) will merge all nodes from $D_v$ into $D_u$. Note that there won't be nodes in $D_v$ after merging. During the execution of \Cref{algorithm:main}, all $\mergeupward$ operations take $O(n \log n)$ time in total.
\end{restatable}

\begin{restatable}{lemma}{splitlemma}
   \label{lemma:split}
    If the current $D_u$ contains all key pairs from $D_u$ and $D_v$ before calling $\mergeupward(u,v)$ and no key pair from $D_{v'}$ exists if $v'$ is after $v$ in $\mathcal{I}$, \splitds($u$, $v$) will extract nodes to $D_v$ from $D_u$, reverting $D_u,D_v$ from the corresponding call of $\mergeupward(u,v)$. After $\splitds(u,v)$, no key pair from $D_{v'}$ exists if $v'$ is no earlier than $v$ in $\mathcal{I}$. During the execution of \Cref{algorithm:main}, all $\splitds$ operations take $O(n \log n)$ time in total. 
\end{restatable}

The implementation and analysis of $\mergeupward$ (\Cref{lemma:merge}) and $\splitds$ (\Cref{lemma:split}) are given in \Cref{sec:mergeds} and \Cref{sec:splitds} respectively.

\subsection{Difference Aggregation by Tree Walk}
\label{sec:ds:tree-walk}
We will analyze $\deltaquery$ (\Cref{algorithm:deltaquery}) by proving the following lemma.
\begin{lemma}
\label{lemma:deltaquery}
$\deltaquery(u,k)$ will return the correct value of $\delta(u, k)$. During the execution of \Cref{algorithm:main}, the $\deltaquery$ operation takes $O(n \log n)$ time in total. 
\end{lemma}

\IncMargin{1em}
\begin{algorithm}[H]
  \SetKwData{Left}{left}\SetKwData{This}{this}\SetKwData{Up}{up}
  \SetKwFunction{PushDown}{PushDown}

  \SetKwProg{Fn}{Function}{:}{}
  \BlankLine
  $now\leftarrow \roott(D_u)$\; \label{deltaquery:takeroot}
  $las\leftarrow \nil$\; \label{deltaquery:lasnil}
  $rank\leftarrow 0$\; \label{deltaquery:rankinit}

  \While{$now\neq \nil$}{ \label{deltaquery:while}
    $\PushDown(now)$\; \label{deltaquery:pushdown}
    $las\leftarrow now$\; \label{deltaquery:lasupdate}
    \If{$\moment_{now}\le k$}{ \label{deltaquery:if}
      $rank\leftarrow rank +\size(\leftson(now))+1$\; \label{deltaquery:add}
      $now\leftarrow \rightson(now)$\; \label{deltaquery:right}
    }
    \Else{
      $now\leftarrow \leftson(now)$\; \label{deltaquery:left}
    }
  }

  $\splay(las)$\; \label{deltaquery:splay}

  \Return $k-rank$\; \label{deltaquery:return}
  
  \caption{\deltaquery($u$, $k$)}\label{algorithm:deltaquery}
\end{algorithm}\DecMargin{1em}

\begin{proof}
First, by \Cref{lemma:deltaequal}, we have $\delta(u,k)=k-\left\lvert\left\{x\mid x\in D_u,\moment_x\leq k\right\}\right\rvert$.

Since we use a splay tree to maintain all these ordered key pairs, to count the number of key pairs for a certain prefix, we will conduct a top-down tree walk on the splay tree starting from the root. This process is described in \Cref{algorithm:deltaquery}. We use $now$ to represent the current visiting node, initialized as $\roott(D_u)$ (\Cref{deltaquery:takeroot}) since we begin from the root. Additionally, we use $las$ to store the previously visited node (updated in \Cref{deltaquery:lasupdate}). We also initialize $rank$ to be $0$ (\Cref{deltaquery:rankinit}), denoting the number of key pairs we should count.

During the loop (\Cref{deltaquery:while}), we are finding node $now$ with maximum rank such that $\moment_{now} \leq k$ while counting the number of key pairs. Specifically, if the current $\moment_{now} \leq k$ (\Cref{deltaquery:if}), by the property of the binary search tree, all key pairs in $\subtreeT(\leftson(now))$ together with $now$ satisfy the condition. Therefore, we directly increase $rank$ by $\size(\leftson(now))+1$ (\Cref{deltaquery:add}), then we go to $\rightson(now)$ for continue searching. Similarly, if $\moment_{now} > k$, all key pairs in $\subtreeT(\rightson(now))$ should be ignored. Thus we should go to $\leftson(now)$. Note that during the loop, whenever we visit a new node, we need to call $\PushDown(now)$ to guarantee the correctness of $\moment_{now}$.

After finding the pair with maximum rank, we also successfully count the number of key pairs $(u,k')$ satisfying the condition $k'\leq k$. Before returning the value computed by the formula (\Cref{deltaquery:return}), we also need to splay the last accessed node to the root by calling $\splay(las)$ (\Cref{deltaquery:splay}).

Since such a walk process is equivalent to the access on the splay tree, it has a $O(\log n)$ amortized cost and $\deltaquery$ will be called $O(n)$ times in \Cref{algorithm:downward}, the total time cost is $O(n \log n)$.
\end{proof}

\subsection{Computing Partial Firing Numbers by Pop-Up Mechanism}
We will analyze the $\cnt$ operation(\Cref{algorithm:calculatenum}) by proving the following lemma.
\begin{lemma}
\label{lemma:computeC}
    When calling $\cnt(u,\sigma'_u)$, if the current $D_u$ is the collection of all key pairs from $D_v,v\in \son(u)$, $\cnt(u,\sigma'_u)$ will compute the correct $\bc^{\down}(u)$. During the execution of \Cref{algorithm:main}, all $\cnt$ operations take $O(n \log n)$ time in total.
\end{lemma}

\IncMargin{1em}
\begin{algorithm}[H]
  \SetKwData{Left}{left}\SetKwData{This}{this}\SetKwData{Up}{up}
  \SetKwFunction{FindMin}{FindMin}
    \SetKw{Break}{break}

  \SetKwProg{Fn}{Function}{:}{}
  \BlankLine
  $now\leftarrow 0$\; \label{initk}
  $count \leftarrow 0$\; \label{computeC:initcount}
  $\bq_u \leftarrow \varnothing$ \; \label{computeC:initq}
  \While{$D_u\neq \nil$}{ \label{computeC:loop}
    $x\leftarrow \FindMin(D_u)$\; \label{computeC:getmin}
    \If{$\moment_x=now \textbf{ or } count+[u\neq r]\cdot (\moment_x-1-now)\leq \sigma'_u-\degree(u)$}{ \label{computeC:if}
      \Delete$(D_u,x)$\; \label{computeC:delete}
      $\bq_u\text{.append}(x)$\; \label{computeC:append}
      $count\leftarrow count+1+[u\neq r]\cdot (\moment_x-now)$\; \label{computeC:accumulate}
      $now\leftarrow \moment_x$\; \label{computeC:increasek}
    }
    \Else{\Break\;} \label{computeC:break}
  }
  $p\leftarrow [u\neq r]\cdot \max\left(0,\sigma'_u-count-(\degree(u)-1)\right)$\; \label{computeC:cornercase}
  \Return $now+p$\; \label{computeC:return}
  \caption{\cnt($u$,$\sigma'_u$)}\label{algorithm:calculatenum}
\end{algorithm}\DecMargin{1em}

\begin{proof}
By \Cref{lemma:c-down}, we know that the value of $\bc^{\down}(u)$ is exactly the non-negative smallest integer $k$ such that $\psi_u(k) < \degree(u)$. Recall that 
\[
\psi_u(k) = \sigma_u - k \cdot \degree(u) + \sum_{v \in \son(u)} \delta(v, k).
\]

Then, we have

\begin{align}
\psi_u(k) = \sigma_u  - \left(k\cdot |\son(u)|-\left( \sum_{v \in \son(u)} \delta(v, k)\right) \right) - k\cdot [u\neq r].\label{formula:dquery1}
\end{align}


Because we assume that $\cnt$ is called when $D_u$ is updated to the collection of all key pairs in $\son(u)$. This should eliminate the summation sign with regard to children in \Cref{formula:dquery1} via introducing the formula below. That is, by \Cref{lemma:deltaequal}, we have
\begin{align} 
    \sum_{v\in \son(u) } \delta(v,k) = k \cdot |\son(u)| - \left\lvert\left\{x \mid x \in D_u, \moment_x \le k\right\}\right\rvert. \label{formula:deltaalter}
\end{align}

By substituting the term in \Cref{formula:dquery1} we have
\begin{align}
\psi_u(k) = \sigma_u  - \left\lvert\left\{x \in D_u \mid \moment_x \le k\right\}\right\rvert - k\cdot [u\neq r].
\end{align}

Now we will prove that \Cref{algorithm:calculatenum} is sufficient to find the minimum $k$. We use $now$ to denote the current value of $k$. Since in $D_u$ the nodes are ordered by the value of $\moment$ in increasing order, we can repeatedly extract the node with the smallest rank from $D_u$ and see if it could increase the value of $now$ while $\psi_u(now) \geq \degree(u)$ holds.

In the beginning, we initialize $now$ as $0$ since $k$ is non-negative (\Cref{initk}). As $now$ increases, we use the variable $count$ to keep track of $\left\lvert\left\{x \mid x \in D_u, \moment_x \le now\right\}\right\rvert + now\cdot [u\neq r]$. Since we extract the minimum rank node repeatedly (\Cref{computeC:loop}), we need to delete it from $D_u$ (\Cref{computeC:delete}) to prevent redundant enumeration. We use an additional list $\bq_u$ for each $u$ to store the deleted node temporarily (\Cref{computeC:append}). The list $\bq_u$ is initialized as empty (\Cref{computeC:initq}) in the beginning. 

In every turn, we first extract the minimum rank node $x$ (\Cref{computeC:getmin}). There are two cases to consider (\Cref{computeC:if}):
\begin{itemize}
    \item If $\moment_x=now$, since the $\moment_x\leq now$ in the definition of $count$ is hold for the $x$ now, we need to increase $count$ by $1$.
    \item If $\moment_x \neq now$, we first test if $now$ can be increased to $\moment_x-1$ without triggering the terminate condition $\psi_u(now) < \degree(u)$. Since there is no $y \in D_u$ such that $now<\moment_y\leq \moment_x-1$, by definition, $\left\lvert\left\{x \mid x \in D_u, \moment_x \le k\right\}\right\rvert + k\cdot [u\neq r]$ equals to $count+[u\neq r] \cdot (\moment_x-1-k)$. If such value $\leq \sigma'_u - \degree(u)$, then the desired $k$ is at least larger than $\moment_x-1$, which allows us to increase $now$ to $\moment_x$.
\end{itemize}

For both cases, $now$ will be increased (or kept) to $\moment_x$ (\Cref{computeC:increasek}). We update $count$ correspondingly (\Cref{computeC:accumulate}) and continue enumeration.

When the loop terminates, it is known that the term $\left\lvert\left\{x \mid x \in D_u, \moment_x \le k\right\}\right\rvert$ has already been accumulated correctly. Now if we increase $now$ by $p$, the increment of $count$ will be $p \cdot [u \neq r]$.  Assuming $u \neq r$, we can compute the maximum possible $p$ to make $\psi_u(now) < \degree(u)$ by simple calculation, which is the difference between $\sigma'_u-count$ and $\degree(u)-1$. We take the max value between this difference and $0$ to avoid corner cases (\Cref{computeC:cornercase}). The desired value $k$ will be $now+p$ (\Cref{computeC:return}).

Since for each node $x$, it will only be deleted from any $D_u$ and inserted into the corresponding $\bq_u$ once. By \Cref{lemma:dsmemory}, there are $O(n)$ nodes in total. Since each deletion in $D_u$ costs $O(\log n)$ in amortized, the total time cost will be $O(n \log n)$.
\end{proof}

\subsubsection{\texorpdfstring{$\deltasum$}{DeltaSum} Calculation}
We will analyze the $\deltasum$ (\Cref{algorithm:deltasum}) operation by proving \Cref{lemma:deltasum}.

\IncMargin{1em}
\begin{algorithm}[H]
  \SetKwData{Left}{left}\SetKwData{This}{this}\SetKwData{Up}{up}
  \SetKwFunction{PushDown}{PushDown}

  \SetKwProg{Fn}{Function}{:}{}
  \BlankLine
  \Return $\bc^{\down}(u) \cdot |\son(u)| - |\bq|$\; \label{deltasum:return}
  
  \caption{\deltasum($u$)}\label{algorithm:deltasum}
\end{algorithm}\DecMargin{1em}

\begin{lemma}
\label{lemma:deltasum}
$\deltasum(u)$ will return the correct value of $\sum_{v \in \son(u)} \delta(v, \bc^{\downarrow}(u))$ which equals to $\bc^{\downarrow}(u)\cdot |\son(u)|-\size (Q_u)$ in the \Cref{algorithm:upward}. Each $\deltasum$ operation takes $O(1)$ time. 
\end{lemma}

\begin{proof}
First, by \Cref{formula:deltaalter}, we have $\sum_{v \in \son(u)} \delta(v, k)=k\cdot |\son(u)|-\left\lvert\left\{x\mid x\in D_u,\moment_x\leq k\right\}\right\rvert$ when $D_u$ is exactly the union of $D_v,v\in \son(u)$. After we merge all $D_v, v \in \son(u)$ to $D_u$, we execute the $\cnt$ before calling $\deltasum$. The term $\left\lvert\left\{x\mid x\in D_u,\moment_x\leq \bc^{\downarrow}(x)\right\}\right\rvert$ should be the number of node $x$ such that $moment_x\le \bc^{\downarrow}(x)$ in the original $D_u$ before $\cnt$ is called. By the proof of \Cref{lemma:computeC}, we know that $\cnt$ splits out all these nodes $x$s and stores in $\bq_u$ temporarily. Therefore, we only need to return $k \cdot |\son(u)| - |\bq_u|$ by the definition, which costs $O(1)$ calculation.
\end{proof}

\subsection{Moment Updating and Reverting}
\label{sec:updaterevert}

We will analyze the $\updatedsupward$(\Cref{algorithm:update}) and $\revertds$(\Cref{algorithm:revert}) operation by proving \Cref{lemma:update} and \Cref{lemma:revert}. First of all, let's assume $u$ is an arbitrary vertex other than $r$. We will prove the following lemma for any non-root vertex $u$. In the last part of this section, we will prove the correctness of the root vertex.


\begin{lemma}
\label{lemma:update} 
    For any vertex $u \in V(G)$, if the current $D_u$ is the union of all key pairs $(v,k), v \in \son(u), k > \bc^{\down}(u)$, $\updatedsupward(u)$ will update $D_u$ correctly that it contains all key pairs $(u,k)$.
    During the execution of \Cref{algorithm:main}, all $\updatedsupward$ operations take $O(n \log n)$ time in total.
\end{lemma}

\IncMargin{1em}
\begin{algorithm}[H]
  \SetKwData{Left}{left}\SetKwData{This}{this}\SetKwData{Up}{up}  \SetKwComment{Comment}{$\triangleright$\ }{}
  \SetKwFunction{Swap}{swap}
  \SetKwFunction{insert}{insert}
  \SetKwFunction{delete}{delete}
  \SetKwFunction{IncTime}{IncTime}
  \SetKwFunction{DecTime}{DecTime}
  \SetKwFunction{UpdateDataStructure}{UpdateDataStructure}
  \SetKwFunction{NewNode}{NewNode}

  \SetKwProg{Fn}{Function}{:}{}
  \BlankLine
  
  $\IncTime (\roott(D_u), 0, -\bc^{\downarrow}(u))$\; \label{update:oper1}
  $\num_u\leftarrow \degree(u)-1-\sigma'_u$\; \label{update:oper2}
  \For{$i$ in $[1,\num_u]$}{ \label{update:oper4}
    \insert($D_u$, \NewNode($0$, $\dfn_u$))\; \label{update:oper5}
  }
  \IncTime($\roott(D_u)$, $1$, $0$)\; \label{update:oper7}
  \caption{\updatedsupward($u$)}\label{algorithm:update}
\end{algorithm}\DecMargin{1em}

To analyze $\updatedsupward$, we first observe the current status of $D_u$: It contains all nodes $x$ such that $(v,\moment_x) \in D_v, v \in \son(u), \moment_x > \bc^{\down}(u)$ and they are all sorted by $\moment_x$. Calling $\updatedsupward$ will modify the information represented by these nodes so that $D_u$ will contain all key pairs of the current vertex $u$. We will show that this can be done without reordering the nodes.


By \Cref{def:keypair}, $D_u$ should contain all key pairs $(u,k)$ such that $\delta(u,k)=\delta(u,k-1)$. We give the following lemma to verify any positive integer $k$.

\begin{lemma}
\label{lemma:keypaircondition}
For any positive integer $k$, $(u,k)$ is a key pair of $u$ if and only if
\begin{align}
b-\left\lvert\left\{x\in D_u\mid \moment_x\leq \bc^{\down}(u)+(k-1-b)\right\}\right\rvert  < \degree(u)-1-\sigma'_u
\label{formula:keypaircondition}
\end{align} where

\begin{align}
    b=k-1-\delta(u,k-1).
    \label{formula:b}
\end{align}.
\end{lemma}
\begin{proof}
By the definition of $\delta$ (\Cref{def:delta}), $\delta(u,k)=\delta(u,k-1)$ if and only if vertex $u$ will not be full after proceeding the following process:

\begin{itemize}
\item Add $k-1$ chips to the vertex $u$.
\item Firing the full vertices in $\subtree(u)$ until the configuration is local terminal in $\subtree(u)$.
\item Add one more chip to the vertex $u$.
\end{itemize}

Therefore, by the definition of firing, we have the following inequality:

\begin{align}
\presigma_u+(k-1)-\degree(u) \cdot \left(\bc^{\downarrow}(u)+\delta(u,k-1)\right) + \sum_{v \in \son(u)} \delta(v,\bc^{\downarrow}(u)+\delta(u,k-1)) <\degree(u)-1 \label{formula:stablize}
\end{align}

On the left-hand side, $\presigma_u$ denotes the number of chips left on $u$ after the configuration becomes local terminal in all $\subtree(v),v\in \son(u)$. Then if we add $k-1$ chips on $u$ (\Cref{def:delta}), all the $\presigma_u+(k-1)$ chips will cause $\bc^{\downarrow}(u)+\delta(u,k-1)$ firings on $u$ in total. In such scenario, every $v\in \son(u)$ will receive $\bc^{\downarrow}(u)+\delta(u,k-1)$ chips and then return $\delta(v,\bc^{\downarrow}(u)+\delta(u,k-1))$ (\Cref{def:delta}). Since adding one more chip will not cause a new firing, the left-hand side should be less than $\degree(u)-1$. 

So the left-hand side of \Cref{formula:stablize} is equal to 
\begin{align}
\presigma_u+(k-1) - \left(\bc^{\down}(u)+\delta(u,k-1)\right) - \sum_{v \in \son(u)} \left(\left(\bc^{\downarrow}(u)+\delta(u,k-1)\right)-\delta(v,\bc^{\downarrow}(u)+\delta(u,k-1))\right) \label{formula:stablize2}
\end{align}

Since $b=k-1-\delta(u,k-1)$, one can observe that $b$ is the number of nodes $x$ satisfying $\moment_x \leq k-1$ in the final $D_u$. Now we can write \Cref{formula:stablize2} as
\begin{align}
\presigma_u+b-\bc^{\down}(u)-\sum_{v \in \son(u)} \left(\left(\bc^{\downarrow}(u)+(k-1-b)\right)-\delta(v,\bc^{\downarrow}(u)+(k-1-b))\right) \\
\end{align}

By \Cref{lemma:deltaequal}, this is equal to 

\begin{align}
\presigma_u+b-\bc^{\down}(u)-\sum_{v\in \son(u)} \sum_{i=1}^{\bc^{\downarrow}(u)+(k-1-b)}[\delta(v,i)=\delta(v,i-1)] \label{formula:stablize3}
\end{align}

Assuming we have a data structure $T$ storing the union of all key pairs $(v,k), v \in \son(u)$. This is equivalent to $D_u$ before executing $\cnt(u)$. We can further transform \Cref{formula:stablize3} to 
\begin{align}
\presigma_u+b-\bc^{\down}(u)-\left\lvert\left\{x\in T\mid \moment_x\leq \bc^{\down}(u)+(k-1-b)\right\}\right\rvert \label{formula:stablize5}
\end{align}

Noticed that we have $D_u\subset T$ and

\begin{align}
\left\lvert\left\{x\in T\backslash D_u\mid \moment_x\leq \bc^{\down}(u)+(k-1-b)\right\}\right\rvert=|\bq_u|
\end{align}
also
\begin{align}
\presigma_u-\bc^{\down}(x)-|\bq_u|=\presigma_u+\deltasum(u)-\degree(u)\cdot \bc^{\down}(x)=\sigma'_u
\end{align}

Thus \Cref{formula:stablize5} is equivalent to

\begin{align}
\sigma'_u+b-\left\lvert\left\{x\in D_u\mid \moment_x\leq \bc^{\down}(u)+(k-1-b)\right\}\right\rvert \label{formula:stablize4}
\end{align}

In all, $(u,k)$ is a key pair if, and only if, we have
\begin{align}
b-\left\lvert\left\{x\in D_u\mid \moment_x\leq \bc^{\down}(u)+(k-1-b)\right\}\right\rvert  < \degree(u)-1-\sigma'_u
\end{align}

\end{proof}

\begin{lemma}
\label{lemma:updatea}
For any $k \in [1,\degree(u)-1-\sigma'_u]$, the inequality \Cref{formula:keypaircondition} is satisfied. Therefore, $(u,k)$ is a key pair of $u$ by \Cref{lemma:keypaircondition}.
\end{lemma}

\begin{proof}
We can prove it by doing the mathematical induction on $k$. Assuming the legal $k$ forms a prefix $[1,p-1],p \leq \degree(u)-1-\sigma'_u$, we have $\delta(u,p-1)=0$ and thus we have $b=p-1$. So, we have $p-1-b=0$. Because in the current $D_u$, there is no node $x$ such that $\moment_x \leq \bc^{\down}(u)$. Therefore, the left-hand side of \Cref{formula:keypaircondition} has only $b$ left. Since $b=p-1 < \degree(u)-1-\sigma'_u$, the inequality holds for $p$. 
\end{proof}

\begin{lemma}
\label{lemma:updateb}
Let $y$ be the node in $D_u$ such that $\rank_{D_u}(y) = b-(\degree(u)-1-\sigma'_u)$, and $x$ be the node in $D'_u$ such that $\rank_{D'_u}(x) = b$, where $D'_u$ is the final $D_u$ that stores all key pairs $(u,k)$. Then the equation $\moment_x=\moment_y-\bc^{\down}(u)+b$ holds for all $b> \degree(u)-1-\sigma'_u$.
\end{lemma}

\begin{proof}
Firstly, we assume that the first $b\geq \degree(u)-1-\sigma'_u$ key pairs are determined for $D'_u$. This is because the first $\degree(u)-1-\sigma'_u$ key pairs are determined in \Cref{lemma:updatea}. 

Now we will find the next key pair $(u,k)$ of rank $b+1$, which is the minimum $k$ satisfying \Cref{lemma:keypaircondition}. Note that $b$ here denotes the number of key pairs $(u,k')$ where $k'\leq k-1$, can be rewritten as $k-1-\delta(u,k-1)$ (\Cref{lemma:deltaequal}) which happens to be the same $b$ as in \Cref{formula:b}. Therefore, we can use this $b$ to test the correctness by \Cref{lemma:keypaircondition}. Note that, this works only when $k$ is less or equal to the minimum possible one. That is equivalent to say, if there are multiple $k$s satisfying \Cref{lemma:keypaircondition} concerning $b$, we should only take the minimum one.

Now we claim that $k$ is equal to $\moment_y-\bc^{\down}(u)+b+1$ where $y$ is the $b+1-(\degree(u)-1-\sigma'_u)$-th key pair in current $D_u$.

To prove this, we first substitute this value into $\bc^{\down}(u)+(k-1-b)$, the term 
\begin{align}
\left\lvert\left\{x\in D_u\mid \moment_x\leq \bc^{\down}(u)+(k-1-b)\right\}\right\rvert \label{formula:momentyterm}
\end{align} becomes 
\[
\left\lvert\left\{x\in D_u\mid \moment_x\leq \moment_y \right\}\right\rvert,
\] which is no less than the rank of $y$, as $b+1-(\degree(u)-1-\sigma'_u)$.

Therefore, the left-hand side of \Cref{lemma:keypaircondition} is no greater than $(\degree(u)-1-\sigma'_u)-1$. Since this is less than the right-hand side of \Cref{formula:keypaircondition}, the inequality holds. Notice that if we decrease $k$ by $1$, the term \Cref{formula:momentyterm} will become strictly less than the rank of $y$, thus the left-hand side of \Cref{lemma:keypaircondition} is larger than $(\degree(u)-1-\sigma'_u)-1$, which is no less than the right-hand side of \Cref{lemma:keypaircondition}. Therefore, this claimed value $k$ is the smallest possible $k$ as desired.

Since $\left\lvert\left\{x\in D_u\mid \moment_x\leq \bc^{\down}(u)+(k-1-b)\right\}\right\rvert$ is no greater than $|D_u|$, in \Cref{lemma:keypaircondition}, we have 
\begin{align}
b+1 \leq \degree(u)-1-\sigma'_u+|D_u|,
\end{align} which is exactly the number of key pairs described in both \Cref{lemma:updatea} and \Cref{lemma:updateb}. Therefore, no key pair belonging to vertex $u$ is missing.
\end{proof}

\begin{lemma}
\label{lemma:dsmemory}
$\newnode$ will be called $O(n)$ times only in the whole execution of \Cref{algorithm:main}. It means that at any moment while executing \Cref{algorithm:main}, there are $O(n)$ nodes storing in any $D_u$ or $\bq_u$ in total.
\end{lemma}
\begin{proof}
We can see that $\newnode$ will be called $num_u$ times (\Cref{update:oper4}) in each $\updatedsupward(u)$. From \Cref{update:oper2} we can see that $num$ is at most $\degree(u)-1$. Therefore, in all $\updatedsupward(u), u \in V(G)$, $\newnode$ will be called at most $\sum_{u \in V(G)} \degree(u)=O(n)$ times. Since no duplicating operation is involved throughout all operations, there are $O(n)$ nodes storing in any $D_u$ or $\bq_u$ in total at any moment while executing \Cref{algorithm:main}.
\end{proof}

Now we are ready to prove \Cref{lemma:update}.
\begin{proof}[Proof of \Cref{lemma:update}]
With \Cref{lemma:updatea} and \Cref{lemma:updateb}, we know how to modify current $D_u$ into one with all key pairs belonging to $u$. Specifically, \Cref{lemma:updatea} implies the first $\degree(u)-1-\sigma'_u$ key pairs and \Cref{lemma:updateb} implies all nodes in current $D_u$ can be directly modified altogether without changing the relative order among them. This allows us to call $\inctime$ to proceed with the update. This is vital because we are modifying $\moment$ while $D_u$ is sorted by $\moment$.

Firstly, we call $\inctime(\roott(D_u),0,-\bc^{\down}(u))$ to add a constant $-\bc^{\down}$ for every existing node in current $D_u$ (\Cref{update:oper1}). Then we create and insert $\degree(u)-1-\sigma'_u$ nodes with $\moment=0$ (\Cref{update:oper2} to \Cref{update:oper5}). Here we store $\num_u = \degree(u)-1-\sigma'_u$ for the future use. Lastly, we call $\inctime(\roott(D_u),1,0)$ to add a value of their rank to themselves. One can observe that these three operations will match the correct $\moment$ value mentioned in \Cref{lemma:updatea} and \Cref{lemma:updateb}. Therefore, \Cref{lemma:update} will update $D_u$ correctly so that it contains all key pairs $(u,k)$. 

Each $\inctime$ operation costs $O(1)$ runtime. By \Cref{lemma:dsmemory}, there are $O(n)$ insertions for all $\updatedsupward$ during the execution of \Cref{algorithm:main}. Since an insertion in $D_u$ costs $O(\log n)$ in amortized time. Therefore, the total time cost is $O(n \log n)$.
\end{proof} 

Now we will analyze $\revertds$, which is a process to revert the modification to $D_u$ in both $\cnt(u)$ and $\updatedsupward(u)$.

\IncMargin{1em}
\begin{algorithm}[H]
  \SetKwData{Left}{left}\SetKwData{This}{this}\SetKwData{Up}{up}  \SetKwComment{Comment}{$\triangleright$\ }{}
  \SetKwFunction{Swap}{swap}
  \SetKwFunction{insert}{insert}
  \SetKwFunction{delete}{delete}
    \SetKw{Break}{break}

  \SetKwProg{Fn}{Function}{:}{}
  \BlankLine

  $\IncTime (\roott(D_u),-1,0)$\; \label{revert:update1}
  \While{$D_u\neq \nil$}{  \label{revert:update2}
    $x\leftarrow \FindMin(D_u)$\;  \label{revert:update3}
    \If{$\moment_x = 0$}{  \label{revert:update4}
      \delete($D_u$,$x$)\;  \label{revert:update5}
    }
    \Else{\Break\;}  \label{revert:update6}
  }
    
  $\IncTime (\roott(D_u),0,\bc^{\downarrow}(x))$\;  \label{revert:update7}
  
  \For{$x \in \bq_u$}{  \label{revert:cnt1}
    $\insert(\roott{D_u},x)$\; \label{revert:cnt2}
  }
  \caption{\revertds($u$)}\label{algorithm:revert}
\end{algorithm}\DecMargin{1em}

\begin{lemma}
\label{lemma:revert}
    For any vertex $u \in V(G)$, $\revertds(u)$ will revert $D_u$ to the exact status before calling $\cnt(u)$ in \Cref{algorithm:upward} that it contains all key pairs $(v,k),v \in \son(u)$.
    During the execution of \Cref{algorithm:main}, all $\revertds$ operations take $O(n \log n)$ time in total.
\end{lemma}

\begin{proof}[Proof of \Cref{lemma:revert}]
$\revertds(u)$ can be divided into two parts:
\begin{itemize}
    \item \Cref{revert:update1} to \Cref{revert:update7}: Revert modification of $\updatedsupward(u)$ on $D_u$ from \Cref{update:oper1} to \Cref{update:oper7}.
    \item \Cref{revert:cnt1} to \Cref{revert:cnt2}: Revert modification of $\cnt(u)$ on $D_u$ from \Cref{computeC:loop} to \Cref{computeC:break}.
\end{itemize}

It is easy to see that these revert operations are symmetric to the previous modification and thus produce the original $D_u$ after calling $\revertds(u)$. Since deletion and insertion on $D_u$ both take in amortized $O(\log n)$, the time complexity for all $\revertds$ costs the same as all $\updatedsupward$ as $O(n \log n)$. 
\end{proof}

\subsection{Overall Analysis}
\label{sec:ds:overall-analysis-of-ds}
Now we are ready to prove \Cref{theorem:ds} from \Cref{lemma:deltaquery} ($\deltaquery$), \Cref{lemma:merge} ($\mergeupward$), \Cref{lemma:computeC} ($\cnt$), \Cref{lemma:deltasum} ($\deltasum$), \Cref{lemma:update} ($\updatedsupward$), \Cref{lemma:revert} ($\revertds$), \Cref{lemma:split} ($\splitds$) and \Cref{lemma:dsmemory} (Memory). After that, we will prove the main result \Cref{theorem:main}.

\dstheorem*

\begin{proof}[Proof of \Cref{theorem:ds}]
Firstly, we focus on \Cref{algorithm:upward}. For any leaf vertex $u$, $D_u$ is set to empty, which is correctly set. 

For the current vertex $u$, the $D_v$ is correctly maintained as containing all key pairs $(v, k)$ for all $v \in \son(u)$ by inductive hypothesis.

We know that $D_u$ contains the collection of all the key pairs from $D_v$ for all $v \in D_u$. This is because we have called $\mergeupward(u,v)$ for all $v \in \son(u)$ in an arbitrary order $\mathcal{I}$. In each call of $\mergeupward(u, v)$, by \Cref{lemma:merge}, $D_u$ will contain all the nodes from $D_v$. This proves that $D_u$ contains all the key pairs from $D_v$ for all $v \in \son(u)$. It implies that the assumption in \Cref{lemma:computeC} has been satisfied. Therefore, $\cnt(u)$ will return the correct value of $\bc^{\down}(u)$ and transport $x\in D_u$ satisfying $moment_x\leq \bc^{\down}(u)$ to $\bq_u$. By \Cref{lemma:deltasum}, $\deltasum$ produces the correct result. Now $D_u$ is the union of all key pairs $(v,k), v \in \son(u), k > \bc^{\down}(u)$, which satisfies the assumption of \Cref{lemma:update}. Therefore, $\updatedsupward(u)$ will produce a correct $D_u$ for $u$ such that it contains all key pairs $(u,k)$. 

Now we analyze \Cref{algorithm:downward}. For a non-root vertex $u$, at the time we visit $u$, we can use the inductive hypothesis to assume the process on $\parent(u)$ has been finished correctly. It means the value of $\bc(\parent(u))$ has been calculated correctly, and $D_u$ is restored to the status before $\mergeupward(u, v)$ happens. By \Cref{lemma:deltaquery}, $\deltaquery$ produces the correct result $\delta(u,\bc(\parent(u))$. By \Cref{lemma:revert}, $\revertds(u)$ restore $D_u$ to the union of all key pairs $(v,k), v \in \son(u)$. By enumerating in the reversed order of $\mathcal{I}$, $\splitds(u,v)$ are called repeatedly. We can see that such order of calling $\splitds(u,v), v \in \son(u)$ will satisfy the assumption of \Cref{lemma:split}. Therefore, $\splitds$ will produce the correct $D_v$ each time, which contains all key pairs of $v$. 

The overall time complexity can be immediately derived from combining \Cref{lemma:deltaquery}, \Cref{lemma:merge}, \Cref{lemma:computeC}, \Cref{lemma:deltasum}, \Cref{lemma:update}, \Cref{lemma:revert}, and \Cref{lemma:split} that all these subroutines cost $O(n\log n)$ time in total.

The space complexity is analyzed in \Cref{lemma:dsmemory}.
\end{proof}

By \Cref{theorem:ds}, we can prove \Cref{theorem:upward} and \Cref{theorem:downward}. With them, we can now prove \Cref{theorem:main}:

\maintheorem*

\begin{proof}
Let's analysis the procedure of \Cref{algorithm:main}:
\begin{itemize}
    \item From \Cref{algorithm:main:sp-start} to \Cref{algorithm:main:sp-end} we will skip the case with recurrent instances. The value of $\sum_{u \in V(G)} \sigma_u$ can be found by summing in $O(|V|)$, so this part will be finished in $O(n)$ time and $O(1)$ costs of memory.
    \item In \Cref{algorithm:main:root} and \Cref{algorithm:main:init} we will initialize the root $r$ and the vector $\sigma'$. Since it's just a memory copy operation, it uses $O(n)$ time and $O(n)$ extra memory.
    \item In \Cref{algorithm:main:upward} we call \upward($r$, $G$, $\sigma'$). By \Cref{theorem:upward} the procedure finishes in $O(n \log n)$ time. 
    \item In \Cref{algorithm:main:downward} we call \downward($r$, $G$, $\sigma'$). By \Cref{theorem:downward} the procedure finishes in $O(n \log n)$ time. 
    \item From \Cref{algorithm:main:recover-start} to \Cref{algorithm:main:recover-end}, we will recover the terminal configuration based on the value of $\bc(u)$ for all $\bc \in V(G)$. The iteration of the pair $(u, v)$ is equivalent to iterating all the edges in the graph $G$. Since $|E(G)| = |V(G)| - 1$ in a tree $G$, iterating over all edges (each edge will be iterated exactly twice) will use $O(n)$ time with $O(1)$ extra memory.
\end{itemize}

In addition to storing the tree structure with $O(n)$ memory, $\upward$ only needs a global variable $\sigma'$ to pass during the recursion calling, which is a vector of size $n$. Similarly, for $\downward$, the algorithm only needs to store a variable $u$ and $k$, which uses $O(1)$ memory on each vertex. We also need two vectors of size $n$ to store the computed $\bc^{\down}$ and $\bc$. By \Cref{sec:ds}, there are also some additional arrays of $O(n)$ length to support operations. Therefore, excluding $\mathcal{D}$, we only need $O(n)$ memory. Since $\mathcal{D}$ takes $O(n)$ memory at any moment by \Cref{theorem:ds} (proved in \Cref{sec:ds:overall-analysis-of-ds}), the whole algorithm still takes $O(n)$ memory.

Therefore, the algorithm finds the correct configuration in $O(n \log n)$ time with $O(n)$ memory.
\end{proof}
\section{Algorithms on Other Structured Graphs}

\label{sec:further}

In this section, we mainly discuss how to modify the tree algorithm to reach $O(n)$ time complexity when the input graph is a path. We also give an algorithm on solving the sandpile prediction problem on a clique, which also runs in $O(n)$ time and $O(n)$ memory.

\subsection{Sandpile Prediction on Paths}
\label{sec:path}

\paththeorem*

\begin{definition}[$\text{Path}_n$]
$\text{Path}_n$ is defined as an undirected graph $G(V,E)$ such that $V=\{1,2,\dots,n\}$ and $E=\{(u,u+1) \mid 1\leq u < n-1\}$.
\end{definition}
Since $\text{Path}_n$ is also of the tree structure, if we call \Cref{algorithm:main} directly, we can solve any sandpile instance on $\text{Path}_n$ with $O(n \log n)$ time and $O(n)$ memory by \Cref{theorem:main}. We conjecture that the runtime is actually $O(n)$. The key idea to prove this result is through the \textit{Deque Conjecture}, which is a corollary of the famous unproven \textit{Dynamic Optimality Conjecure} \cite{sleator1985self}. The recent known result of the Deque Conjecture is by Seth Pettie \cite{pettie2007splay}. 

Now we will show that by modifying $\downward$ (\Cref{algorithm:downward}) and $\revertds$ (\Cref{algorithm:revert}), \Cref{algorithm:main} will have an $O(n)$ runtime when the input graph $G$ is $\text{Path}_n$. The modified algorithm does not rely on \textit{Deque Conjecture}. 

We fix the root at vertex $1$. In this way, every vertex $u \in [1,n-1]$ has exactly one child, $u+1$. Firstly, we will prove that $\upward(r,G,\sigma')$ (\Cref{algorithm:upward}) runs in $O(n)$ time. 

\begin{lemma}
\label{chain:upward}
Given an sandpile instance $S(G,\sigma)$ such that $G$ is a $\text{Path}_n$, $\upward(r,G,\sigma')$ runs in $O(n)$ time with $O(n)$ extra memory.
\end{lemma}

\begin{proof}
We will prove the lemma similar to the proof of \Cref{theorem:downward}, where the total time complexity relies on the time cost for all $D_u$ operations. We assume currently we are at non-leaf vertex $u$. Since $D_u$ is initialized to be $\varnothing$ and the current visit vertex $u$ has exactly one child $v=u+1$, $\mergeupward(u,v)$ will be called only once and merges $D_v$ with an empty splay tree. By the small-to-large principle, $D_u$ will inherit $D_v$ directly taking $O(1)$ time. Therefore, throughout the whole execution of \Cref{algorithm:upward}, we are doing operations on one splay tree $T$.

By \Cref{lemma:computeC}, the time cost of all $\cnt$ is dominated by $O(n)$ calls of $\findmin(T)$ operation on $T$ in total, which finds the minimum rank node each time. For $\updatedsupward(u)$, by \Cref{lemma:update}, the time cost is dominated by the insertion of $O(n)$ nodes, each with $\moment=0$, which is always being inserted as the node with the minimum rank each time. Here we combine these two parts, and apply the dynamic finger theorem \Cref{theorem:dynamic-finger-theorem}, since the rank difference between any two accesses is at most $1$, the total time complexity is $O(n)$. Therefore, $\upward(r,G,\sigma')$ runs in $O(n)$ time.

The memory usage remains $O(n)$.
\end{proof}

In the original $\downward$, although we can analyze $\splitds(u,v)$ similar to $\mergeupward(u,v)$, for $\deltaquery$ and $\revertds$, we cannot apply the dynamic finger theorem directly to achieve the linear runtime. Therefore, we propose the following alternate process \Cref{algorithm:downwardm} ($\downwardm$) and \Cref{algorithm:revertm} ($\revertm$) for the path case. 

\IncMargin{1em}
\begin{algorithm}
  \SetKwData{Left}{left}\SetKwData{This}{this}\SetKwData{Up}{up}
  \SetKwComment{Comment}{$\triangleright$\ }{}
  \SetKwProg{Fn}{Function}{:}{}
  \If{$u=r$}{
    $k \leftarrow 0$\;
    $count \leftarrow 0$\;
  }
  \Else {
    $count \leftarrow count+\delquem(D_u, \bc(u-1))$\;\label{downwardm:calrank}
    $k\leftarrow \bc(\parent(u))-count$\;\label{downwardm:calk}
  }
  $count\leftarrow$ \revertm($u$,$count$)\; 
  $\bc(u) \leftarrow \bc^{\down}(u) + k$\;

  \If{$\son(u) \neq \nil$}{
    \downwardm($\son(u)$, $G$,$count$)\;
  }
  \caption{\downwardm($u$,$G$,$count$)}\label{algorithm:downwardm}
\end{algorithm}\DecMargin{1em}

\IncMargin{1em}
\begin{algorithm}[H]
  \SetKwData{Left}{left}\SetKwData{This}{this}\SetKwData{Up}{up}  \SetKwComment{Comment}{$\triangleright$\ }{}

  \SetKwProg{Fn}{Function}{:}{}
  \BlankLine

  $\inctime (\roott(D_u),-1,-count)$\;\label{revertm:minus}
  \If{$num_u\leq count$}{\label{revertm:del}
    $count\leftarrow count-num_u$\;
  }
  \Else{
    $num_u\leftarrow num_u-count$\;
    $count \leftarrow 0$\; \label{revertm:rankzero}
    \While{$num_u>0$}{ \label{revertm:delete1}
      $\delete(D_u,\findmin(D_u))$\;
      $num_u\leftarrow num_u-1$\; \label{revertm:delete2}
    }
  }

  $\IncTime (\roott(D_u),0,\bc^{\downarrow}(x))$\;\label{revertm:plus}

  \If{$count>0$}{\label{revertm:ins}
    $count\leftarrow count+|\bq_u|$\; \label{revertm:rankadd}
  }
  \Else{
    \For{$x \in \bq_u$}{ \label{revertm:insert1}
      $\insert(\roott(D_u),x)$\; \label{revertm:insert2}
    }
    
  }
  \Return $count$ \; \label{revertm:return}
  \caption{\revertm($u$,$count$)}\label{algorithm:revertm}
\end{algorithm}\DecMargin{1em}

The following lemma shows an additional property for nodes in $D_u$, which is helpful to our path algorithm: 

\begin{lemma}
\label{chain:property}
Consider the execution of $\downward(u,G)$ (\Cref{algorithm:downward}). For any node $x \in D_u$ such that $\moment_x \leq \bc(\parent(u))$, we have $\moment_x \leq \bc(u)$ after calling $\revertds(u)$ (\Cref{algorithm:revert}).
\end{lemma}

\begin{proof}
Let $\tau$ denote $\left\lvert\left\{x\in D_u|\moment_x\leq \bc(\parent(u))\right\}\right\rvert$. During the execution of \Cref{algorithm:downward}, we will change all $\moment_x$ back to $\moment_x -\rank_{D_u}(x)+\bc^{\down}(u)$ and compute $\bc(u)=\bc^{\down}(u)+k=\bc(\parent(u))-\tau+\bc^{\down}(u)$ by \Cref{lemma:deltaequal}. Since the relative order between nodes remains the same after calling $\revertds(u,v)$, we only have to prove $\moment_z \leq \bc(u)$ where $\rank_{D_u}(z)=\tau$. Since we have $\moment_z -\rank_{D_u}(z)+\bc^{\down}(u)=\moment_z -\tau+\bc^{\down}(u)\leq \bc(\parent(u))-\tau+\bc^{\down}(u)=\bc(u)$, we have proved our lemma.
\end{proof}

\IncMargin{1em}
\begin{algorithm}[H]
  \SetKwData{Left}{left}\SetKwData{This}{this}\SetKwData{Up}{up}

  \SetKwProg{Fn}{Function}{:}{}
  \BlankLine
  $count\leftarrow 0$\; 

  \While{$D_u\neq \varnothing$}{ \label{pathquery:while}
    $x\leftarrow \findmin(\roott(D_u))$\; \label{pathquery:findmin}
    \If{$\moment_x\leq k$}{ \label{pathquery:if}
      $\delete(D_u,x)$\; \label{pathquery:delete}
      $count\leftarrow count+1$\; \label{pathquery:count}
    }
    \Else{\Break\;} \label{pathquery:break}
  }

  \Return $count$\; \label{pathquery:return}
  
  \caption{\delquem($u$, $k$)}\label{algorithm:deltaquerym}
\end{algorithm}\DecMargin{1em}

\Cref{algorithm:main:downward} of \Cref{algorithm:main} ($\downward(r,G)$) will be replaced with $\downwardm(r,G,0)$ (\Cref{algorithm:downwardm}). An extra subroutine $\delquem$ (\Cref{algorithm:deltaquerym}) is also needed in our path algorithm:
\begin{lemma}
\label{lemma:chainquery}
$\delquem(u,k)$ returns the number of node $x \in D_u$ such that $\moment_x \leq k$ and deletes them from $D_u$.
\end{lemma}

\begin{proof}
During the process of $\delquem(u,k)$, we repeatedly find the node $x \in D_u$ where $\rank_{D_u}(x)$ is the minimum until $D_u$ becomes empty (\Cref{pathquery:while}). For each $x$, we check if $\moment_x \leq k$ (\Cref{pathquery:if}) holds. If so, we will delete it from $D_u$ (\Cref{pathquery:delete}) and increase the counter by $1$ (\Cref{pathquery:count}). Otherwise, since nodes are ordered by $\moment$ from small to large, we have found all nodes satisfying the condition. Thus we exit the loop (\Cref{pathquery:break}). Since $count$ keeps track of the number of nodes $x$ with $\moment_x \leq k$, we should return $count$ as the result (\Cref{pathquery:return}).
\end{proof}

Now we are ready to analyze \Cref{algorithm:downwardm}.

\begin{lemma}[Correctness of \Cref{algorithm:downwardm}]
\label{chain:downward}
\Cref{algorithm:downwardm} calculates the correct value of all $\bc(u)$ for $u \in V(G)$.
\end{lemma}

\begin{proof}
We will prove the correctness of \Cref{algorithm:downwardm} by induction. We assume the computation on any vertex $v$ visited before $u$ is correct. For vertex $u$, we maintain a variable $count$ in \Cref{algorithm:downwardm} denoting the number of pairs $(u,\moment_u)$ satisfying $v \leq \bc(\parent(u))$. When $u$ is the root, we initialize $count$ to $0$. In the following, we assume $u$ is not the root. 

By \Cref{lemma:chainquery}, we know that after executing \Cref{downwardm:calrank}, $count$ will be increased by the number of key pairs $(u,x)$ satisfying $x \leq \bc(\parent(u))$. By induction, $count$ stores the number of key pairs $(\parent(u),y)$ satisfying $y \leq \bc(\parent(\parent(u)))$. By \Cref{chain:property}, we know these key pairs $(\parent(u),k)$ satisfy $k \leq \bc(\parent(u))$ after being reverted. Therefore, after calling $\downwardm$ on $u$ (\Cref{downwardm:calrank}), $count$ will be correctly maintained. Moreover, all nodes satisfying $\moment_x \leq \bc(\parent(u))$ are deleted from $D_u$. This process can be regarded as aggregating all previously counted key pairs to one counter since they remain to be legal in the following recursion.
 
By \Cref{lemma:deltaequal}, $k=\delta(u,\bc(\parent(u)))$ is computed as $\bc(\parent(u))-count$ (\Cref{downwardm:calk}).

$\revertm$ is a modified version $\revertds$, which also reverts $\updatedsupward$'s modification on $D_u$. In the beginning, the last $\inctime$ operation in $\updatedsupward$ (\Cref{revertm:minus}) still needs to be reverted. Previously in $\updatedsupward(u)$, $\num_u$ stores the number of insertions. Since now we aggregate nodes of $\rank \in [1,count]$, the actual rank of any node $x \in D_u$ should be $\rank_{D_u}(x)+count$. Thus when we revert insertions in $\updatedsupward(u)$, we have to check how many of them are already aggregated in $count$. We update $\num_u$ and $count$ correspondingly (\Cref{revertm:del} to \Cref{revertm:rankzero}). If there is any inserted node in $D_u$ needs deleted, we can simply repeatedly acquire them by calling $\findmin(D_u)$ and deleting them (\Cref{revertm:delete1} to \Cref{revertm:delete2}). After undoing insertions, the first $\inctime$ operation should be reverted as well (\Cref{revertm:plus}). Lastly, to deal with nodes stored in $bq_u$, we first check if nodes in $bq_u$ should be added back to $D_u$ (\Cref{revertm:insert1} to \Cref{revertm:insert2}). If not, we simply increase $count$ (\Cref{revertm:rankadd}). This can be done by checking if $count$ is positive (\Cref{revertm:ins}. We return the new $rank$ back to $\downwardm$ in the end (\Cref{revertm:return}). After $\revertm$, $\bc(u)$ is computed as $\bc^{\down}(u)+k$ by \Cref{lemma:value-of-c}. Then we continue visiting $u$'s only child if exists.

By applying mathematical induction on the arguments above, the algorithm proceeds the correct value of $\bc(u)$ for all $u \in V(G)$, which proves the correctness of \Cref{algorithm:downwardm}.
\end{proof}

\begin{lemma}[Time and Memory Used in \Cref{chain:downward}]
\label{chain:downward2}
\Cref{algorithm:downwardm} calculates the correct value of all $\bc(u)$ for $u \in V(G)$ in a total of $O(n)$ time and $O(n)$ memory.
\end{lemma}
\begin{proof}
Similar to the proof of \Cref{theorem:downward}, the total time complexity relies on the time cost for all $D_u$ operations. $\delquem$ is implemented as deleting the node with the minimum rank after finding it. Thus in each $\delquem$, every node will be deleted exactly once except for one node which will only be found but not deleted. In $\revertm(u,rank)$, we can have a similar analysis for $\insert$, $\delete$ and $\findmin$.

Overall, there are $O(n)$ times of operations of these three kinds in total. Notice that all these operations during the execution of \Cref{algorithm:downwardm} access the node with the minimum rank. By applying the dynamic finger theorem \Cref{theorem:dynamic-finger-theorem}, the total time complexity is $O(n)$. By \Cref{lemma:dsmemory}, the memory usage is also $O(n)$ same as the previous $\downward$.
\end{proof}

Combining \Cref{chain:upward}, \Cref{chain:downward}, \Cref{chain:downward2}, we are able to prove \Cref{theorem:main-path}.

\subsection{Sandpile Prediction on Cliques}
\label{sec:clique}
We also study one of the most classic structured graphs and come up with a bound showing that one only needs to simulate $O(n)$ firings to reach the terminal configuration or determine that it will not terminate.

\begin{restatable}[Sandpile Prediction on a Clique]{theorem}{theoremclique}
\label{theorem:clique}
Given a sandpile instance $S(G,\sigma)$ such that $G$ is a clique on $n$ vertices. There is an algorithm that can determine whether $S$ will terminate and compute the terminal configuration of $S$ in $O(n)$ time and $O(n)$ memory.
\end{restatable}

\IncMargin{1em}
\begin{algorithm}

  \SetKwInOut{Input}{input}\SetKwInOut{Output}{output}  

  \SetKwProg{Fn}{Function}{:}{}
  \Input{$G$, configuration $\sigma$}
  \Output{the terminal configuration $\sigma^{T}$ of the instance $S(T, \sigma)$}
  \BlankLine 
  $count \leftarrow 0$ \; \label{clique:countinit}
  \For{$u \in V$} {
    \While {$\sigma_u \geq n-1$} {  \label{clique:firingfirstloop}
        $\sigma_u \leftarrow \sigma_u - (n-1)-1$ \; \label{clique:firingfirst1}
        $count \leftarrow count + 1$ \; \label{clique:firingfirst2}
        \If {$count \geq n-1$} { \label{clique:checka1}
            \Return $\bot$\; \label{clique:checka2}
        }
    }
  }
  $j \leftarrow 0$ \;
  \For{$u \in V$} { \label{clique:fori}
    $\bucket_{\sigma_i}.\text{append}(u)$ \;
    $j \leftarrow \max(j,\sigma_i)$
  }
  \While{$j > 0$} { \label{clique:forj}
        \For {$x \in \bucket_j$} {
            \If {$\sigma_{x}+count\geq n-1$} {
                $\sigma_{x} \leftarrow \sigma_{x} - (n-1)-1$ \; \label{clique:fire1}
                $count \leftarrow count + 1$ \; \label{clique:fire2}
                \If {$count \geq n-1$} { \label{clique:checkb1}
                    \Return $\bot$\; \label{clique:checkb2}
                }
            }
        }
        $j \leftarrow j-1$ \;
    }
    \For{$u \in V$} { \label{clique:lastfor}
        $\sigma_i \leftarrow \sigma_i + count$ \;
  }
  \Return{$\sigma$}\;
  
  \caption{\textsc{SolveClique}($n$,$\sigma$)}\label{algorithm:clique}
\end{algorithm}\DecMargin{1em}

To begin with, we first bound the total number of firing on cliques. 

\begin{lemma}[Firing Bound for Sandpile on Clique]
\label{lemma:cliquebound}
Given a sandpile instance $S(G,\sigma)$ such that $G$ is a clique on $n$ vertices. The instance will terminate if and only if the total number of firings is no greater than $n-2$.
\end{lemma}

\begin{proof}
Assume that the configuration will become terminal after firing vertices $u_1, u_2, \cdots, u_k$ ($k \geq n-1$). Then consider the last $n-1$ firing operations $A = \{u_{k-n+2}, u_{k-n+3}, \cdots, u_{k}\}$. There must exist a vertex $v \in V(G)$ such that $v \notin A$ since $|A| < |V|$. Note that $\sigma_{v} \geq 0$ before the $(k-n+2)$-th operation, and it will receive one additional chip in the last $n-1$ firings. This implies after the last operation, $\sigma_{v} \geq n-1 = \degree(v)$. This contradicts with our assumption that the configuration will become terminal after all the $k$ operations. Therefore, if a sandpile instance is a terminal instance, then the total number of the firing operations must be no greater than $n-2$.
\end{proof}

\begin{proof}[Proof of \Cref{theorem:clique}]
The general idea of \Cref{algorithm:clique} is to simulate for the first $n-1$ times of firings. By \Cref{lemma:cliquebound}, we know that if it terminates, we have the final configuration. We define a variable $count$ to keep track of the number of firings, initialized as $0$ (\Cref{clique:countinit}). If it exceeds $n-1$, then the instance will not terminate. We check this condition whenever we apply a firing (\Cref{clique:checka1} to \Cref{clique:checka2};\Cref{clique:checkb1} to \Cref{clique:checkb2}). 

For any vertex $u$, we have $\degree(u)=n-1$. Therefore, when firing vertex $u$, it is not efficient to add chips to each neighbor. By keeping a counter $count$ to keep track of the number of firings, we are able to express the current number of chips on vertex $u$ as $\sigma_u + count$. For any firing on vertex $u$, we first increase $count$ by $1$. Since vertex $u$ cannot profit from this firing and $n-1$ chips will be removed, we decrease $\sigma_u$ by $(n-1)+1$.

Initially, since $\sigma_u$ could be large, we keep firing every vertex $u$ (\Cref{clique:firingfirst1} to \Cref{clique:firingfirst2}) until $\sigma_u<n-1$ (\Cref{clique:firingfirstloop}). After this stage, we have $0 \leq \sigma_u<n-1$ for $u \in V$. We store each vertex $u$ into the corresponding $\bucket_{\sigma_u}$. Here we implement each $\bucket$ as a deque.

Notice that in our simulation, $\sigma_u$ will only decrease after any firing on vertex $u$. Since we maintain each vertex $u$ in the corresponding $\bucket_{\sigma_u}$, we can simulate all firings by simply iterating $\bucket$ in descending order. Assuming we are currently visiting vertex $u$ in $\bucket_j$, we have the number of chips equal to $j+count$ by definition. If $j+count \geq n-1$, then vertex $u$ can be fired once. We update $\sigma_u$ and $count$ (\Cref{clique:fire1} to \Cref{clique:fire2}) if $u$ can be fired. $\sigma_u$ will become negative after this firing. Since $count < n-1$, if $\sigma_u \leq 0$, $\sigma_u + count$ must be 
smaller than $n-1$. Thus any vertex will be fired at most once and no vertex will be missed if it can be fired. Since $count<n-1$, when $j=0$, $\sigma_u+count=0+count<n-1$, no firing will happen. Thus we terminate the enumeration.

In this way, we successfully track and simulate all firings if there are less than $n-1$ of them. By \Cref{lemma:cliquebound}, there will be at most $n-1$ firings. \Cref{clique:fori} is $O(n)$. Since $0 \leq \sigma_u < n-1$,  we only need $O(n)$ $\bucket$s in total and thus \Cref{clique:forj} is $O(n)$. Therefore, \Cref{algorithm:clique} runs in $O(n)$ times. Since each vertex only exists in one $\bucket$ at any moment, \Cref{algorithm:clique} takes $O(n)$ memory.

\end{proof}

\printbibliography
\appendix
\section{Uniqueness Analysis on Sandpile with Sinks}
\label{sec:sinkmodel}
\begin{definition}[Auxiliary Graph and Auxiliary Instance]
\label{def:auxiliary-graph}
Let $S(G, \sigma, M)$ be a sandpile instance with a non-empty set of sinks $M$, and $A = |V(G)| + \lvert\lvert \sigma \rvert\rvert_1$. For any sink vertex $v_i \in M$, we will create $A$ auxiliary nodes $v'_{i,1}, v'_{i,2}, \cdots, v'_{i,A}$ and add an edge between $(v_i, v'_i)$. The constructed graph $G'(V', E')$ is called the Auxiliary Graph of $S(G, \sigma, M)$.

Define the configuration $\sigma'$ as 

\[
\sigma'_v = \begin{cases}\sigma_v & v \in V\\ 0 & \text{otherwise}\end{cases}
\]

The instance $S'(G', \sigma')$ is called the auxiliary instance of $S(G, \sigma, M)$.

\end{definition}

In short, the auxiliary graph can be regarded as a graph obtained by attaching $A$ vertices without any chips to each sink in the graph, which makes it impossible to fire.

For any sandpile instance $S(G, \sigma, M)$ with sinks $M \ne \varnothing$. The following conditions were held for the auxiliary graphs.

\begin{itemize}
    \item For all $u \in V$, we have $u \in V'$. That is, $V \subseteq V'$.
    \item For all $(u, v) \in E$, we have $(u, v) \in E'$. That is, $E \subseteq E'$.
    \item For all $(u, v) \in E'$ and $u, v \in V$, we have $(u, v) \in E$.
    \item For all $u \in V \setminus M$, $\degree_{G}(u) = \degree_{G'}(u)$.
\end{itemize}

These conditions imply that for any $v \in V$, firing vertex $v$ in both instances will have the same behavior. Formally, every firing operation on $G$ corresponds to a firing operation on $G'$ (and vice versa), and the following equation always holds if we perform the firing operation on both graphs simultaneously.

Furthermore, since the firing operation won't increase the number of chips in the whole graph. So at any time, for the sink vertex $v \in M$, we have $\sigma_v \leq A < \degree_{G'}(v)$. Thus, for any $v \in M$, the vertex $v$ will never be full in the instance $S'$.

More precisely:

\begin{enumerate}
    \item For any $u \in V$, if it is full in the instance $S$, then it is also full in the instance $S'$.
    \item For any $u \in V'$, if it is full in the instance $S'$, then we must have $u \in V$, and $u$ is also full in the instance $S$.
    \item For any full vertex $u \in V$, the equation \Cref{formula:thesigmaintheoriginalgraphequalstothesigmainthenewgraph} still holds after firing vertex $u$ in both instances $S$ and $S'$.
\end{enumerate}

This shows that firing operation in the auxiliary graph is equivalent to performing operations in the original graph. It tells us the properties of the original sandpile instance can be transformed into the sandpile instance with sinks. \Cref{lemma:unique-terminal-configuration} is the most important one, and it can be generalized as \Cref{lemma:unique-terminal-configuration-with-sinks}.

\begin{lemma}[Unique Terminal Configuration with Sinks]
\label{lemma:unique-terminal-configuration-with-sinks}
Let $S(G, \sigma, M)$ be a sandpile instance with a non-empty set of sinks $M$. Let $T \subseteq V(G)\setminus M$ be any subset of vertices excluding sinks. Suppose the process that keeps firing all the full vertices in $T$ until $\sigma$ is local terminal in $T$. Then: 
\begin{enumerate}
    \item Any firing order will reach the same terminal configuration.
    \item For each vertex $u$, Any firing order will vertex $u$ the same number of times.
\end{enumerate}
\end{lemma}

\begin{proof}
Let $G'$ be the auxiliary graph and $S'(G', \sigma)$ be the auxiliary instance. Note that the following condition holds for the instance $S'$:

\begin{align}
\sigma_u = \sigma'_u \forall u \in V(G)
\label{formula:thesigmaintheoriginalgraphequalstothesigmainthenewgraph}
\end{align}

For any subset of non-sink vertices in the original graph $T \subseteq V\setminus M$, consider the firing process to make $S$ local terminal in $T$. Apply the same firing operation in $S'$ will obtain the same results as in $S$. By \Cref{lemma:unique-terminal-configuration}, all the firing orders will obtain the same configuration. This proves all the firing orders will make $\sigma$ become the same final configuration.
\end{proof}

The proof showed \Cref{lemma:unique-terminal-configuration-with-sinks} even further, that the unique configuration obtained by making $S$ local terminal in $T$ is exactly the same as the one in $S'$.

\begin{corollary}[Corollary of \Cref{lemma:unique-terminal-configuration-with-sinks}]
    
Let $S(G, \sigma, M)$ be a sandpile instance with a non-empty set of sinks $M$, and $S'(G', \sigma', M)$ be the auxiliary instance of $S$. Let $T \subseteq V(G)\setminus M$ be any subset of vertices excluding sinks. Consider the following procedure:

\begin{itemize}
\item Keep firing all the full vertices in $T$ in the instance $S(G, \sigma, M)$ until $S$ is local terminal in $T$.
\item Keep firing all the full vertices in $T$ in the instance $S'(G', \sigma')$ until $S'$ is local terminal in $T$.
\end{itemize}

Then $\sigma_v = \sigma'_v$ for all $v \in V(G) \setminus M$ after the procedure.
\end{corollary}

Finally, we will prove that the number of firings that happened on each vertex is bounded by $O(|M|^4 \cdot (\lvert\lvert \sigma \rvert\rvert_1 + n)^4 )$.

\begin{lemma}[Upper Bound of the Firing Number]
\label{lemma:upper-bound-of-firing-number-with-sinks}
For a sandpile instance $S(G, \sigma, M)$ with connected graph $G$ and the non-empty set of sinks $M$. The firing number of each vertex is bounded by $O(|M|^4 \cdot (\lvert\lvert \sigma \rvert\rvert_1 + n)^4 )$.
\end{lemma}

\begin{proof}
Consider the auxiliary instance $S'(G', \sigma')$. By \Cref{lemma:unique-terminal-configuration}, each firing on $S$ corresponds to a firing on $S'$, so for all $v \in V(G) \setminus M$, the number of firings that happened on the vertex $v$ in $S$ equals to the number of firings that happened on the vertex $v$ in $S'$.

On the other hand, in the graph $G'$, we have $|V(G')| = n + |M| \cdot (\lvert\lvert \sigma\rvert\rvert_1 + n)$. By \cite{tardos1988polynomial}, the firing number of each vertex will be no more than $|V(G')|^4$, which is exactly $O(|M|^4 \cdot (\lvert\lvert \sigma \rvert\rvert_1 + n)^4 )$
\end{proof}

\section{Sandpile on Trees with Sinks}
\label{sec:sinktree}
We will discuss how to adapt our tree algorithm to solve the sandpile prediction problem with at most three sinks on a tree.

\begin{restatable}[Sandpile Prediction on Tree with Three Sinks]{theorem}{treesink}
\label{theorem:treesink}
Given a sandpile instance $S(G,\sigma,M)$ such that $G$ is a tree and the sink vertices set $M$ satisfying $|M| \leq 3$, there is an algorithm that can compute the terminal configuration of $S$ in $O(n \log n + \log \lvert\lvert \sigma \rvert\rvert_1 \cdot \log n)$ time, with $O(n)$ memory. $\sigma$ denotes the initial configuration.
\end{restatable}

We first introduce a decomposing subroutine to convert the sink to the leaf of the tree, and we can regard the sink as a special vertex in the tree and apply the algorithm on trees without sink, thus the algorithm \Cref{algorithm:main}.

Therefore, we need to decompose the tree into several components and design an alternative way to store key pairs information on each $D_u$, and prepare different data structure realizations of the function $\cnt$ and $\deltaquery$, $\mergeupward$ and $\splitds$,$\updatedsupward$ and $\revertds$.

\subsection{Decomposing the Tree into Several Components}
\label{sec:treedecom}

We will decompose the given tree into several new trees so that each sink vertex will be the leaf of each tree. The chips on the sink vertex will be ignored, and no firing operation could happen on sink vertices. Since there's only one simple path between any pair of vertices in a tree, all the sink vertices divide the tree into several independent parts. Hence, we can divide the tree using the sink vertex, as described in \Cref{figure:decomposing-trees}.

\begin{figure}[h]
\centering
\begin{tikzpicture}

\begin{scope}[shift={(-4,0)}]
 \tikzstyle{every node}=[fill=black!30,circle,minimum size=0.6cm,inner sep=1pt]
 \tikzstyle{sink}=[fill=blue!50,circle,minimum size=0.6cm,inner sep=1pt]
  \node[sink] (1) at (0,0){};
  \node (2) at (-1.5,-1){};
  \node (3) at (1.5,-1){};
  \node (5) at (-2.5,-2){};
  \node[sink] (6) at (-1.5,-2){};
  \node (7) at (-0.5,-2){};
  \node (4) at (1,-2){};
  \node (8) at (2,-2){};
  \node (12) at (-1,-3){};
  \node (9) at (0.5,-3){};
  \node[sink] (13) at (1.5,-3){};
  \node (14) at (2.5,-3){};
  \node (15) at (-1.5,-4){};
  \node (16) at (-0.5,-4){};
  \node (17) at (2,-4){};
  \foreach \from/\to in {1/2,1/3,3/4,2/5,2/6,2/7,3/8,4/9,6/12,8/13,8/14,12/15,12/16,13/17}
    \draw (\from) -- (\to);
\end{scope}

\begin{scope}[shift={(4,0)}]
 \tikzstyle{every node}=[fill=black!30,circle,minimum size=0.6cm,inner sep=1pt]
 \tikzstyle{sink}=[fill=blue!50,circle,minimum size=0.6cm,inner sep=1pt]
  \node[sink] (1) at (0,0){};
  \node (2) at (-1.5,-1){};
  \node (3) at (1.5,-1){};
  \node (5) at (-2.5,-2){};
  \node[sink] (6) at (-1.5,-2){};
  \node (7) at (-0.5,-2){};
  \node (4) at (1,-2){};
  \node (8) at (2,-2){};
  \node (12) at (-1,-3){};
  \node (9) at (0.5,-3){};
  \node[sink] (13) at (1.5,-3){};
  \node (14) at (2.5,-3){};
  \node (15) at (-1.5,-4){};
  \node (16) at (-0.5,-4){};
  \node (17) at (2,-4){};
  \foreach \from/\to in {1/2,1/3,3/4,2/5,2/6,2/7,3/8,4/9,6/12,8/13,8/14,12/15,12/16,13/17}
    \draw (\from) -- (\to);
    
    \draw[draw=red,line width=2pt,dotted] (-1.25,-1.5)[rounded corners = 0.3cm] -- (-1.75,-1.5)[rounded corners = 0.3cm] --(-2.5,-3.5)[rounded corners = 0.3cm] --(-1.5,-4.5)[rounded corners = 0.3cm] --(0.5,-4.5)[rounded corners = 0.3cm] --cycle;
    \draw[draw=blue,line width=2pt,dotted] (-3.5,-2.5)[rounded corners = 0.3cm] --(0.5,-2.5)[rounded corners = 0.3cm] --(-1,-1)[rounded corners = 0.3cm]--(0.75,0)[rounded corners = 0.3cm]--(-0.125,0.625)[rounded corners = 0.3cm]--(-2,-1)[rounded corners = 0.3cm]  --cycle;
    \draw[draw=purple,line width=2pt,dotted] (1.75,-4.75)[rounded corners = 0.3cm] --(2.5,-4.25)[rounded corners = 0.3cm] --(1.75,-2.25)[rounded corners = 0.3cm] --(1,-2.75)[rounded corners = 0.3cm] --cycle;
    \draw[draw=green,line width=2pt,dotted] (0,-3.5)[rounded corners = 0.3cm] --(3,-3.5)[rounded corners = 0.3cm] --(2,-0.875)[rounded corners = 0.3cm] --(0.125,0.75)[rounded corners = 0.3cm]--(-0.625,-0.125)[rounded corners = 0.3cm] --(1,-1)[rounded corners = 0.3cm] --cycle;
\end{scope}

\end{tikzpicture}
\caption{Decompose a tree with sink vertices into multiple components, where each component is a tree in which all the sink vertices have a degree of exactly $1$.}
\label{figure:decomposing-trees}
\end{figure}
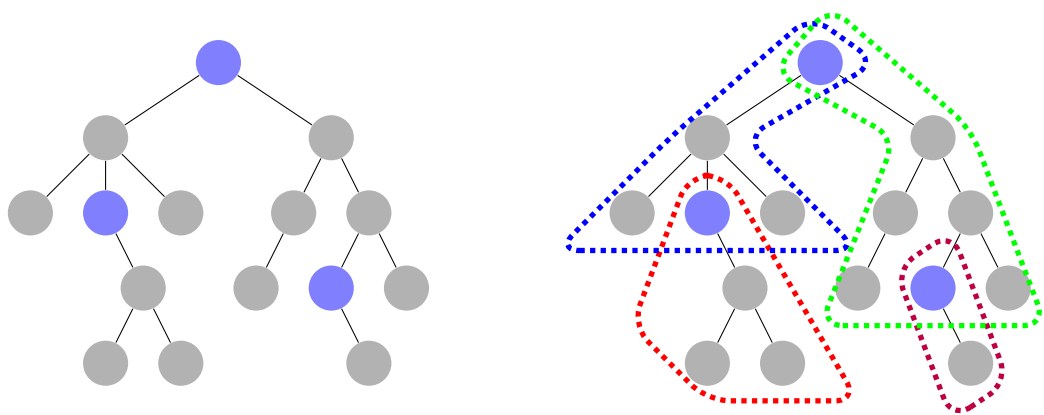

More precisely, consider the forest obtained by removing all the sink vertices in the tree. Because we cannot perform any firing operations on a sink vertex, each component is independent from each other. It implies we can treat the original tree as the forest we obtained. Furthermore, since removing sinks in the graph will change the degree of the neighbours for the sinks, we have to add a dummy sink for these vertices to make sure their degree won't be changed after the decomposition.


The procedure we described to divide the tree is exactly the way to remove vertices described in \Cref{definition:vertex-removal}. The instance we obtained after the decomposition is exactly $S \setminus M$, as we showed in \Cref{figure:decomposed-forests}.

\begin{figure}[H]
\centering
\begin{tikzpicture}

\begin{scope}[shift={(-6,0)}]
 \tikzstyle{every node}=[fill=black!30,circle,minimum size=0.6cm,inner sep=1pt]
 \tikzstyle{sink}=[fill=blue!50,circle,minimum size=0.6cm,inner sep=1pt]
  \node[sink] (6) at (0,0){};
  \node (12) at (0.5,-1){};
  \node (15) at (0,-2){};
  \node (16) at (1,-2){};
  \foreach \from/\to in {6/12,12/15,12/16}
    \draw (\from) -- (\to);
\end{scope}

\begin{scope}[shift={(-2,0)}]
 \tikzstyle{every node}=[fill=black!30,circle,minimum size=0.6cm,inner sep=1pt]
 \tikzstyle{sink}=[fill=blue!50,circle,minimum size=0.6cm,inner sep=1pt]
  \node[sink] (1) at (0,0){};
  \node (2) at (-0.5,-1){};
  \node (5) at (-1.25,-2){};
  \node[sink] (6) at (-0.5,-2){};
  \node (7) at (0.25,-2){};
  \foreach \from/\to in {1/2,2/5,2/6,2/7}
    \draw (\from) -- (\to);
\end{scope}

\begin{scope}[shift={(0,0)}]
 \tikzstyle{every node}=[fill=black!30,circle,minimum size=0.6cm,inner sep=1pt]
 \tikzstyle{sink}=[fill=blue!50,circle,minimum size=0.6cm,inner sep=1pt]
  \node[sink] (1) at (0,0){};
  \node (3) at (0.5,-1){};
  \node (4) at (0,-2){};
  \node (8) at (1,-2){};
  \node (9) at (-0.5,-3){};
  \node[sink] (13) at (0.5,-3){};
  \node (14) at (1.5,-3){};
  \foreach \from/\to in {1/3,3/4,3/8,4/9,8/13,8/14}
    \draw (\from) -- (\to);
\end{scope}

\begin{scope}[shift={(3,0)}]
 \tikzstyle{every node}=[fill=black!30,circle,minimum size=0.6cm,inner sep=1pt]
 \tikzstyle{sink}=[fill=blue!50,circle,minimum size=0.6cm,inner sep=1pt]
  \node[sink] (13) at (0,0){};
  \node (17) at (0.5,-1){};
  \foreach \from/\to in {13/17}
    \draw (\from) -- (\to);
\end{scope}

\end{tikzpicture}
\caption{The forest we obtained after the decomposition. Each component can be treated as an independent sandpile prediction problem with sinks. After our decomposition, all the sinks are located in the leaf vertices.}
\label{figure:decomposed-forests}
\end{figure}
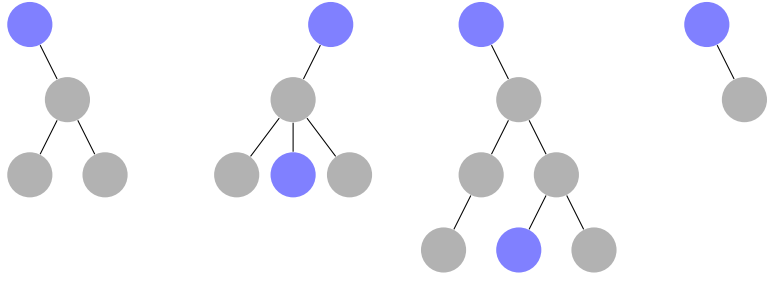

Specially, for a component, if all the vertices in the component are sink vertices, then we can remove this component directly since no firing operation could happen in this component. For the rest of the section, we will assume there's at least one non-sink vertex in each component.

By our decomposition, each component will be a tree such that the degree of every sink vertex is exactly $1$. And the number of vertices in all components is bounded by the following lemma.

\begin{lemma}
The number of the vertices in the decomposed graph is at most $2n-2$.
\end{lemma}

\begin{proof}
For each vertex $v \in V(G)$, the number of the components containing $v$ is at most $\degree(v)$. So the number of the vertices in the decomposed graph is no more than $\sum_{u \in V(G)} \degree(u) = 2|E(G)| = 2n - 2$.
\end{proof}

Thus our decomposition divides the whole graph into several trees in which only leaf vertices could be sinks while maintaining a total vertex number of $O(n)$.

Furthermore, we will root each tree in our decomposed graph. For the case in which the tree contains no more than three sinks, each component will also contain no more than three sinks. Then it's possible to choose a proper root, such that all the sinks lies in the different subtrees $\subtree(v_i)$ for $v_i \in \son(r)$. More formally:

\begin{lemma}
\label{lemma:root-each-component}
For each component, there exists a vertex $v \in V(G)$, such that if we root the tree at vertex $v$, then the sinks in the component will be located within the subtrees of different children of $v$.
\end{lemma}

\begin{proof}
If there is only one sink in this tree, we can root the tree at an arbitrary non-sink vertex.

If there are two sinks in the tree, let's denote them as $u_1$ and $u_2$. Consider the unique simple path $P$ between vertex $u_1$ and $u_2$. Since they are two leaves in the tree, there must be at least one vertex on the path $P$ excluding $u_1$ and $u_2$ (Otherwise, there will be no non-sink vertex in this component). Take any such vertex and root the tree at it, then both the sinks will be located within the subtrees of different children.

If there are three sinks in the tree, let's denote them as $u_1, u_2$ and $u_3$. Consider the unique simple path $P$ between vertex $u_1$ and $u_2$ (excluding $u_1$ and $u_2$), let $v$ be the vertex on the path with the minimum distance from the vertex $u_3$. Then $u_1, u_2, u_3$ must be located within the subtrees of different children of $v$ if we root the tree at $v$. If not, without loss of generality, assume $u_2$ and $u_3$ are all located in the subtree of $v'$ ($v' \in \son(v)$), this implies:

\begin{enumerate}
    \item $v'$ is on the unique simple path from $u_1$ and $u_2$.
    \item $v'$ is on the unique simple path from $v$ and $u_3$.
\end{enumerate}

Denote $\dist(u,v)$ as the distance between the vertex $u$ and $v$. The condition (2.) implies $\dist(v,u_3) = \dist(v,v') + \dist(v',u_3) = 1 + \dist(v',u_3)$, which means the distance between $v'$ and $u_3$ are smaller than the distance between $v$ and $u_3$. And by the condition (1.) $v'$ is also on the path from $u_1$ and $u_2$. This contradicts that $v$ is the one with the minimum distance.

\end{proof}

After decomposing the given tree into several components in \Cref{sec:treedecom}, each component becomes a tree in which every sink vertex is a leaf vertex. Since the components are independent of each other, we can decompose each component $C_i$ as a sub-instance $S_i(G_i, \sigma_i, M_i)$. By taking the root described in \Cref{lemma:root-each-component}, we can assume all the sink vertices are in the different subtrees of the children of the root.

For each $S_i(G_i, \sigma_i, M_i)$, let's consider its auxiliary graph $G'_i$. For each sink vertex $v \in M_i$, since $\degree(v) > \sum_{u \in V(G_i)} \sigma_u$, no firing operation is possible on vertex $v$. 

For any sink vertex $u$, since there's no chips on the vertex $u$ in the beginning, and $\degree(u) \geq A = n + \lvert\lvert \sigma \rvert\rvert_1$, we have $\delta(u, k) = 0$ for all $0 \leq k \leq n + \lvert\lvert \sigma \rvert\rvert_1$. The bound is sufficiently large as the value of $\delta(u, k)$ for $k > n + \lvert\lvert \sigma \rvert\rvert_1$ won't be used in our procedure. Thus we can treat $\delta(u, k)$ as always equals to $0$ for any sink vertex $u$. This is also consistent with our intuition from the original problem: no matter how many chips are put into the sink, no additional chips are returned upwards.

For any sink vertex $u$, as no firing operation occurs on $u$, we have $\delta(u, k) = 0$ for any $k\geq 0$. According to the definition of key pairs (\Cref{def:keypair}), every $(u,k)$ is a key pair for $u$. Consequently, there exists an infinite number of key pairs corresponding to a sink vertex.

\subsection{Key Pairs Maintenance with Difference}
\label{sec:difference}

To facilitate our operations, we will introduce an additional vector called $\diff_x$. Let $x$ is currently in the splay tree $D_u$, and let $k = \rank_{D_u}(x)$. Then $\diff_x$ is defined as follows:

\begin{itemize}
    \item If $k = 1$ (i.e. $x$ is the node corresponding to the pair $(u, k)$ with the minimum value of $k$), then $\diff_{x} = \moment_{x}$.
    \item Otherwise. $\diff_{x} = \moment_{x} - \moment_{\pre_{D_u}(x)}$.
\end{itemize}

In other words, $\diff_{x}$ denotes the difference of the value of $\moment$ for the node $x$ and its predecessor in the splay tree. 

\begin{lemma}
\label{lemma:diff}
For any splay tree $D_u$ and an arbitrary node $x \in D_u$, the following equation holds:
\end{lemma}
\begin{align}
\moment_x = \sum_{\substack{y \in D_u\\ \rank_{D_u}(y) \leq \rank_{D_u}(x)}} \diff_{y}\label{formula:momentequalssumofdiff}
\end{align}

\begin{proof}
We will prove the lemma by induction on $\rank_{D_u}(x)$.

If $\rank_{D_u}(x) = 1$, then by the definition of $\diff_{x}$, we have $\moment_{x} = \diff_{x}$, which is equivalent to \Cref{formula:momentequalssumofdiff}.

Now assume the induction hypothesis holds for all $\rank_{D_u}(x) < k$. For the node $x$ with $\rank_{D_u}(x) = k$, we have $\moment_{x} = \moment_{y} + \diff_{y}$, where $y = \pre_{D_u}(x)$. By the induction hypothesis, we have:

\[\moment_{y} = \sum_{\substack{z \in D_u\\ \rank_{D_u}(z) \leq k-1}} \diff_{z}\]

This implies $\moment_x = \sum_{y \in D_u \land \rank_{D_u}(y) \leq \rank_{D_u}(x)} \diff_{y}$. That means the lemma also holds for $\rank_{D_u}(x) = k$, establishing the induction step.
\end{proof}

\Cref{lemma:diff} shows that we can obtain the information for all $\moment_x$ by only maintaining $\diff_x$. Therefore, in our modified data structure, we will only maintain $\diff_x$ during the procedure. To obtain the value of $\moment_z$ for a given node $z$, we will additionally maintain $\sumdiff_z$ for every node $z$, representing the sum of $\diff_x$ for all $x \in \subtree(x)$ (or $0$ if $z = \nil$). With this, we can determine the value of $\moment_x$ and perform a binary search on the splay tree while properly maintaining the sum of $\diff_x$. More precisely:

\begin{lemma}
\label{lemma:findpre}
Let $D_u$ be a splay tree with omitted nodes and let $W$ be a given integer. There exists an algorithm to find an unomitted node $x \in D_u$ with the largest value of $\moment_{x}$ such that $\moment_{x} \leq W$, which costs an amortized time of $O(\log n)$.
\end{lemma}

\begin{proof}
We will traverse the splay tree starting from the root $r$. During the traversal, we will maintain the sum of $\diff_z$ for all $\rank_{D_u}(z) \leq \rank_{D_u}(x)$ at the time we visit node $x$, denoted as $s$.

According to \Cref{formula:momentequalssumofdiff} in \Cref{lemma:diff}, we know that:

\[\moment_{x} = \sum_{\substack{z \in D_u\\ \rank_{D_u}(z) \leq \rank_{D_u}(x)}} \diff_{z}\]

Hence, the sum of $\diff_z$ we maintain is exactly the value of $\moment_x$.

\begin{itemize}
\item If $\moment_x \leq W$, then $x$ will be one of the candidates. This implies that we only need to find a better answer among the nodes with $\moment_y \geq \moment_x$, which are located in the right subtree of node $x$. If $\rightson(x) = \nil$, we have finished our traversal. Otherwise, we need to perform the following updates:
\begin{itemize}
    \item $x \leftarrow \rightson(x)$
    \item $s \leftarrow s + \sumdiff_{\leftson(x)} + \diff_{x}$
\end{itemize}
\item Otherwise, the value of $\moment_z$ for the candidate nodes $z$ must be smaller than $\moment_x$, and they should be located in the left subtree of node $x$. If $\leftson(x) = \nil$, we have finished our traversal. Otherwise, we need to perform the following updates:
\begin{itemize}
    \item $s \leftarrow s - \sumdiff_{\leftson(x)} - \diff_{x}$
    \item $x \leftarrow \leftson(x)$
\end{itemize}
\end{itemize}

By applying the procedure above, we can maintain the value of $\moment_x$ at the time we visit it. Thus, the procedure is equivalent to accessing a node in the splay tree. According to \cite{sleator1985self}, the procedure costs an amortized $O(\log n)$ time.
\end{proof}

Since it is not possible to store an infinite number of nodes in a splay tree, we can maintain a splay tree with fewer nodes, but not necessarily full. Consequently, we omit some key pairs and include additional information on the nodes in the compact splay tree to represent the series of omitted key pairs. The key pairs corresponding to the nodes in the splay, together with the omitted key pairs expressed through information on the splay tree, precisely constitute the full set of key pairs. The $\diff_x$ still represents the difference between $\moment_x$ and $\moment_{\pre(x)}$ (or if $\rank_{D_u}(x)=1$, $\diff_x$ simply equals $\moment_x$) in the skeletal splay tree while skipping the omitted key pairs.


Specifically, for any node $x \in D_u$, we will maintain two tags on the node $x$ called $t_x$ and $d_x$. These tags represent a series of key pairs $(u, \moment_x - i \cdot d_x)$ for all $1 \leq i \leq t_x$. Additionally, we ensure that the splay tree maintains the correct ordering of the nodes. All nodes, including the omitted ones, must be ordered by $\moment_x$. Formally, for any node $x$ such that $\pre(x) \neq \nil$, the following condition must hold.

\begin{align}
\moment_{\pre(x)} \leq \moment_{x} - t_{x} \cdot d_{x} \label{formula:omittednodesareordered}
\end{align}

We will maintain an additional pair of tags $(t, d)$ on each splay tree $D_u$, representing another series of nodes in the form $\moment_z + d, \moment_z + 2d, \cdots, \moment_z + t \cdot d$, where $z$ is the node with the maximum rank in the splay tree $D_u$. During maintenance, the value of $t$ will be either $0$ or $+\infty$, depending on whether there is a sink in the subtree.

By compressing nodes in this manner, for any sink vertex $v \in M$, after the decomposition (\Cref{sec:treedecom}), $v$ will become a leaf. As a result, $D_v$ can be represented as a splay tree with satisfying:

\begin{itemize}
\item The tree contains exactly one node $x$ with $\diff_x = 1$.
\item The pair associated with the splay tree $D_v$ is $(+\infty, 1)$.
\end{itemize}

Consider the $\inctime$ modification operation described in $\updatedsupward$ and $\revertds$ (\Cref{algorithm:update} and \Cref{algorithm:revert}). We can easily modify the implementation to adapt to the way of maintaining key pairs described earlier.

\begin{lemma}
\label{lemma:sinks:inctime}
The $\inctime$ clause in $\updatedsupward$ and $\revertds$ can be modified to accommodate the key pair maintenance approach, without compromising the time and space complexity.
\end{lemma}

\begin{proof}
For a modification $\inctime(\roott(D_u), a, b)$, it will increment the value of $k$ in the key pair $(u,k)$ by $i \cdot a + b$ for all the key pairs maintained by $D_u$, where $i$ represents the rank of the key pair when they are sorted in ascending order.

Observing on the $\diff_x=\moment_x-\moment_{\pre(x)}$ in the skeletal splay tree, since there is $t_x$ omitted nodes between $\moment_x$ and $\moment_{\pre(x)}$(or $0$ if $\rank_{D_u}(x)=1$), the difference will be increased by $a\cdot (t_x+1)+b\cdot [\rank_{D_u}(x)=1]$. For the omitted key pairs, we only increase the $d_x$ by $a$, since the gap between them is expanded by $a$.

This observation indicates that the modification can be converted to performing subtree addition operations and querying the sum of a subtree. This can be done by using the same lazy propagation technique described in Section \Cref{sec:dsrotatemaintain}. To obtain a specific value of $\moment_x$ for a given node $x$, we can apply the same traversal procedure described in \Cref{lemma:findpre}.. These operations will require a total of $O(n \log n)$ time in the entire process, with $O(n)$ memory.

\end{proof}

We show that there exists a way to insert a node to a splay tree containing omitted nodes in an amortized $O(\log n)$ time.

\begin{lemma}
\label{lemma:sinks:insert}
There exists an algorithm to insert a node to a splay tree containing omitted nodes in an amortized $O(\log n)$ time.
\end{lemma}

\begin{proof}
Consider how we insert a new node $x$ to a splay tree with omitted nodes. If the splay tree is empty, we can simply let the new node as the root of the splay tree. Otherwise,  We will find the node $x_L$ with the maximum rank such that $\moment_{z} \leq \moment_{x}$. This can be done by \Cref{lemma:findpre}. And we will find $x_{R} = \suc(x_L)$ as the successor of $x_L$. Specially, if there's no such node $x_L$, then we will define $x_{L} = \nil$ and $x_{R}$ be the node with the minimum rank.

Our algorithm will insert the node to the proper position so that $\moment_{x_L} \leq \moment_{x} \leq \moment_{x_R}$ (Note that if $x_L = \nil$, then $\moment_{x_L} = 0$). However, it is possible that there are nodes $y$ whose $\moment_y$ falls within the interval $(\moment_{x_L}, \moment_{x_R})$, but were omitted and represented as $\moment_{x_R} - i \cdot d_{x_R}$ for some $1 \leq i \leq t_{x_R}$. If we perform the insertion operation directly, \Cref{formula:omittednodesareordered} might no longer hold, which breaks the order in the splay tree. In such cases, we need to break this series of omitted nodes.

\begin{itemize}
    \item If $x$ will become the node with the maximum rank after the insertion, we need to update the pair $(t, d)$ on the splay tree $D_u$.
    \begin{enumerate}
        \item Let $z$ be the node with the maximum rank in $D_u$ before we insert $x$ to the splay tree.     \item Before the insertion, there's a series of omitted nodes in the form of $\moment_{z} + i \cdot d$ for all $1 \leq d \leq t$. Since the value of $t$ could be either $0$ or $+\infty$, we can ignore the case for $t = 0$ and assume $t = +\infty$.
        \item For the nodes in the form $\moment_{z} + d$, $\moment_{z} + 2d$, $\cdots$, $\moment_{z} + Kd$ for $K = \left\lceil \frac{\moment_x - \moment_z}{d}\right\rceil-1$, we can insert a node $y$ with $\moment_{y} = \moment_{z} + Kd$, $d_{y} = d$ and $t_{y} = K-1$ to the splay tree directly, so that these nodes will be omitted on the vertex $y$
        \item For the nodes in the form $\moment_{z} + (K+1)d, \moment_{z} + (k+2)d, \cdots$, we can insert a node $y$ with $moment_{y} = \moment_{z} + (K+1)d$ and $t_{y} = d_{y} = 0$. After that, we won't need to update the pair $(t, d)$ on the splay tree at all, since $t = +\infty$ in this case.
        
    \end{enumerate}
    \item If there's no such $1 \leq i \leq t_{x_R}$ such that $\moment_{x_R} - (i+1) \cdot d_{x_R} \leq \moment_x < \moment_{x_R} - i \cdot d_{x_R}$, then we don't need to break any series of omitted nodes.
    \item Otherwise, let $j$ be the integer that satisfies $\moment_{x_R} - (j+1) \cdot d_{x_R} \leq \moment_x < \moment_{x_R} - j \cdot d_{x_R}$. In this case, we need to break the series of omitted nodes $\moment_{x_R} - i \cdot d_{x_R}$ ($1 \leq i \leq t_{x_R}$) at $i = j$. Let $A = \moment_{x_R} - (j+1) \cdot d_{x_R}$, and we need to perform the following operations:
    \begin{enumerate}
        \item Let $T = t_{x_R}$.
        \item Update the tag $t_{x_R}$ to $j$. This ensures that any node $z$ omitted at vertex $x_R$ satisfies $\moment_z \geq \moment_x$. The nodes in the form $\moment_{x_R} - i \cdot d_{x_R}$ for $j+1 \leq i \leq T$ will be lost after this update.
        \item Insert a new node $z$ into the splay tree with $\moment_{z} = \moment_{x_R} - (j+1) \cdot d_{x_R}$. Note that there will be no omitted nodes with values of $\moment$ in the interval $(\moment_{z}, \moment_{x})$, so no series of omitted nodes will be broken during this insertion.
        \item Update the tags $d_{z}$ and $t_{z}$ to $d_{x_R}$ and $T - j - 1$, respectively. After this update, any nodes in the form $\moment_{z} - i \cdot d_{z}$ will be added as omitted nodes. Since $d_{z} = d_{x_R}$ and $\moment_{z} = \moment_{x_R} - (j+1) \cdot d_{x_R}$, all these nodes (including the node $z$ itself) correspond exactly to the lost nodes mentioned in step (2). Therefore, all the omitted nodes remain unchanged after fixing the ordering.
    \end{enumerate}
\end{itemize}
Since the procedure contains at most two \insert{} operations in an ordinary splay tree, the procedure finishes in an amortized $O(\log n)$ time.
\end{proof}

And similarly, we can delete a specific node in a splay tree with the omitted nodes.

\begin{lemma}
\label{lemma:sinks:delete}
There exists an algorithm to delete a node to a splay tree containing omitted nodes in an amortized $O(\log n)$ time.
\end{lemma}

\begin{proof}
The call of a $\delete(D_u, x)$ can be divided into two cases.

\begin{itemize}
    \item If $x$ is not an omitted node in $D_u$:
    \begin{itemize}
        \item If $t_{x} = 0$, which means there are no omitted nodes represented by node $x$, we can simply delete node $x$ from the splay tree.
        \item Otherwise, we update $t_{x} \leftarrow t_{x} - 1$ and $\diff_{x} \leftarrow \diff_{x} - d_{x}$. It is straightforward to see that this update does not affect any of the omitted nodes, and node $x$ is no longer a part of the splay tree $D_u$.
    \end{itemize}
    \item If $x$ is a node omitted in the form $\moment_{y} - i \cdot d_{y}$ for $1 \leq i \leq t_{y}$ on the node $y$.
    \begin{itemize}
        \item If $i = t_{y}$, i.e. $x$ is the node with the minimum value of $\moment_{x}$ in all the nodes omitted on the node $y$. Then we can update $t_{y} \leftarrow t_{y} - 1$ and finish the deletion.
        \item Otherwise, the omitted nodes will be divided into two parts. The first part consists of nodes in the form $\moment_{y} - j \cdot d_{y}$ for $1 \leq j < i$, and the second part consists of nodes in the form $\moment_{y} - j \cdot d_{y}$ for $i < j \leq t_{y}$. To maintain all these nodes properly, we need to update $t_{y} \leftarrow i-1$ and insert a new node $z$ with $\moment_{z} = \moment_{y} - (i+1) \cdot d_{y}$, $d_{z} = d_{y}$, and $t_{z} = t_{y} - i - 1$.
    \end{itemize}
\end{itemize}
\end{proof}

The following four lemmas show that, in such a way of maintaining key pairs, interfaces used in \Cref{algorithm:main} can be modified to adapt. Proofs can be located in \Cref{sec:proofinprelim}.

\begin{lemma}[Merge and Split]
\label{general:mergeandsplit}
   $\mergeupward{}$ and $\splitds$ can be modified to have the same performance as in \Cref{lemma:merge} and \Cref{lemma:split} when $\mathcal{D}$ maintains key pairs with difference in the way described in \Cref{sec:difference}. An additional cost of $O(\log^2 n + \log \lvert\lvert \sigma \rvert\rvert_1 \log n)$ is added to the overall complexity.
\end{lemma}

\begin{lemma}[Update and Revert]
\label{general:updateandrevert}
$\updatedsupward${} and $\revertds{}$ can be modified to have the same performance as in \Cref{lemma:update} and \Cref{lemma:revert} when $\mathcal{D}$ maintains key pairs with difference in the way described in \Cref{sec:difference}.
\end{lemma}

\begin{lemma}[ComputeC and DeltaSum]
\label{general:computecanddeltasum}
$\cnt{}$ and $\deltasum{}$ can be modified to have the same performance as in \Cref{lemma:computeC} and \Cref{lemma:deltasum} when $\mathcal{D}$ maintains key pairs with difference in the way described in \Cref{sec:difference}.
\end{lemma}

\begin{lemma}[DeltaQuery]
\label{general:deltaquery}
$\deltaquery${} can be modified to have the same performance as in \Cref{lemma:deltaquery} and \Cref{lemma:deltasum} when $\mathcal{D}$ maintains key pairs with difference in the way describing in \Cref{sec:difference}.
\end{lemma}

\subsection{Algorithm}
\label{sec:generalalgo}

Now we are ready to prove the main theorem of solving sandpile prediction on trees with sinks.

\begin{proof}[Proof of \Cref{theorem:treesink}]

Firstly, let's decompose the tree using the decomposition described in \Cref{sec:treedecom}. After the decomposition, every component will be a tree such that every sink vertex is a leaf.

This converts the problem into the traditional sandpile prediction problem described in \Cref{problem:prediction} such that the given graph is a tree. We now proceed to prove that the subroutine described in \Cref{general:mergeandsplit}, \Cref{general:updateandrevert} and \Cref{general:computecanddeltasum} returns the correct value, and operates the key pairs in the same manner as described in \Cref{sec:ds}.

\begin{itemize}
\item In \Cref{general:mergeandsplit}, we proved that the $\mergeupward$ and $\splitds$ operations return the splay tree that contains all the key pairs in $D_v$ for all $v \in \son(v)$. Thus, they exhibit the same behavior as the merge and split operations described in \Cref{lemma:merge} and \Cref{lemma:split}.
\item In \Cref{general:updateandrevert}, we consider implementing each clause in the $\updatedsupward$ and $\revertds$ operations under the current circumstances as described in \Cref{sec:ds}. Therefore, this would accurately replicate the original version.
\item In \Cref{general:computecanddeltasum}, we established the correctness of the result and the consistency of the information.
\end{itemize}

Combining the consistent results and the correctness established in \Cref{general:mergeandsplit}, \Cref{general:updateandrevert} and \Cref{general:computecanddeltasum}, and applying the analysis found in \Cref{sec:ds}, we can obtain the final result.






\end{proof}

\begin{remark}[Sandpile Prediction on Path with Sinks]
The time complexity can be improved to $O(n)$ if the given graph is $\text{Path}_n$. We can modify algorithms in \Cref{sec:path} in a similar way.
\end{remark}




\section{Splay Trees Maintenance}
\subsection{Rotation}
\label{sec:dsrotatemaintain}
The fundamental operation of the splay tree is $\splay$. A $\splay(x)$ operation is called whenever we access any node $x$. By \cite{sleator1985self} splaying the node after we access it will give us amortized $O(\log n)$ time complexity for inserting, deleting and searching. In $\splay(x)$, we will make node $x$ to the root of the splay tree while maintaining the in-order traverse of the tree unchanged. The way to do this is by performing a series of \textit{splay steps} \cite{sleator1985self}. Each splay step might contain one or two rotations of $x$ and its current $\father(x)$, which moves $x$ closer to the root while the overall in-order sequence remains unchanged. More precisely, let $x$ be an arbitrary non-root node in $D_u$ and $y = \father(x)$. a single $\rotate(x)$ will make $x$ become the father of $y$ and $y$ become one of the two children of $x$. 

We analyze here that as long as we call $\pushup$ and $\pushdown$ at the proper timing during each splay step, we are able to maintain $\timemin$, $\timemax$ correctly and the lazy propagation mechanism is still correct. Specifically, when a rotation happens, it might cause:

\begin{itemize}
    \item $\timemin_x$ and $\timemax_x$ should be recalculated base on the current $\subtreeT$.

    \item $\ta_x$ and $\tb_x$ might take effect on the wrong set of nodes.
\end{itemize}


We design $\pushdown(x)$ described in \Cref{algorithm:pushdown} to modify $\leftson(x)$ and $\rightson(x)$'s information according to $\ta_x$ and $\tb_x$. After that, we clear up these two values. In this way, we ensure that these lazy tags are always taking effect on the correct set of nodes.

\IncMargin{1em}
\begin{algorithm}
  \SetKwData{Left}{left}\SetKwData{This}{this}\SetKwData{Up}{up}
  \SetKwComment{Comment}{$\triangleright$\ }{}
  \SetKwFunction{SolveUpward}{\upward}
  \SetKwFunction{SolveDownward}{\downward}
  \SetKwFunction{CalculateNumberOfFirings}{\cnt}
  \SetKwFunction{UpdateDataStructure}{\updatedsupward}
  \SetKwFunction{MergeDataStructure}{\mergeupward}
  \SetKwFunction{DeltaSum}{\deltasum}
  \SetKwFunction{InitializeDataStructure}{\initds}
  \SetKwFunction{DeltaQuery}{\deltaquery}
  \SetKwFunction{RevertDataStructure}{\revertds}
  \SetKwFunction{SplitDataStructure}{\splitds}
  \SetKwInOut{Input}{input}\SetKwInOut{Output}{output}  
  \SetKwFunction{FSolveDownward}{SolveDownward}
  \SetKwFunction{IncTime}{IncTime}

  \SetKwProg{Fn}{Function}{:}{}

  \If{$\leftson(u) \ne \nil$}{
    $\IncTime(\leftson(u), \ta_u, \tb_u)$\;
  }
  \If{$\rightson(u) \ne \nil$}{
    $\IncTime(\rightson(u), \ta_u, \tb_u + \left(\size(\leftson(u)) + 1\right) \cdot \ta_u)$\;
  }
  \BlankLine

  \caption{\pushdown($u$)}\label{algorithm:pushdown}
\end{algorithm}\DecMargin{1em}

After $x$ is involved in a rotation, it is likely that $\leftson(x)$ and $\rightson(x)$ will change, thus $\timemin$ and $\timemax$ will need to be recalculated since they should denote the minimum and maximum $\timestamp$ value of the current subtree of $x$.

We design $\pushup(x)$ described in \Cref{algorithm:pushup} to recalculate $\timemin_x$ and $\timemax_x$ based on the information of its children for any node $x \in D_u$. To prove the correctness, we concentrate on the change of $x$'s corresponding $\leftson(x)$, $\rightson(x)$, assuming $\leftson(x)$ and $\rightson(x)$'s subtree structure remains unchanged, thus having the correct value of $\timemin$ and $\timemax$.

\IncMargin{1em}
\begin{algorithm}
  \SetKwData{Left}{left}\SetKwData{This}{this}\SetKwData{Up}{up}
  \SetKwComment{Comment}{$\triangleright$\ }{}
  \SetKwFunction{SolveUpward}{\upward}
  \SetKwFunction{SolveDownward}{\downward}
  \SetKwFunction{CalculateNumberOfFirings}{\cnt}
  \SetKwFunction{UpdateDataStructure}{\updatedsupward}
  \SetKwFunction{MergeDataStructure}{\mergeupward}
  \SetKwFunction{DeltaSum}{\deltasum}
  \SetKwFunction{InitializeDataStructure}{\initds}
  \SetKwFunction{DeltaQuery}{\deltaquery}
  \SetKwFunction{RevertDataStructure}{\revertds}
  \SetKwFunction{SplitDataStructure}{\splitds}
  \SetKwInOut{Input}{input}\SetKwInOut{Output}{output}  
  \SetKwFunction{FSolveDownward}{SolveDownward}

  \SetKwProg{Fn}{Function}{:}{}

    $\timemin_x \leftarrow \timestamp_x$\;
    $\timemax_x \leftarrow \timestamp_x$\;
    \If{$\leftson(x) \ne \nil$}{
        $\timemin_x \leftarrow \min(\timemin_x, \timemin_{\leftson(x)})$\;
        $\timemax_x \leftarrow \max(\timemax_x, \timemax_{\leftson(x)})$\;
    }
    \If{$\rightson(x) \ne \nil$}{
        $\timemin_x \leftarrow \min(\timemin_x, \timemin_{\rightson(x)})$\;
        $\timemax_x \leftarrow \max(\timemax_x, \timemax_{\rightson(x)})$\;
    }
  
  \BlankLine

  \caption{\pushup($x$)}\label{algorithm:pushup}
\end{algorithm}\DecMargin{1em}

Let $L = \{ y \mid y \in \subtreeT(\leftson(x))\}$ and $R = \{ y \mid y \in \subtreeT(\rightson(x))\}$. Specially, if $\leftson(x) = \nil$, then $L = \varnothing$ (and similarly for $R$). Then $\subtreeT(x) = L \cup R \cup \{x\}$. We only analyze $\timemin$ here since $\timemax$ shares the same transition logic:

\begin{align*}
\min_{y \in \subtreeT(x)} \timestamp_y &= \min(\min_{y \in L} \timestamp_y, \min_{y \in R} \timestamp_y, \timestamp_x) \\
&= \min(\timemin_{\leftson(x)}, \timemin_{\rightson(x)}, \timestamp_x)
\end{align*}

Terms are ignored when $\leftson(x)$ or $\rightson(x)$ is $\nil$. \Cref{algorithm:pushup} consists of the exact transition as the aggregation from $\leftson(x)$ and $\rightson(x)$'s result, which compute the correct value of $\timemin_x$ and $\timemax_x$ as result.

The key operation to change the structure of the splay tree is $\rotate(x)$, as we mentioned in \Cref{figure:zig}. For any non-root node $x$, the $\rotate(x)$ procedure will make the parent of the node $x$ become the child of $x$, which means the distance of the node $x$ to the root is decreased by exactly one. 

Assume $y$ is the parent of the node $x$ before the rotation. After calling $\rotate(x)$, the set of the nodes corresponding to the vertex $x$ and $y$ are all changed, so we need to perform $\pushup$ operation on them. Note that $y$ is one of the children of $x$, thus we need to update the information of $y$ before updating the node $x$, so we need to call $\pushup(y)$ and $\pushup(x)$ in order every time we finish a $\rotate(x)$.

We will show the following example of how the push-up performs on the $\zig$, $\zigzag$ and $\zigzig$ operation in the splay tree, which was mentioned in \cite{sleator1985self}.

\begin{itemize}
    \item The $\zig(x)$ function is called once $x$ is not the root node of the splay but $\father(x)$. It will perform a $\rotate$ operation on vertex $x$, so that $x$ becomes the root of the splay. The \Cref{figure:zig} shows the procedure of $\zig(x)$. So after $\zig(x)$, we need to call $\pushup(y)$ and then $\pushup(x)$. The details of the operation are described in \Cref{figure:zig}.
\begin{figure}[h]
\centering
\begin{tikzpicture}
\begin{scope}[shift={(-3,0)}]
 \tikzstyle{every node}=[fill=black!30,circle,minimum size=0.6cm,inner sep=1pt]

  \node (1) at (0,0) {$y$}; 
  \node (2) at (-0.5,-1) {$a$}; 
  \node[special] (3) at (0.5,-1) {$x$}; 
  \node (4) at (0,-2) {$b$}; 
  \node (5) at (1,-2) {$c$}; 

  \foreach \from/\to in {1/2,1/3,3/4,3/5}
    \draw[edge] (\from) -> (\to);
\end{scope}

\begin{scope}[shift={(0,0)}]
\draw[line width=2pt,->] (-1.25,-1) -- (1.25,-1);
\node at (0,-0.5){$\rotate(x)$};
\node at (0,-1.5){$\pushup(y)$};
\node at (0,-2){$\pushup(x)$};
\end{scope}

\begin{scope}[shift={(3,0)}]
 \tikzstyle{every node}=[fill=black!30,circle,minimum size=0.6cm,inner sep=1pt]

  \node (1) at (-0.5,-1) {$y$}; 
  \node (2) at (0,-2) {$b$}; 
  \node[special] (3) at (0,0) {$x$}; 
  \node (4) at (-1,-2) {$a$}; 
  \node (5) at (0.5,-1) {$c$}; 

  \foreach \from/\to in {3/1,3/5,1/2,1/4}
    \draw[edge] (\from) -> (\to);
\end{scope}
\end{tikzpicture}
\caption{The figure corresponds to a call of $\zig(x)$, which is simply rotating the node $x$.}
\label{figure:zig}
\end{figure}
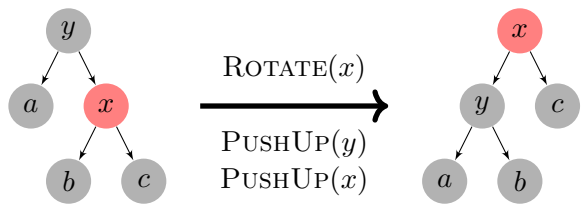
    \item The $\zigzig(x)$ function is called once $x$ and $y = \father(x)$ is not the root node of the splay, and they are both left (or right) children. In the procedure, we will rotate the node $y$ and then rotate the node $X$. After the rotation $y$ becomes the parent of $z$ and $x$ becomes the parent of $x$. The details of the operation are described in \Cref{figure:zigzig}.
\begin{figure}[h]
\centering
\begin{tikzpicture}
\begin{scope}[shift={(-3.5,0)}]
 \tikzstyle{every node}=[fill=black!30,circle,minimum size=0.6cm,inner sep=1pt]

  \node (0) at (0,0) {$z$}; 
  \node (1) at (-0.5,-1) {$y$}; 
  \node[special] (2) at (-1,-2) {$x$}; 
  \node (3) at (-1.5,-3) {$a$}; 
  \node (4) at (-0.5,-3) {$b$}; 
  \node (5) at (0,-2) {$c$}; 
  \node (6) at (0.5,-1) {$d$}; 

  \foreach \from/\to in {0/1,1/2,2/3,2/4,1/5,0/6}
    \draw[edge] (\from) -> (\to);
\end{scope}

\begin{scope}[shift={(-0.5,0)}]
\draw[line width=2pt,->] (-2,-1) -- (1,-1);
\node at (-0.5,-0.5){$\rotate(y)$};
\node at (-0.5,-1.5){$\pushup(z)$};
\node at (-0.5,-2){$\pushup(y)$};
\end{scope}

\begin{scope}[shift={(2.5,0)}]
 \tikzstyle{every node}=[fill=black!30,circle,minimum size=0.6cm,inner sep=1pt]
  \node (y) at (0,0) {$y$}; 
  \node[special] (x) at (-1,-1) {$x$}; 
  \node (z) at (1,-1) {$z$}; 
  \node (a) at (-1.5,-2) {$a$}; 
  \node (b) at (-0.5,-2) {$b$}; 
  \node (c) at (0.5,-2) {$c$}; 
  \node (d) at (1.5,-2) {$d$}; 
  \foreach \from/\to in {y/x,y/z,x/a,x/b,z/c,z/d}
    \draw[edge] (\from) -> (\to);
\end{scope}

\begin{scope}[shift={(5.5,0)}]
\draw[line width=2pt,->] (-1.5,-1) -- (1.5,-1);
\node at (0,-0.5){$\rotate(x)$};
\node at (0,-1.5){$\pushup(y)$};
\node at (0,-2){$\pushup(x)$};
\end{scope}

\begin{scope}[shift={(8,0)}]
 \tikzstyle{every node}=[fill=black!30,circle,minimum size=0.6cm,inner sep=1pt]
  \node[special] (x) at (0,0) {$x$}; 
  \node (y) at (0.5,-1) {$y$}; 
  \node (z) at (1,-2) {$z$}; 
  \node (a) at (-0.5,-1) {$a$}; 
  \node (b) at (0,-2) {$b$}; 
  \node (c) at (0.5,-3) {$c$}; 
  \node (d) at (1.5,-3) {$d$}; 
  \foreach \from/\to in {x/y,y/z,x/a,y/b,z/c,z/d}
    \draw[edge] (\from) -> (\to);
\end{scope}

\end{tikzpicture}
\caption{The figure corresponds to a call of $\zigzig(x)$}
\label{figure:zigzig}
\end{figure}
    \item The $\zigzag(x)$ function is called once $x$ and $\father(x)$ is not the root node of the splay, and $x$ is a left child and $\father(x)$ is a right child, or vice-versa. In the procedure, we will rotate the node $x$ twice. After the rotation $x$ will be the parent of the node $y$ and $z$. The details of the operation are described in \Cref{figure:zigzag}.
\begin{figure}[h]
\centering
\begin{tikzpicture}
\begin{scope}[shift={(-3.5,0)}]
 \tikzstyle{every node}=[fill=black!30,circle,minimum size=0.6cm,inner sep=1pt]

  \node (z) at (0,0) {$z$}; 
  \node (y) at (-0.5,-1) {$y$}; 
  \node (d) at (0.5,-1) {$d$}; 
  \node (a) at (-1,-2) {$a$}; 
  \node[special] (x) at (0,-2) {$x$}; 
  \node (b) at (-0.5,-3) {$b$}; 
  \node (c) at (0.5,-3) {$c$}; 
  \foreach \from/\to in {z/y,z/d,y/a,y/x,x/b,x/c}
    \draw[edge] (\from) -> (\to);

\end{scope}

\begin{scope}[shift={(-0.5,0)}]
\draw[line width=2pt,->] (-2,-1) -- (1.25,-1);
\node at (-0.5,-0.5){$\rotate(x)$};
\node at (-0.5,-1.5){$\pushup(y)$};
\node at (-0.5,-2){$\pushup(x)$};
\end{scope}

\begin{scope}[shift={(2.5,0)}]
 \tikzstyle{every node}=[fill=black!30,circle,minimum size=0.6cm,inner sep=1pt]
  \node (z) at (0,0) {$z$}; 
  \node[special] (x) at (-0.5,-1) {$x$}; 
  \node (d) at (0.5,-1) {$d$}; 
  \node (y) at (-1,-2) {$y$}; 
  \node (c) at (0,-2) {$c$}; 
  \node (a) at (-1.5,-3) {$a$}; 
  \node (b) at (-0.5,-3) {$b$}; 
  \foreach \from/\to in {z/x,z/d,x/y,x/c,y/a,y/b}
    \draw[edge] (\from) -> (\to);
\end{scope}

\begin{scope}[shift={(5,0)}]
\draw[line width=2pt,->] (-1.5,-1) -- (1.5,-1);
\node at (0,-0.5){$\rotate(x)$};
\node at (0,-1.5){$\pushup(z)$};
\node at (0,-2){$\pushup(x)$};
\end{scope}

\begin{scope}[shift={(8,0)}]
 \tikzstyle{every node}=[fill=black!30,circle,minimum size=0.6cm,inner sep=1pt]
  \node[special] (x) at (0,0) {$x$}; 
  \node (y) at (-1,-1) {$y$}; 
  \node (z) at (1,-1) {$z$}; 
  \node (a) at (-1.5,-2) {$a$}; 
  \node (b) at (-0.5,-2) {$b$}; 
  \node (c) at (0.5,-2) {$c$}; 
  \node (d) at (1.5,-2) {$d$}; 
  \foreach \from/\to in {x/y,x/z,y/a,y/b,z/c,z/d}
    \draw[edge] (\from) -> (\to);
\end{scope}

\end{tikzpicture}
\caption{The figure corresponds to a call of $\zigzag(x)$}
\label{figure:zigzag}
\end{figure}
\end{itemize}

Finally, we need to consider when we should perform the $\pushdown{}$ operation. Once we traverse the splay tree from the root, we need to push down all the tags from the root to the current node. This implies every time we need to search a specific node, we need to push down all the nodes from the root to that node in order. This includes the \Insert{}, \Delete{} and \splay{} function in the splay tree. Furthermore, for the function \findmin{} (described in \Cref{algorithm:findmin}), \deltaquery{} (described in \Cref{algorithm:deltaquery}), \findoneintree{} (described in \Cref{algorithm:findoneintree}) and \findonebeforetree{} (described in \Cref{algorithm:findonebeforetree}), it is equivalent to search a specific node in the splay tree. Since we will splay the node we searched in all these functions, the tags will be pushed down correctly in the \splay{} procedure.

\subsection{Merging by Small-To-Large Technique}
\label{sec:mergeds}
We will analyze the $\mergeupward$ (\Cref{algorithm:merge}) operation by proving the following lemma. 

\mergelemma*

\IncMargin{1em}
\begin{algorithm}[H]
  \SetKwData{Left}{left}\SetKwData{This}{this}\SetKwData{Up}{up}
  \SetKwFunction{swap}{Swap}
  \SetKwFunction{insert}{insert}

  \SetKwProg{Fn}{Function}{:}{}
  \BlankLine

  \If{$\size(D_u)< \size(D_v)$}{ \label{merge:if}
    $\res_v\leftarrow 1$\; \label{merge:bool1}
    \swap($D_u$,$D_v$)\; \label{merge:swap}
  }
  \Else{$\res_v\leftarrow 0$\;} \label{merge:bool2}

  \For{$x \in D_v$ by increasing order}{ \label{merge:for}
    \insert($D_u$, $x$)\; \label{merge:insert}
  }
  \caption{\mergeupward($u$, $v$)}\label{algorithm:merge}
\end{algorithm}\DecMargin{1em}

When merging two splay trees $D_u$ and $D_v$, we will need to use the classic small-to-large technique to merge them efficiently.

The small-to-large technique is easy to describe: when we want to merge two splay trees $D_u$ and $D_v$ , we always insert all nodes from the splay tree with a smaller size to the one with a larger size one by one in ascending order (\Cref{merge:insert}).

In \Cref{algorithm:merge}, we need to guarantee that $D_u$ contains all nodes from $D_v,v \in \son(u)$ after merging. Therefore, we will simply swap $D_u$ and $D_v$ (\Cref{merge:swap}) if $\size(D_u)<\size(D_v)$ (\Cref{merge:if}). Here we store a boolean value $\res_v$ to keep track of such swapping (\Cref{merge:bool1} and \Cref{merge:bool2}).

The analysis is also very trivial to reach a $O(n \log^2 n)$ upper bound: We can see that the time cost of all $\mergeupward$ operations is equivalent to merging a series of splay trees into one. Note that every $\mergeupward$'s time cost is proportional to the size of the smaller splay. In the next $\mergeupward$, the size value of the smaller splay is at most doubled than the previous one. Therefore, for each node, it will be inserted into other splay trees at most $O(\log n)$ times. Since the insertion on splay trees is $O(\log n)$ amortized, it costs $O(n\log^2 n)$ time in total. However, if we guarantee that all nodes are inserted in increasing order during the process (\Cref{merge:for}), we will be able to trigger \Cref{theorem:dynamic-finger-theorem} to reach a better bound.

\begin{corollary}
\label{corollary:seriessplay}
Given a series of splay trees: $T_1,T_2,\cdots,T_k$ such that $\sum_{i=1}^{k}\size(T_i)=n$, if we call $\mergeupward$ $k-1$ times in arbitrary order to merge them into one, the total time cost for all $\mergeupward$ operations is $O(n \log n)$. Note that between any two $\mergeupward$ operations, it is allowed to have other operations on splay trees separately.
\end{corollary}

We will need a modified version of \Cref{theorem:dynamic-finger-theorem}:
\begin{theorem}[\cite{brodal2018finger}]
\label{theorem:merge-two-set}
The total time to perform $n$ insertions on a splay tree of size $m$ is $O(n\log \frac{m+n}{n})$ if the insertions are performed on items in increasing order of ranks.
\end{theorem}

Now we are ready to prove \Cref{lemma:merge}. 
\begin{proof}[Proof of \Cref{corollary:seriessplay}]
The goal is to proof the overall complexity is no more than $C\cdot n\log n$, where $C$ is a deterministic constant. We can do the mathematical induction here. For $k=1$, it is obviously true. Now assuming it is true for $k=1 \cdots i-1$, now we look at the case where $k=i$. Focusing on the last $\mergeupward(A,B)$ operation we performed where $A$ and $B$ are splay trees, we have $\size(A)+\size(B) = n$. By the inductive hypothesis, we know that the previous merging processes cost $C\cdot (\size(A)\log \size(A)+\size(B)\log \size(B))$ time. Without loss of generality, we assume $\size(A)>\size(B)$. Therefore, we will insert the nodes of $B$ to $A$ one by one in increasing order. By \Cref{theorem:merge-two-set}, this costs at most $C_1\cdot \size(B)\log \frac{n}{\size(B)}$ time, where $C_1$ is a constant. Taking $C\ge C_1$, we get $C\cdot \size(A)\log \size(A)+C\cdot \size(B)\log \size(B)+C_1\cdot \size(B)\log(\frac{n}{\size(B)})\le C \cdot \size(A)\log n+C\cdot \size(B)\log \size(B)+C\cdot \size(B)\log(\frac{n}{\size(B)})=C\cdot n\log n$, thus we complete the induction.
\end{proof}

By \Cref{corollary:seriessplay} and previous analysis, \Cref{lemma:merge} is proved.

\subsection{Splitting by Undoing Merges}
\label{sec:splitds}
We will analyze the $\splitds$ (\Cref{algorithm:nonlocal}) operation by proving the following lemma. 

\splitlemma*

\IncMargin{1em}
\begin{algorithm}[H]
  \SetKwData{Left}{left}\SetKwData{This}{this}\SetKwData{Up}{up}

  \SetKwProg{Fn}{Function}{:}{}
  \BlankLine

  \While{$\True$}{\label{split:while}
    $x\leftarrow \nil$\;\label{split:init}
    \If{$\res_v=0$}{\label{split:ifa}
      $x\leftarrow \findoneA(u,v)$\;\label{split:finda}
    }
    \Else{
      $x\leftarrow \findoneB(u,v)$\; \label{split:findb}
    }
    \If{$x=\nil$}{\Break\;} \label{split:breakloop}
    $\delete(D_u,x)$\; \label{split:delete}
    $\insert(D_v,x)$\; \label{split:insert}
  }
  \If{$\res_v=1$}{ \label{split:ifb}
    $\swap(D_u,D_v)$\;\label{split:swap}
  }
  
  \caption{\splitds($u$,$v$)}\label{algorithm:nonlocal}
\end{algorithm}\DecMargin{1em}

The purpose of $\splitds(u,v)$ is to derive $D_v$ back to the previous state. Based on the assumption that the current $D_u$ contains nodes from the previous $D_u$ and $D_v$ only before $\mergeupward(u,v)$, we can do this by finding nodes from $D_v$ in increasing order on $D_u$. In this way, the process can be regarded as deleting and inserting nodes in increasing order in both $D_u$ and $D_v$, which is a symmetric process to $\mergeupward(u,v)$, thus sharing the same overall complexity.

Since $\mergeupward(u,v)$ follows the small-to-large mechanism, it might execute a $\textsc{Swap}(D_u,D_v)$. Since we always plan on merging $D_v$ into $D_u$, we need two similar functions to find the minimum rank node belonging to $D_v$ or $D_u$. Here we refer to $D_u$ and $D_v$ before the possible swapping. These two functions are described in \Cref{algorithm:findoneintree} and \Cref{algorithm:findonebeforetree}.

\begin{lemma}
\label{lemma:splitfind}
$\findoneintree(u,v)$ is able to find the minimum rank node belongs to the original $D_v$ and $\findonebeforetree(u,v)$ is able to find the minimum rank node belongs to the original $D_u$.
\end{lemma}

\IncMargin{1em}
\begin{algorithm}[H]
  \SetKwData{Left}{left}\SetKwData{This}{this}\SetKwData{Up}{up}
  \SetKwFunction{PushDown}{PushDown}

  \SetKwProg{Fn}{Function}{:}{}
  \BlankLine
  $x\leftarrow \roott(D_u)$\; \label{findoneA:takeroot}

  \If{$x=\nil \text{ or } \timemax_x<\dfn_v$}{ \label{findoneA:rootcheck}
    \Return $\nil$\;
  }

  \While{$\True$}{ \label{findoneA:while}
    \If{$\leftson(x)\neq \nil $\text{ and }$\timemax_{\leftson(x)}\ge \dfn_v$}{\label{findoneA:ifa}
      $x\leftarrow \leftson(x)$\;\label{findoneA:turnleft}
    }
    \ElseIf{$\timestamp_x\ge \dfn_v$}{\label{findoneA:ifb}
      \Break\;\label{findoneA:brk}
    }
    \Else{
      $x\leftarrow \rightson(x)$\;\label{findoneA:turnright}
    }
  }

  $\splay(x)$\; \label{findoneA:splay}

  \Return $x$\; \label{findoneA:return}
  
  \caption{\findoneA($u$, $v$)}\label{algorithm:findoneintree}
\end{algorithm}\DecMargin{1em}

\IncMargin{1em}
\begin{algorithm}[H]
  \SetKwData{Left}{left}\SetKwData{This}{this}\SetKwData{Up}{up}
  \SetKwFunction{PushDown}{PushDown}

  \SetKwProg{Fn}{Function}{:}{}
  \BlankLine
  $x\leftarrow \roott(D_u)$\; \label{findoneB:takeroot}

  \If{$x=\nil \text{ or } \timemin_x\ge\dfn_v$}{
    \Return $\nil$\;
  }

  \While{$\True$}{ \label{findoneB:while}
    \If{$\leftson(x)\neq \nil $\text{ and }$\timemin_{\leftson(x)}< \dfn_v$}{
      $x\leftarrow \leftson(x)$\;
    }
    \ElseIf{$\timestamp_x< \dfn_v$}{
      \Break\;
    }
    \Else{
      $x\leftarrow \rightson(x)$\;
    }
  }

  $\splay(x)$\; \label{findoneB:splay}

  \Return $x$\; \label{findoneB:return}
  
  \caption{\findoneB($u$, $v$)}\label{algorithm:findonebeforetree}
\end{algorithm}\DecMargin{1em}

\begin{proof}[Proof of \Cref{lemma:splitfind}]
Without losing the generality, we can only analyze $\findoneintree(u,v)$ here and the analysis of $\findonebeforetree(u,v)$ is almost the same.

Determining if a node $x$ in $D_u$ comes from the original $D_v$ is equivalent to determining whether $p$ is in $\subtree(v)$, where $p$ is the tree vertex that node $x$ is generated by $\newnode$ during the execution of $\updatedsupward(p)$.

Such verification can be done by comparing the $\timestamp_x$ and $\dfn_v$. Since we assume that no key pair from $D_{v'}$ exists if $v'$ is after $v$ in $\mathcal{I}$, we can determine if $x$ is in the original $D_v$ by checking if $\timestamp_x \ge \dfn_v$ holds. If the inequality holds, it means that $x$ is generated no earlier than visiting vertex $v$. Therefore, it must come from $\subtree(v)$.

Now we further generalize this condition to a tree walk on the splay tree. We need to check if there is any node $y$ in $\subtreeT(x)$ satisfying this condition. Since we maintain $\timemax$ for every splay tree node, we can determine this by checking if $\timemax_x \geq \dfn_v$ holds.

\Cref{algorithm:findoneintree} is a tree walk supported by the above verification. We first initialize the $x$ by $\roott(D_u)$ (\Cref{findoneA:takeroot}). If $x$ is empty or $x$ does not satisfy the inequality (\Cref{findoneA:rootcheck}), there is no such node that exists in the entire $D_u$, thus, we return $\nil$.

Since we are finding the one with the smallest rank, we first check if $\leftson(x)$ satisfies this condition (\Cref{findoneA:ifa}). If so, we go to $\leftson(x)$ (\Cref{findoneA:turnleft}) and continue the walking. Otherwise, we check if the current node $x$ satisfies (\Cref{findoneA:ifb}). If so, we find the one with the minimum rank, and thus we exit the loop (\Cref{findoneA:brk}). If both verifications fail, since we already determine there is at least one node satisfying the condition, we go to $\rightson(x)$ directly (\Cref{findoneA:turnright}) and continue walking.

After finding the desired node $x$, we will need to call $\splay(x)$ to guarantee the amortized access cost (\Cref{findoneA:splay}).

In $\findonebeforetree$ (\Cref{algorithm:findonebeforetree}), the subtree verification condition becomes $\timemin_x \geq \dfn_v$ for any $\subtreeT(x)$ on the splay tree. The rest analysis remains the same as above.
\end{proof}

Now we are ready to prove \Cref{lemma:split}.

\begin{proof}
In $\mergeupward(u,v)$, we store $\res_v$ to keep track of whether the swapping occurs. Therefore, by the value of $\res_v$, we can tell if we use $\findoneintree$ (\Cref{split:finda}) or $\findonebeforetree$ (\Cref{split:findb}) to find the minimum rank node in $D_u$ that belongs to the original $D_v$. By \Cref{lemma:splitfind} we know it will return the correct node if it exists. We use a variable $x$ to store the search result (\Cref{split:init}). If $x=\nil$, then all nodes are found, and thus we exit the loop (\Cref{split:breakloop}). Every time we find a node $x$, we will delete it from $D_u$ (\Cref{split:delete}) and insert it back to $D_v$ (\Cref{split:insert}). In the end, $D_v$ will be restored and no node from $D_u$ belongs to $D_v$. Lastly, we also need to swap back to revert the previous swapping in $\mergeupward$ if necessary (\Cref{split:ifb} and \Cref{split:swap}).

One can notice that we are actually splitting nodes in the same \textit{small-to-large} idea as $\mergeupward$. Moreover, it is exactly the symmetric of nodes' insertions in $\mergeupward(u,v)$. Since on a splay tree, both $\Insert$ and $\Delete$ have the dynamic finger property. We can derive a similar theorem to \Cref{theorem:merge-two-set} that the total time to perform $n$ deletions on a splay tree of size $m$ is $O(n\log \frac{m+n}{n})$. Same as \Cref{lemma:merge}, we can eventually derive the total time cost for all $\splitds$ as $O(n \log n)$ as for all $\mergeupward$ operations.
\end{proof}

\section{Omitted Proofs}
\label{sec:proofinprelim}
To prove \Cref{lemma:unique-terminal-configuration}, we will first give the following lemmas.

\begin{lemma}
\label{lemma:fire-commutative}
Let $S(G, \sigma)$ be a given sandpile instance. For two distinct vertices $u,v \in V(G)$ ($u \ne v$), if $\sigma_u \geq \degree(u)$ and $\sigma_v \geq \degree(v)$, then: 
\begin{enumerate}
    \item It is possible to fire the vertex $u$ and then fire the vertex $v$.
    \item It is possible to fire the vertex $v$ and then fire the vertex $u$.
    \item Both order of firing vertex $u, v$ obtains the same configuration. That is, $\fire(\fire(\sigma, u), v) = \fire(\fire(\sigma, v), u)$
\end{enumerate}

\end{lemma}

\begin{proof}[Proof of \Cref{lemma:fire-commutative}]
By $\sigma_u \geq \degree(u)$ and $\sigma_v \geq \degree(v)$, we know that $\sigma^{(u)} = \fire(\sigma, u)$ and $\sigma^{(v)} = \fire(\sigma, v)$ both exist. 

Note that $\sigma^{(u)}_v \geq \sigma_v \geq \degree(v)$, because firing vertex $u$ won't decrease the number of the chips on all other vertices. Similarly we have $\sigma^{(v)}_u \geq \sigma_u \geq \degree(u)$, so $\fire(\sigma^{(u)}, v)$ and $\fire(\sigma^{(v)}, u)$ both exist.

Since both of the configuration exist, we have $\fire(\fire(\sigma, u), v) = \sigma + F(u) + F(v)$ and $\fire(\fire(\sigma, v), u) = \sigma + F(v) + F(u)$. By the commutative property of the vector addition, $\sigma + F(v) + F(u) = \sigma + F(u) + F(v)$, which proves that both of the configurations are equal.
\end{proof}

\begin{lemma}
\label{lemma:fire-commutative-ex}
Let $S(G, \sigma)$ be a given sandpile instance. Suppose it is possible to fire the vertices $u_1, u_2, \cdots, u_t$ in order and obtain another configuration $\sigma'$. Then for any $2 \leq j \leq t$ satisfying $u_k \ne u_j$ for all $1 \leq k < j$, the following conditions are hold

\begin{itemize}
    \item It is possible to fire the vertices $u_j, u_1, u_2, \cdots u_{j-1}, u_{j+1}, u_{j+2}, \cdots, u_t$ in order.
    \item The configuration obtained by firing the vertices $u_j, u_1, u_2, \cdots u_{j-1}, u_{j+1}, u_{j+2}, \cdots, u_t$ in order equals to $\sigma'$.
\end{itemize}
\end{lemma}

\begin{proof}[Proof of \Cref{lemma:fire-commutative-ex}]
Consider the original firing sequence $u_1, u_2, \cdots, u_{j-1}, u_j, u_{j+1}, \cdots, u_t$. Since $u_{j-1} \ne u_j$, by \Cref{lemma:fire-commutative} we can swap the order of the vertex $u_{j-1}$ and $u_j$. After that, the vertex fired before the $u_j$ is $u_{j-2}$, which $u_{j-2} \ne u_j$ also holds since $u_j \ne u_k$ holds for all $1 \leq k < j$. So we can swap $u_{j-2}$ and $u_j$ again. Repeatably swap $u_j$ with the previous firing vertex until $u_j$ becomes the first vertex to be fired. In the end we will fire the vertices $u_j, u_1, u_2, \cdots u_{j-1}, u_{j+1}, u_{j+2}, \cdots, u_t$ in order, so it's possible to fire the vertices in this order while not changing the obtained configuration.
\end{proof}

The configuration addition and fire operation give us the following corollary.


\begin{corollary}
\label{lemma:add-then-fire}
Let $S(G, \sigma)$ and $S(G, \sigma')$ be two sandpile instances. For any vertex $u \in V(G)$, if $\sigma_u \geq \degree(u)$, then $\fire(\sigma + \sigma', u) = \fire(\sigma, u) + \sigma'$.
\end{corollary}

\begin{proof}[Proof of \Cref{lemma:unique-terminal-configuration}]
Assume there are arbitrary two sequence of vertices $u_1, u_2, \cdots, u_{a}$ and $v_1, v_2, \cdots, v_{b}$, such that we can get a terminal configuration $\sigma^{(u)}$ by firing vertex $u_1, u_2, u_{a}$ in order, and a terminal configuration $\sigma^{(v)}$ by firing vertex $v_1, v_2, \cdots, v_{b}$ in order. We will prove $\sigma^(u)$ must equals to $\sigma^(v)$.

We can show that by mathematical induction. Let $k = \max(a, b)$. The lemma is obviously correct for $k = 0$, as they both equal to $\sigma$.

Otherwise, consider the vertex $u_1$ and $v_1$, there are two different cases:

\begin{itemize}
    \item If $u_1 = v_1$, then $\fire(\sigma, u_1) = \fire(\sigma, v_1)$. 
    \begin{itemize}
        \item Let $\sigma' = \fire(\sigma, u_1)$, then $\sigma^{(u)}$ is obtained by firing $a-1$ vertices $u_2, u_3, \cdots, u_a$ in order, and $\sigma^{(v)}$ is obtained by firing $b-1$ vertices $v_2, v_3, \cdots, v_b$. 
        \item By the induction hypothesis, the lemma is correct for $\max(a-1,b-1) \leq k-1$, which means $\sigma^{(u)} = \sigma^{(v)}$, and each vertex will be fired the same number of times in $u_2, u_3, \cdots, u_a$ and $v_2, v_3, \cdots, v_b$. 
        \item Since $u_1 = v_1$, the vertex $u_1$ will be fired once more in both of the firing plans, so the number of the firings on $u_1$ remains equal.
    \end{itemize}
    \item Otherwise, there must exists an $2 \leq i \leq a$ such that $u_i = v_1$. 
    \begin{itemize}
        \item This is because we have $\sigma_{v_1} \geq \degree(v_1)$ in the beginning. Since the only way to decrease the number of the chips on a vertex is to perform a fire operation, we must perform at least one firing operation on $v_1$ to obtain a terminal configuration. It implies that there exists at least one $1 \leq i \leq a$ such that $u_i = v_1$.
        \item Let's take the smallest $i$ such that $u_i = v_1$. The condition that $u_j \ne u_i$ for all $1 \leq j < i$ must be held. By \Cref{lemma:fire-commutative-ex}, $\sigma^{(u)}$ equals to the configuration obtained by firing the vertices $u_i, u_1, u_2, \cdots, u_{i-1}, u_{i+1}, u_{i+2}, \cdots, u_a$ in order.
        \item Since $u_i = v_1$, by applying the proof of the case $u_1 = v_1$, we have $\sigma^{(u)} = \sigma^{(v)}$ and each vertex $u$ were fired the same number of times in both plans.
    \end{itemize} 
\end{itemize}

In general, if the lemma is correct for $\max(a, b) \leq k - 1$, then it is also correct for $\max(a,b) = k$. By using mathematical induction, the lemma is true for all values of $k \in \N_{\geq 0}$.
\end{proof}

\begin{proof}[Proof of \Cref{lemma:fire-in-final}]
Suppose $\sigma'$ is obtained by firing vertex $v_1, v_2, \cdots, v_t \in \subtree(u)$. Let $\sigma^{(0)} = \sigma$ and $\sigma^{(i)} = \fire(\sigma^{(i-1)}, v_i)$ for all $1 \leq i \leq n$. Then $\sigma' = \sigma^{(t)}$. 

By \Cref{def:local-final}, we know $\final(\sigma^{(i)} + \sigma^{*}, u) = \final(\fire(\sigma^{(i-1)}, v_{i}) + \sigma^{*}, u)$. Since we can fire $v_i$ in the configuration $\sigma^{(i-1)}$, we must have $\sigma^{(i-1)}_{v_i} \geq \degree(v_i)$. By \Cref{lemma:add-then-fire}, $\fire(\sigma^{(i-1)}, v_i) + \sigma^* = \fire(\sigma^{(i-1)} + \sigma^*, v_i)$.

Since $v_i \in \subtree(u)$, firing any vertex in $\subtree(u)$ will not affect $\final(\sigma, u)$. So 

This gives us $\final(\sigma', u) = \final(\sigma_t, u) = \final(\sigma_0, u) = \final(\sigma, u)$. $\final(\fire(\sigma^{(i-1)} + \sigma^*, v_i),u) = \fire(\sigma^{(i-1)} + \sigma^*, u)$. This gives us $\final(\sigma^{(i)} + \sigma^*, u) = \final(\sigma^{(i-1)} + \sigma^*, u)$ for all $1 \leq i \leq t$. It implies $\final(\sigma' + \sigma^*, u) = \final(\sigma^{(t)} + \sigma^*, u) = \final(\sigma^{(0)} + \sigma^*, u) = \final(\sigma + \sigma^*, u)$.
\end{proof}

\begin{proof}[Proof of \Cref{lemma:final-sigma-add}]
By \Cref{def:local-final} $\final(\sigma, u)$ is obtained by firing several vertex $v \in \subtree(u)$. By applying \Cref{lemma:fire-in-final} $\final(\final(\sigma, u) + \final(\sigma', u), u) = \final(\sigma + \final(\sigma', u), u)$. Applying the lemma to $\final(\sigma', u)$ again we will get $\final(\sigma + \sigma(\sigma', u), u) = \final(\sigma + \sigma', u)$.
\end{proof}

\begin{proof}[Proof of \Cref{lemma:remaining-chips}]
\label{proof:lemma:remaining-chips}
We can prove the lemma by induction. The lemma is trivial for $k = 1$, as $\psi_u(0) = \sigma_u \geq \degree(u)$, we can fire vertex $u$ directly.

For all $k \geq 2$, and $\psi_u(k-1) \geq \degree(u)$. By \Cref{lemma:f-is-monotonically-decreasing}, $\psi_u(k-2) \geq \psi_u(k-1) \geq \degree(u)$, thus we can fire vertex $u$ at least $k-1$ times by inductive hypothesis.

Assume we have fired vertex $u$ exactly $k-1$ times, and we fired all full vertices in $\subtree(v_i)$ for all $v_i \in \son(u)$. By the inductive hypothesis, there are $\psi_u(k-1)$ chips on vertex $u$. Since $\psi_u(k-1) \geq \degree(u)$, we can perform one firing operation on vertex $u$, and the number of chips on vertex $u$ will become $\psi_u(k-1) - \degree(u)$.

Consider all $v_i \in \son(u)$. Before the $k$-th firing on vertex $u$, it receives $k-1$ chips from $u$, and gives $u$ back $\delta(v_i,k-1)$ chips. After the $k$-th operation on $u$, it will receive one more chip. For the initial configuration $\sigma$, adding $k$ more chips on vertex $v_i$ will let vertex $u$ receive $\delta(v_i, k)$ chips. So there will be $\left(\delta(v_i,k) - \delta(v_i,k-1)\right)$ more chips received from vertex $v_i$ after making $\sigma$ local terminal in $\subtree(v_i)$. Thus the number of chips on vertex $u$ will become $\psi_u(k-1) - \degree(u) + \sum_{v \in \son(u)} \left(\delta(v,k) - \delta(v,k-1)\right)$, which is exactly $\psi_u(k)$.
\end{proof}

\begin{proof}[Proof of \Cref{lemma:value-of-c}]
By \Cref{def:local-final}, a configuration that is local terminal in $\subtree(r)$ is equivalent to being a terminal configuration. So $\bc(r) = \bc^{\downarrow}(r)$.

For every $u \in V(G)$ such that $u \ne r$, consider the following way to find the terminal configuration of $\sigma$.
 
\begin{enumerate}
    \item For any vertex $v \in \subtree(u)$ such that $v$ is a full vertex, perform a firing operation on $v$. Repeat until there are no such $v \in \subtree(u)$ exists.
    \item For any vertex $w \in V(G)$ such that $w$ is a full vertex, perform a firing operation on $w$. Then check if there is any vertex $v \in \subtree(u)$ such that $v$ is a full vertex. If so, find such $v \in \subtree(u)$ repeatedly and fire the vertex $v$, until there is no such $v$ exists. Repeat this process until there is no full vertex in $\sigma$.
\end{enumerate}

The procedure will find a terminal configuration if it exists, since there will not be any full vertex after the procedure.

In the first stage, it's equivalent to performing a local finalize operation described in \Cref{def:local-final}. By definition, the vertex $u$ will be fired exactly $\bc^{\downarrow}(u)$ times in this stage.

Now we consider the second stage. Since the given graph is a tree, there is only one vertex, which is $\parent(u)$ exactly, that serves as a neighbor of vertex $u$ but does not belong to $\subtree(u)$. So in the second stage, the only way that vertex $u$ receives an additional chip is by firing vertex $\parent(u)$.

Since vertex $\parent(u)$ has never been fired in the first stage, all the firing operations on vertex $\parent(u)$ will be happening in the second stage. So, the vertex $u$ will receive exactly $\bc(\parent(u))$ chips.

Note that if there are $\bc(\parent(u))$ additional chips placed on vertex $u$, then $\parent(u)$ will receive $\delta(u, \bc(\parent(u)))$ more chips after the configuration becomes a local terminal in $\subtree(u)$. And it is equivalent to vertex $u$ being fired $\delta(u, \bc(\parent(u)))$ times in this stage.

Adding the two stages together, vertex $u$ is fired a total of $\bc(u) = \bc^{\downarrow}(u) + \delta(u, \bc(parent(u)))$ times.

\end{proof}





\begin{proof}[Proof of \Cref{theorem:inequality}]
For any vertex $u \in G$, the number of chips after all the firing operations finish should be $\sum_{v \in N(u)} \bc(v)-\bc(u) \cdot \degree(u)+\sigma_u$. Such a number should be less than $\degree(u)$. Otherwise, it is possible to perform one more firing operation on the vertex $u$. This is where the inequalities in \Cref{linearsystem} come from.

We will prove the theorem by showing that all feasible solutions form a meet-semilattice $L=(F,\land)$, which elements $\ba\in F$ are vectors representing feasible solutions and the meet operation is the pointwise min operation, denoted as $\land$. We also define the partial order as follows: we say $\bx\leq \by$ if and only if $\bx(u)\leq \by(u)$ for every $u\in V(G)$. This satisfies the requirement of the definition of the meet-semilattice.

Considering two arbitrary feasible solutions $\ba,\bb$, without losing the generality, we assume that $\ba(u)\geq \bb(u)$ for an arbitrary vertex $u \in V(G)$. Since we have 

\begin{align*}
    \left(\sum_{v\in N(u)}\min(\ba(v),\bb(v))\right) &-\min(\ba(u),\bb(u))\cdot \degree(u) +\sigma_u
 \\ &\leq \left(\sum_{v\in N(u)}\bb(v)\right)-\bb(u)\cdot \degree(u) +\sigma_u
 \\ &<\degree(u) 
\end{align*}, $\ba \land \bb$ is also a feasible solution. This gives us that $\land$ is a well-defined operation. Assume the solution space is non-empty. For an arbitrary feasible solution $\bx$, we can construct a new meet-semilattice $L'=(\{\bx\land \bd\},\land),\bd\in F$ which is finite. There is a feasible solution $\bp$ with the minimum partial order in $L'$. We can easily verify that $\bp$ is also equal to the minimum element in $L$.

Now we will prove that $\bp$ is equal to the firing vector $\bc$. Let's denote $\vq$ as a vector containing all zeros. Whenever a firing operation on vertex $u$ happens, we will increase $\vq_u$ by $1$. After all firings happen following any firing order and the instance terminates, $\vq$ will always be equal to the firing number vector $\bc$ by \cite{bjorner1991chip}. 

We claim that if there exist some vertices $u$ such that $\vq_u < \bp_u$, there exists at least one vertex can be fired among them. Otherwise, the instance terminates with $\bp=\vq$. By induction, we assume the first $k-1$ operations are firing on some vertices $u$ where $\vq_u < \bp_u$. For the $k$-th firing, there are two cases to consider:
\begin{itemize}
    \item The instance is not terminated. There exist some vertices $u$ satisfying $\vq_u<\bp_u$, but none of them can be fired.
    \item The instance is terminated. There exists $u$ satisfying $\vq_u<\bp_u$.
\end{itemize}

For the first case, we can only fire from any vertex $u'$ such that $\vq_{u'} = \bp_{u'}$. Since $\bp$ is a feasible solution, we have
\[
\sum_{v \in N(u')}\bp(v)-\bp(u') \cdot \degree(u') + \sigma_{u'} < \degree(u')
\].

Since $\sum_{v \in N(u')}\vq(v) \leq \sum_{v \in N(u')}\bp(v)$ and $\vq(u')=\bp(u')$, we have 

\[
\sum_{v \in N(u')}\vq(v)-\vq(u') \cdot \degree(u') + \sigma_{u'} < \degree(u')
\], indicating that no firing operation can be performed on the vertex $u'$, which is a contradiction.

For the second case, if the instance is terminated, the current $\vq$ is the firing vector and thus a feasible solution of \Cref{linearsystem}. Since $\vq$ will have a smaller partial order than $\bp$, it contradicts the assumption that $\bp$ has the smallest partial order.

Therefore, none of these two cases is possible. Therefore, when the instance terminates, $\bp=\vq=\bc$.

\end{proof}

\begin{proof}[Proof of \Cref{lemma:generalsink}]
At first, we select a non-sink vertex $u$, and we binary search its firing number $\bc(u)$ in the range of $[0,L_2]$.
If we cannot find a solution that satisfies for any $v \in V(G)$, $\bc(v) \leq L_2$, then there is no solution to this bounded prediction problem. Therefore, we return with a $\ovf$.
Assuming we need to determine if the current binary search value $mid$ is no less than $\bc(u)$. By \Cref{corollary:monofiring}, we only need to check if there is a feasible solution with $\bf(u)=mid$.

We apply $mid$ times of firings on vertex $u$, replace $\bf(u)$ into $mid$. Then we will turn $u$ into a sink vertex.
 In the view of the inequality system, it is equivalent to substituting $mid$ into all the terms of $\bf(u)$ within the system. Then eliminate the inequality on the vertex $u$. In this way, we reduce to a new sandpile instance $S'(G',\sigma',M')$ where $G'=G \setminus u$, $\sigma'_v=\sigma_v+[v\in N(u)]\cdot mid$, and $M'=M \cup \{u\}$.
If we can compute the terminal configuration and its corresponding firing number vector $\bd$ of this instance, by \Cref{corollary:inequality}, $\bd$ is equivalent to the feasible solution of the system $S'$ with the minimum partial order.

There are two cases for the result we obtained in solving the bounded sandpile prediction of $S'$.

\begin{itemize}
    \item If the result we obtained is $\ovf$, then either the terminal configuration of the original problem is $\ovf$, or we are setting the threshold $mid$ too large.
    \begin{itemize}
        \item In the first case, any value of $mid$ will result in $\ovf$. In this case, we do not care about the value of $mid$ we obtained after the binary search procedure. So we can just assume the threshold $\mid$ is too large and does not affect the results of the algorithm.
        \item Otherwise, due to the monotonicity in \Cref{corollary:monofiring}, we are setting the threshold $mid$ larger than the correct one. 
    \end{itemize} 
    In all, we continue with the binary search process considering the lower potential value.
    \item Otherwise, assume the firing vector corresponding to the result we obtained is $\bd$. Note that we haven't considered the inequality with the vertex $u$, substituting $\bf(v) = \bd_v$ for all $v \in V(G)$ and $v \ne u$, might not be a solution for the inequality \Cref{formula:inequality-on-vertex-u}:
    \begin{align}
    \left(\sum_{v\in N(u) \setminus M}\bf(v)\right)-\bf(u)\cdot \degree(u)+\sigma_u<\degree(u), \forall v \in V(G) \setminus M
    \label{formula:inequality-on-vertex-u}
    \end{align}
        Let's check if the inequality holds for $\bf(u) = mid$ and $\bf(v) = \bd_v$ for all $v \ne u$. There are two different cases.
        \begin{itemize}
            \item If the inequality holds, we found a feasible solution. Since we only care about the solution with the minimal partial order, we continue to search the solutions with smaller partial orders. So we will return a value of $mid$ which is no greater than the correct value of $\bc(u)$.
            \item Otherwise, we can see that we are setting the threshold $mid$ too small. There are two different cases:
            \begin{itemize}
                \item If the terminal configuration is $\ovf$, then similarly any value of $mid$ will result in $\ovf$. Thus we can choose the new value of $mid$ arbitrarily.
                \item Otherwise, the feasible value of $\bc(u)$ must be larger than $mid$. Note that adding all the value of $\bf(v)$ by one might decrease the left-hand side of the inequality \Cref{formula:linearsystemwithsinks}. So if the feasible value of $\bc(u)$ is smaller than $mid$, then let's take the feasible answer $\bc$ and consider a solution $\bf'(v) = \bc(v)+mid-\bc(u),v\in V(G)$. Thus add a value $mid - \bc(u) > 0$ to the whole $\bc$, obtaining another feasible solution. By \Cref{corollary:inequality}, the solution $\bd$ we obtained by assuming $\bf(u) = mid$ will be the solution with the minimum partial order in all the feasible solutions of the inequality system of $S'$. And since the $\bf'(u)$ is a feasible solution in \Cref{formula:linearsystemwithsinks}, we have
                \begin{align}
                \sum_{v \in N(u)\setminus M} \bd_v \leq \sum_{v \in N(u) \setminus M} \bf'(v)
                \label{formula:sum-of-bd-is-no-more-than-sum-of-bf-prime}
                \end{align}
                Since $\bf'$ is a solution for \Cref{formula:inequality-on-vertex-u}, we have
                \begin{align}
                \sum_{v\in N(u) \setminus M} \bf'(v) - \bf'(u) \cdot \degree(u) + \sigma_u < \degree(v)
                \end{align}
                By \Cref{formula:sum-of-bd-is-no-more-than-sum-of-bf-prime} and $\bf(u) = mid$ we have 
                
                \begin{align}
                \sum_{v \in N(u) \setminus M} \bf(v) - \bf(u) \cdot \degree(u) + \sigma_u  &= \sum_{v \in N(u) \setminus M} \bd_v - mid \cdot \degree(u) + \sigma_u \\ &\leq \sum_{v \in N(u) \setminus M} \bf'(v) - \bf'(u) \cdot \degree(u) + \sigma_u
                \end{align}
                Thus $\bf$ is also a feasible solution for \Cref{formula:inequality-on-vertex-u}, which is a contradiction. Thus the feasible value of $\bf(u)$ must be larger than $mid$.
            \end{itemize}
        \end{itemize}
        We continue with the binary search process considering the higher potential value.
\end{itemize}

Now we consider the bounded parameter $L_1'$ and $L_2'$ of $S'$. Since we fire vertex $u$ for $mid$ times in the beginning, thus we have $L_1'=L_1+\degree(u)\cdot mid \leq L_1+\degree(u)\cdot L_2$ and $L_2'=L_2$. In all, we analyze the logic of proceeding the binary search and reduce the problem to a new one after regarding $u$ as a sink vertex. Thus, we proved the theorem. 



\end{proof}

\begin{proof}[Proof of \Cref{general:mergeandsplit}]

$\mergeupward$ is a procedure used to merge the information in a subtree. Let's consider the case when we perform $\mergeupward$ on a vertex $u$ that is not the root. According to our decomposition described in \Cref{sec:treedecom}, there will be at most one subtree containing a sink vertex, and we choose the merge order $\mathcal{I}$ such that the first element is the subtree with the sink vertex. 
Therefore, we never merge two subtrees such that both of them contain a sink vertex. Consequently, we only need to handle the following two cases.

\begin{enumerate}
    \item The first case is when we need to merge two subtrees that do not contain any sink.
    \begin{itemize}
        \item In this case, we can perform the same small-to-large trick described in \Cref{sec:mergeds}. 
        \item The total time complexity of the algorithm remains $O(n \log n)$ and it uses up to $O(n)$ memory.
    \end{itemize}
    \item The second case is more complex, which involves merging a subtree containing a sink vertex with another subtree.
    \begin{itemize}
        \item Let's denote the subtree with a sink vertex as $\subtree(v_{s})$ and the other subtree as $\subtree(v_0)$.
        \item When merging a subtree with a sink vertex into another subtree, we can ignore the small-to-large technique and simply add all the nodes from $D_{v_0}$ to $D_{v_s}$.
        \item This is because whenever a node is inserted to a subtree with a sink vertex, it will remains to there and won't be moved to another subtree. So each node will be inserted at most once during this procedure.
    \end{itemize}
\end{enumerate}
Therefore, $\mergeupward$ and $\splitds$ can be modified to maintain key pairs with differences while preserving the time and space complexity.

So we can apply the modified merge procedure above for any vertex $v \in V(G)$ other than the root. However, for the root vertex, we cannot directly apply the merge procedure because we need to merge two subtrees that contain a sink vertex simultaneously. Therefore, we will skip the $\mergeupward$ and $\splitds$ operations on the root vertex. To do this, we need to manually calculate the value of $c^{\downarrow}(r)$ for the root vertex $r$. This can be finished by performing a binary search on $c^{\downarrow}(r)$ by \Cref{lemma:c-down}. Since the the firing number is bounded by $O(|M|^4 \cdot (\lvert\lvert \sigma \rvert\rvert_1 + n)^4)$ in \Cref{lemma:upper-bound-of-firing-number-with-sinks}, the binary search procedure will perform $O(\log \left(|M|^4 \cdot (\lvert\lvert \sigma \rvert\rvert_1 + n)^4\right)) = O(\log \lvert\lvert \sigma \rvert\rvert_1 + \log n)$ turns. In each turn, we have to calculate the value of $\sigma_r - k \cdot \degree(r) + \sum_{v \in \son(u)} \delta(v, k)$. This can be done by using our maintained splay tree in $O(\log n)$ time. So the procedure on the root vertex will take $O(\log \lvert\lvert \sigma \rvert\rvert_1 \cdot \log n + \log^2 n)$ time.

As a result, we no longer delete nodes from $\bq_u$, and the $\updatedsupward$ and $\revertds$ operations at the root will be complete cancellations, so we can simply ignore them.

For the $\splitds$ operation, we follow a similar implementation as described in Section \Cref{sec:splitds}. We keep track of the timestamp when each node first joins $D_u$ during merging, and we reallocate them accordingly when executing $\splitds(u,v)$. As for the additional nodes generated from insertions in the special subtree with a sink vertex, since we choose the special order where the special subtree appears first and set the $\timestamp$ value to $0$, we can ensure that the nodes will be assigned to the correct subtree with the sink vertex.

\end{proof}

\begin{proof}[Proof of \Cref{general:updateandrevert}]

For the $\updatedsupward$ operation, it relies on the $\inctime$ and $\insert$ operations. We have already shown in Lemma \Cref{lemma:sinks:inctime} that these operations can be modified to adapt to the new way of maintaining key pairs without affecting the time and space complexity. Therefore, the $\updatedsupward$ operation can still be performed correctly.

For the $\revertds$ operation, it additionally relies on the $\delete$ operation. We have also shown in Lemma \Cref{lemma:sinks:delete} that the $\delete$ operation can be modified to adapt to the new key pair maintenance method while maintaining the time and space complexity.

Hence, both the $\updatedsupward$ and $\revertds$ operations can be modified to accommodate the key pair maintenance method without compromising the time and space complexity.

\end{proof}

\begin{proof}[Proof of \Cref{general:computecanddeltasum}]

For the $\cnt$ operation, we handle two cases separately: $u = r$ (the root vertex) and $u \neq r$ (other vertices).

For $u \neq r$, we can inherit the $\cnt$ procedure as it is. However, we need to make some modifications to account for the omitted key pairs and the difference in the maintained key pairs.

For $u \neq r$, we can inherit the $\cnt$ procedure as it is. The line \Cref{computeC:if} checks whether we can move the pointer $now$ to the next nodes in the splay. Thus in the succinct splay with omitted nodes, the later formula should be modified to $count+(\moment_x-1-now)+t_x$, since we have to consider the cost of skipping the omitted $t_x$ key pairs. Accordingly, the line \Cref{computeC:accumulate} will be changed to $count\leftarrow count+t_x+1+(\moment_x-now)$.

At line \Cref{computeC:cornercase}, we now have two cases to handle:

\begin{itemize}
    \item If $now$ is not equal to the maximum value among all $\moment_x$ values, we can observe that the final value will certainly be less than the next $\moment_y$ value, where $y$ represents the earliest node for which $now<\moment_y$.Our next task is to determine the increment $p$ for the minimum non-negative solution of the inequality $count+p+\max\left(\left\lfloor \frac{now+p-\moment_y+(t_y+1)d_y}{d_y}\right\rfloor,0\right)>\sigma'_u-\degree(u)$. Therefore, we should set \[p=\min\left(\sigma'_u-count-(\degree(u)-1),\left\lceil\frac{\sigma'_u+\moment_y-now-(\degree(u)-1+count+t_y+1)d_y}{d_y+1}\right\rceil\right)\]. As we increase $now$ by $p$, we need to include some omitted nodes in $\bq_u$. The number of nodes to be transferred should be \[num=\max\left(\left\lfloor \frac{now+p-\moment_y+(t_y+1)d_y}{d_y}\right\rfloor,0\right)\]. If $num$ is $0$, then no action is required. Otherwise, we create a node $z$ with $\moment_z=\moment_y-d_y\cdot(t_y-num+1)$, $t_z=num-1$, $d_z=d_y$, and $\timestamp_z=0$. We place the node $z$ into $\bq_u$ and update $t_y\leftarrow t_y-num$.
    \item If there are no nodes $y$ satisfying $\moment_y>now$, we need to consider the potential infinite key pairs in the tail. If $t=0$, we set $p$ to $\sigma'_u-(\degree(u)-1)-count$. Alternatively, if $t=\infty$, we calculate $p$ as $\left\lceil\frac{d_y(\sigma'_u-(\degree(u)-1)-count)}{d_y+1}\right\rceil$. In either case, we create a new node $z$ with $\moment_z=now+d\cdot \lfloor\frac{p}{d}\rfloor$, $t_z=p-1$, $d_z=d$, and $\timestamp_z=0$. Additionally, we create a new node $w$ with $\moment_w=now+d\cdot (\lfloor\frac{p}{d}\rfloor+1)$, $t_w=d_w=0$, and $\timestamp_w=0$. We then place $z$ into $\bq_u$ and let $w$ be the only remaining node in $D_u$.

\end{itemize}

We can establish the consistency between the original $\cnt$ and the modified splay tree maintained by differences by comparing the complete key pairs represented in the modified splay tree.

For the case when $u=r$, we have $\degree(r)$ subtrees and their corresponding values of $now$. Based on the decomposition and selection of $r$ (as mentioned above), we can conclude that there are at most $3$ subtrees with sinks. We merge all the remaining regular $D_u$ trees using the same small to large technique, resulting in at most $4$ isolated splay trees. We can obtain the final $\bc^{\down}(r)$ value by applying a binary search on the splay trees. Then we compute the $\psi_r(k)$ value by querying the $\delta$ values separately in the $4$ splays. This process takes a total of $O(\log^2 n)$ time since we can limit the binary search range to $n^3+ n\cdot \sum_{u\in V(G)} \sigma_u$, and each splay operation takes $\log n$ time.

\end{proof}

\begin{proof}[Proof of \Cref{general:deltaquery}]
Similarly, by \Cref{lemma:deltaequal} we only need to determine the number of key pairs that are not greater than a given $k$.

By \Cref{lemma:findpre}, we find the node $x$ with the largest $\moment_x\leq k$ and find $y=\suc(x)$ or the node with $\rank_{D_u}=1$ if $x=\nil$. First, we compute the number of key pairs $(u,k)$ such that $k \leq \moment_x$, denoted as $pc$. This value can be computed by summing $t_y+1$ for each node $y$ with $\rank_{D_u}(y)\leq \rank_{D_u}(x)$, which is a standard operation in a splay tree. Next, let's consider the additional omitted key pairs after $x$ and before $k$. There are two cases to consider:


\begin{itemize}
\item If $y=\nil$, then $x$ is the node with the largest $\moment_x$ among all the nodes. If $t=0$, there are no additional key pairs. Otherwise, there are $\lceil \frac{k-\moment_x}{d}\rceil$ additional key pairs.
\item If $y\neq \nil$, then there are $\max(\lceil\frac{k-\moment_y+(t_y+1)d_y}{d_y}\rceil,0)$ additional key pairs.
\end{itemize}
\end{proof}

\end{document}